\documentclass[a4paper,11pt]{article}
\usepackage{a4wide}
\usepackage{amsmath,makeidx,amssymb,amscd,amsfonts, amsthm,caption,color,epsfig,
fancyhdr,graphicx,latexsym,mathrsfs}
\usepackage{amscd,color,epsfig,
fancyhdr,graphicx,latexsym,mathrsfs,multicol,setspace,wrapfig}
\usepackage{psfrag}
\usepackage{multirow}
\usepackage{natbib}
\setcitestyle{notesep={; },round,aysep={},yysep={;}}

\def\argmin{\text{argmin}}

\newcommand\R{\mathbb{R}}

\newcommand\N{\mathbb{N}}

\def\argmin{\text{argmin}}

\newtheorem{lem}{Lemma}

\newtheorem{pr}{Proposition}

\newtheorem{df}{Definition}
\newtheorem{rem}{Remark}

\newtheorem{ass}{Assumption}

\newcommand{\tu}{{\tilde{u}}}
\newcommand{\tv}{{\tilde{v}}}

\title{\bf A Splitting Strategy for the Calibration of Jump-Diffusion Models}

\author{Vinicius Albani\thanks{Federal University of Santa Catarina,  
88.040-900 Florianopolis, Brazil, \href{mailto:v.albani@ufsc.br}{\tt v.albani@ufsc.br}} 
\,and Jorge P. Zubelli\thanks{IMPA, 
Rio de Janeiro, RJ 22460-320, Brazil, \href{mailto:zubelli@impa.br}{\tt zubelli@impa.br}}}
\date{\today}
\begin{document}

\maketitle

\begin{abstract}
We present a detailed analysis and implementation of a splitting strategy to identify simultaneously
the local-volatility surface and the jump-size distribution from quoted European prices. 
The underlying model consists of a jump-diffusion driven asset with time and price dependent volatility.
Our approach uses a forward Dupire-type partial-integro-differential equations for the option prices 
to produce a parameter-to-solution map. The ill-posed inverse problem for such map is then solved by 
means of a Tikhonov-type convex regularization. The proofs of convergence and stability of the 
algorithm are provided together with numerical examples that substantiate the robustness of the method both for synthetic and real data. 
\end{abstract}

\noindent {\bf keywords:} Jump-Diffusion Simulation, Partial Integro-Differential Equations, Finite Difference Schemes, Inverse Problems, Tikhonov-type regularization.

\section{Introduction}\label{sec:introduction} 
Model selection and calibration is still one of the crucial problems in derivative trading and hedging.
From a mathematical view-point it should be treated as an ill-posed inverse problem by suitable regularization as in the work of \cite{ern} and \cite{schervar}. The subject is deeply connected to nonparametric statistics as described by \cite{somersalo}. The problem of model selection and calibration has thus attracted the attention of number of authors as can be seen in the book of \cite{ConTan2003} and references therein. 

Amongst the most successful nonparametric approaches, 
the local volatility model of  \cite{dupire} has become one of the 
market's standards. It consists in assuming that the underlying price $S_t$ satisfies a stochastic dynamics of the form 
$$
dS_t = r S_t dt + \sigma(t,S_t) dW_t \mbox{, }
$$
where $W_t$ is the Wiener process under the risk-neutral measure and $\sigma$ is the so-called local volatility.  
Besides its intrinsic elegance and simplicity, Dupire's model success is due to at least two factors:
Firstly, the existence of a forward partial differential equation (PDE) satisfied by the price of
call (or put) options when considered as functions of the strike price $K$ and the time to expiration $\tau$. Secondly, to the importance of having a backwards pricing PDE to compute other (perhaps exotic) derivatives. 

Yet, one of the main shortcomings of local volatility models is the fact that such models are still diffusive ones. Thus, well-known stylized facts such as fat tails and jumps in the log-returns become awkward to fit and justify (\cite{ConTan2003}).

The present article is concerned with the calibration of jump-diffusion models with local volatility. We make use of a fairly recent contribution to the literature, namely the existence of a forward equation of Dupire's type for such models present in \cite{BenCon2015}. The availability of a forward equation, allows us, for each fixed time and underlying price,  to look at the option prices as a function of the time to expiration and the strike price. Furthermore, by considering collected data from past underlying and derivative prices, we can enrich our observed data and strive for better calibration prices. 

Efforts to calibrate jump-diffusion models from option prices have been undertaken by a number of authors either from a parametric and a nonparametric perspective. See~\cite{AndAnd2000}, \cite{ConTan2003} and \cite{volguide}.
In this work, differently from previous efforts, we focus on using Dupire's forward equation, as generalized in the work of \cite{BenCon2015} and propose a splitting calibration methodology to recover simultaneously the local volatility surface and the jump-size distribution. For a fixed dataset of European vanilla option prices we calibrate, for example, the volatility surface for some fixed jump-size distribution. Then, we find a new reconstruction of the jump-size distribution for the volatility surface previously calibrated. We repeat these steps until a stopping criteria for convergence is satisfied. The resulting pair of functional parameters is indeed a stable approximation of the true local volatility surface and the jump-size distribution, whenever they exist. It is important to mention that, the dataset used to identify this pair of functional parameters is the same one used in Dupire's local volatility calibration problem. No additional data is required as it would be necessary if we wanted to calibrate both parameters at the same time using standard regularization techniques (\cite{ern}). 
The resulting methodology is amenable to regularization techniques as those studied in \cite{AlbZub2014} and \cite{AlbAscZub2016}. In particular, different a priori distributions could be used. As a byproduct, we prove convergence estimates for the calibration of the jump-diffusion models as the data noise decreases. We also obtain stable and robust calibration algorithms which perform well either under real or synthetic data. 

The plan for this work is the following: In Section~\ref{sec:preliminaries} we set the notation and review some basic facts, including the fundamental forward equation for jump-diffusion processes. 
In Section~\ref{sec:par2sol}
we discuss the main functional-analytic properties of the parameter to solution map. 
Section~\ref{sec:splitting} 
is concerned with the splitting strategy and the regularization of  inverse problems. 
In particular, we review the tangential cone condition and prove its validity under certain assumptions in our context. This condition, in turn, ensures the convergence of Landweber type methods. 
The results in this section are not specific to the jump-diffusion model under consideration. Indeed, they apply to more general inverse problems, although, to the best of our knowledge we have not seen presented in this form.
In Section~\ref{sec:calibration} we compute the gradient of the nonlinear parameter-to-solution map, which is crucial for the iterative
methods.
Section~\ref{sec:numerics} is concerned with the numerical methods for the solution of the calibration problem and its validation. Differently from \cite{ConVol2005a,ConVol2005b}, we consider also the case where the 
jumps may be infinite. 
Section~\ref{sec:examples} presents a number of numerical examples that validate the theoretical results and display the effectiveness of our methodology. 
We close in Section~\ref{sec:conclusion} with some final remarks and suggestions for further work. 

\section{Preliminaries}\label{sec:preliminaries} 
Let us consider the probability space $(\Omega,\mathcal{G},\mathbb{P})$ with a filtration $\{\mathscr{F}_t\}_{t\geq 0}$. Denote by $S_t$ the price at time $t \geq 0$ of our underlying asset and assume that it satisfies
\begin{multline}
S_t = S_0 + \int_0^t rS_{t^\prime-}dt^\prime + \int_0^t\sigma(t^\prime,S_{t^\prime-})S_{t^\prime}dW_{t^\prime} + \\
\int_0^t \int_{\R}S_{t^\prime-}(\text{e}^y-1)\tilde N(dt^\prime dy),~ ~ ~0\leq t\leq T,
\label{ito1}
\end{multline}
where $W$ is a Brownian motion, and $\tilde N$ is the compensated version of the Poisson probability measure on $[0,T]\times \R$, denoted by $N$, with compensator $\nu(dy)dt$. See \citet{ConTan2003}.

Assume also that $\sigma$ is positive and bounded from below and above by positive constants, and the compensator $\nu$ satisfies
\begin{equation}
 \int_{|y|>1} \text{e}^{2y}\nu(dy) < \infty.
\end{equation}

Since $\sigma$ is uniformly bounded and nonnegative, then, by setting $t=0$ and denoting $\tau$ the time to maturity and $K$ the strike price, the seminal work of \citet{BenCon2015} shows that the price of an European call option on the asset in \eqref{ito1}, defined by 
$$
C(\tau,K) = \text{e}^{-r\tau}\mathbb{E}[\max\{0,S_\tau-K\}|\mathscr F_{0}],
$$
is the unique solution in the sense of distributions of the partial integro-differential equation (PIDE):
\begin{multline}
 C_\tau(\tau,K) - \frac{1}{2}K^2\sigma(\tau,K)^2 C_{KK}(\tau,K) + rKC_K(\tau,K)=\\
 \int_{\R}\nu(dz)\text{e}^z\left(C(\tau,K\text{e}^{-z}) - C(\tau,K) - (\text{e}^{-z}-1)KC_K(\tau,K)\right),
 \label{pide1}
\end{multline}
with $\tau\geq 0$, $K>0$, and the initial condition
\begin{equation}
 C(0,K) = \max\{0,S_0-K\}, ~K>0.
 \label{pide1b}
\end{equation}

Since the diffusion coefficient in Equation \eqref{pide1} is unbounded and goes to zero as $K\rightarrow 0$, let us perform the change of variable $y = \log(K/S_0)$ and define
$$
a(\tau,y) = \frac{1}{2}\sigma(\tau,S_0\text{e}^y)^2 \qquad \text{and} \qquad u(\tau,y) = C(\tau,S_0\text{e}^y)/S_0.
$$
So, denoting $D = [0,T]\times\R$, the PIDE problem \eqref{pide1}-\eqref{pide1b} becomes
\begin{multline}
 u_\tau(\tau,y) - a(\tau,y) \left(u_{yy}(\tau,y) - u_{y}(\tau,y)\right)  + ru_y(\tau,y)=\\
 \int_{\R}\nu(dz)\text{e}^z\left(u(\tau,y-z) - u(\tau,y) - (\text{e}^{-z}-1)u_y(\tau,y)\right),
 \label{eq:pide2}
\end{multline}
with $(\tau,y)\in D$, and the initial condition
\begin{equation}
 u(0,y) = \max\{0,1-\text{e}^y\}, ~y\in\R.
 \label{eq:pide2b}
\end{equation}

Instead of using $\nu$ directly in the PIDE problem~\eqref{eq:pide2}-\eqref{eq:pide2b}, we consider, as in \citet{KinMay2011}, the double-exponential tail of $\nu$
\begin{equation}
\varphi(y) = \varphi(\nu;y) = \left\{
\begin{array}{ll}
\int_{-\infty}^y(\text{e}^y-\text{e}^x)\nu(dx), & y < 0\\ 
\int_{y}^\infty(\text{e}^x-\text{e}^y)\nu(dx), & y > 0,
\end{array}
\right.
\label{eq:tail}
\end{equation}
and the convolution operator
$$
I_\varphi f (y) := \varphi * f (y) = \int_{\R}\varphi(y-x)f(x)dx.
$$
Applying Lemma~2.6 in \citet{BenCon2015} to the integral part of the PIDE \eqref{eq:pide2},
\begin{multline}
\int_{\R}\nu(dz)\text{e}^z\left(u(\tau,y-z) - u(\tau,y) - (\text{e}^{-z}-1)u_y(\tau,y)\right)\\ = \int_{\R}\varphi(y-z)(u_{yy}(\tau,z) - u_y(\tau,z))dz.
\label{eq:integral_part}
\end{multline}
In what follows, we replace the integral part of the PIDE \eqref{eq:pide2} by the right-hand side of \eqref{eq:integral_part}.

\begin{rem}
 Define $g(\tau,y) := \max\{0,1-\text{e}^y\}$, so, by the definition of $u$, it follows that $u(\tau,y) = \tv(\tau,y) + g(\tau,y)$, where $\tv$ is the solution of the PIDE: 
\begin{equation}
 \tv_\tau = a\left(\tv_{yy}-\tv_{y}\right) - r\tv_{y} + I_\varphi \left(b(\tv_{yy}-\tv_{y})\right) + G
 \label{eq:pide3}
 \end{equation}
 with homogeneous boundary and initial conditions, where
 $$
 G = a\left(g_{yy}-g_{y}\right) - rg_{y} + I_\varphi \left(b(g_{yy}-g_{y})\right),
 $$
with $g_{yy}$ and $g_{y}$ weak derivatives of $g$.  
\end{rem} 
By assuming that $a \in C_B(D)$ with $a_y \in L^\infty\left([0,T],L^2(\R)\right)$ and $a(\tau,y) \geq c > 0$, for every $(\tau,y)\in D$, and $\nu$ is a L\'evy measure satisfying $\displaystyle\int_{x\geq 1}x\text{e}^x\nu(dx) < \infty$, Theorem~3.9 in \citet{KinMay2011} states the existence and uniqueness of $\tv$. The proof that $u$ is a weak solution of the PIDE problem \eqref{eq:pide2}-\eqref{eq:pide2b} is an easy adaptation of the proof of Theorem~3.9 in \citet{KinMay2011}. To see that, just replace the test functions in $H^1(\R)$, by test functions with compact support in $C_0^\infty(D)$, as in \citet{BenCon2015}, and replace $\tv$ by $u-g$. 

Uniqueness of solution can also be proved by analytical methods as in \citet{BarImb2008,GarMen2002} or by probabilistic arguments as in Theorem~2.8 in \citet{BenCon2015}. An alternative proof is to consider the difference between two different solutions of the PIDE problem \eqref{eq:pide2}-\eqref{eq:pide2b}. The resulting function is the solution of the PIDE \eqref{eq:pide3} with $G \equiv 0$. By Theorem~3.7 in \citet{KinMay2011}, the norm of the solution of the PIDE \eqref{eq:pide3} is dominated by the norm of $G$, which is zero. So, the difference is also zero and uniqueness holds. Since $v$ and $g$ are continuous in $D$, it follows that $u$ is also a continuous function.

In 
Section~\ref{sec:par2sol}
we give an alternative proof for the existence and uniqueness of a solution of the PIDE problem~\eqref{eq:pide2}-\eqref{eq:pide2b} based on the classical theory of parabolic partial differential equations. See \cite{lady}.

\section{The Parameter to Solution Map and its Properties} \label{sec:par2sol}
The goal of the present section is to show the well-posedness of the PIDE problem~\eqref{eq:pide2}-\eqref{eq:pide2b} and some regularity properties of the parameter-to-solution map.

We make the following additional assumption:
\begin{ass}
  The restrictions of the double-exponential tail $\varphi$ to the sets $(-\infty,0)$ and $(0,+\infty)$ are in the Sobolev spaces $W^{2,1}(-\infty,0)$ and $W^{2,1}(0,+\infty)$, respectively.
\label{ass:4}
\end{ass}
$W^{2,1}(-\infty,0)$ (respectively $W^{2,1}(0,+\infty)$) is the Sobolev space of $L^1(-\infty,0)$ (respectively $L^1(0,+\infty)$) functions such that its first and second weak derivatives are in $L^1(-\infty,0)$ (respectively $L^1(0,+\infty)$).

The above assumption holds, for example, if we assume that the measure $\nu$ is such that the functions $x \in (-\infty,0) \mapsto \nu((-\infty,x])$ and $x \in (0,+\infty) \mapsto \nu([x,+\infty))$ are continuous. 

We recall, that the set of non-negative non-increasing functions has a nonempty interior in $W^{2,1}(0,+\infty)$ as well as the set of non-negative non-decreasing functions in $W^{2,1}(-\infty,0)$. This is of particular importance since we need to show that the direct operator has a Frech\'et derivative.

In order to define the domain of the direct operator in the Banach space 
$$X = H^{1+\varepsilon}(D)\times W^{2,1}(-\infty,0) \times W^{2,1}(0,+\infty),$$
let $0<\underline{a}\leq\overline{a}<\infty$ be fixed constants and $a_0:D\rightarrow (\underline{a},\overline{a})$ be a fixed continuous function such that its weak derivatives with respect to $\tau$ and $y$ are in $L^2(D)$.  
\begin{multline*}
 \mathcal{D}(F) = \left\{ (\tilde a,\varphi_-,\varphi_+)\in X ~:~\mbox{let}~ a = \tilde a + a_0,~\mbox{be s.t.}~  \underline{a}\leq a \leq\overline{a}, \right.\\
 \left. ~\mbox{let}~\varphi ~\mbox{be s.t.},~\varphi = \varphi_- ~\mbox{in}~ (-\infty,0) ~\mbox{and}~ \varphi = \varphi_+ ~\mbox{in}~(0,+\infty)
\right\}
\end{multline*}
It is easy to see that $\varphi \in L^1(\R)$.

For simplicity, in what follows we shall write $(a,\varphi) \in \mathcal{D}(F)$, meaning that $a$ and $\varphi$ are given as in the definition of $\mathcal{D}(F)$.
\begin{pr}
 Let $(a,\varphi)$ be in $\mathcal{D}(F)$, in addition, assume that $\|\varphi\|_{L^ 1(\R)} < C^{-1}$, where the constant $C$ depends on $\underline{a}$, $\overline{a}$ and $r$. Then, there exists a unique solution of the PIDE problem \eqref{eq:pide2}-\eqref{eq:pide2b} in $W^{1,2}_{2,loc}(D)$.
 \label{prop:existence}
\end{pr}
\begin{proof}
 The existence and uniqueness proof follows by a fixed point argument. Given $(a,\varphi)$ in $\mathcal{D}(F)$ and $f\in L^2(D)$, define the operator $G$ 
 that associates each $v\in W^{1,2}_2(D)$ to $w\in W^{1,2}_2(D)$, solution of
 \begin{equation}
  w_\tau = a(w_{yy}-w_y) - rw_y + I_\varphi(v_{yy}-v_y) + f
  \label{eq:pde1}
 \end{equation}
with homogeneous boundary conditions. By Young's inequality, $I_\varphi(v_{yy}-v_y) \in L^2(D)$. So, by Proposition A.1 in \cite{EggEng2005}, it follows that the PDE problem \eqref{eq:pde1} has a unique solution and $\|w\|_{W^{1,2}_2(D)}\leq C\|I_\varphi(v_{yy}-v_y)+f\|_{L^2(D)}$.
Again, by Young's inequality, $\|I_\varphi(v_{yy}-v_y)\|_{L^2(D)}\leq \|\varphi\|_{L^1(\R)}\|v_{yy}-v_y\|_{L^2(D)} \leq \|\varphi\|_{L^1(\R)}\|v\|_{W^{1,2}_2(D)}$. Since $\|\varphi\|_{L^1(\R)}< C^{-1}$, it follows that, for any $v\in W^{1,2}_2(D)$ with $v\not=0$, 
$\|w\|_{W^{1,2}_2(D)} < \|v\|_{W^{1,2}_2(D)} + C\|f\|_{L^2(D)}$. 
Let us see that $G$ is a contraction. For any $v_1,v_2 \in W^{1,2}_2(D)$, set $w_1 = G(v_1)$, $w_2 = G(v_2)$ and $w=w_1-w_2$. It follows that $w$ is the solution of \eqref{eq:pde1} with $f=0$ and $\|w_1-w_2\|_{W^{1,2}_2(D)} < \|v_1-v_2\|_{W^{1,2}_2(D)}$. So, $G$ is indeed a contraction in $W^{1,2}_2(D)$ and 
has a unique fixed point $\tilde w$, which is the unique solution of
\begin{equation}
 \tilde w_\tau = a(\tilde w_{yy}-\tilde w_y) - r\tilde w_y + I_\varphi(\tilde w_{yy}-\tilde w_y) + f,
 \label{eq:aux_pide}
\end{equation}
with homogeneous boundary conditions.

 Any solution $u$ of the PIDE problem \eqref{eq:pide2}-\eqref{eq:pide2b} can be written as  $u = \tilde w + \tilde u$, where $\tilde w$ is the solution of the PIDE problem \eqref{eq:aux_pide} with $f = -I_{\varphi}(\tilde u_{yy} - \tilde u_y)$ and $\tilde u$ the solution of \eqref{eq:pde1} with $\varphi =0$, $f=0$ and the same boundary and initial conditions as the PIDE problem \eqref{eq:pide2}-\eqref{eq:pide2b}. The existence and uniqueness of $\tilde u \in W^{1,2}_{2,loc}(D)$ is guaranteed by Corollary A.1 in \cite{EggEng2005}. Therefore, the assertion follows.
 \end{proof}
\begin{rem}
 Since $\tilde w$ in the above proof is a fixed point of $G$, it satisfies the inequality
$$
  \|\tilde w\|_{W^{1,2}_{p}(D)} \leq C\left(\|\varphi\|_{L^1(\R)}\|\tilde w\|_{W^{1,2}_{p}(D)}+\|f\|_{L^2(D)}\right) \mbox{ .}
$$
Assuming further that $\|\varphi\|_{L^1(\R)} \leq K/C$ with the constant $0<K<1$, we have that
\begin{equation}
 \|\tilde w\|_{W^{1,2}_{p}(D)} \leq \frac{C}{1-K}\|f\|_{L^2(D)}.
 \label{eq:aux_estimate}
\end{equation}
\end{rem}

\begin{df}
 The direct operator $F : \mathcal{D}(F) \rightarrow W^{1,2}_2(D)$ associates $(\tilde a,\varphi_-,\varphi_+)$ to $u(a,\varphi) - u(a_0,0)$, where $u(a,\varphi)$ is the solution of the PIDE problem \eqref{eq:pide2}-\eqref{eq:pide2b}, with $(a,\varphi)$ in $\mathcal{D}(F)$. In other words, $F(\tilde a,\varphi_-,\varphi_+)$ is the solution of the PIDE problem \eqref{eq:aux_pide} with homogeneous boundary condition and $f = -I_{\varphi}( u(a_0,0)_{yy} - u(a_0,0)_y)$.
 \label{def:direct_op}
 \end{df}

\begin{lem}
 For any $(a,\varphi)$ given by $\mathcal{D}(F)$, $u(a,\varphi)$ solution of the PIDE problem \eqref{eq:pide2}-\eqref{eq:pide2b} satisfies
\begin{equation}
 \|u(a,\varphi)_{yy} - u(a,\varphi)_y\|_{L^2(D)} \leq \displaystyle\frac{C}{1-K},
 \label{eq:unif_estimate}
\end{equation}
with $C$ and $K$ depending on the bounds of the coefficients $a$ and $\varphi$.
\end{lem}
\begin{proof}
  By Corollary~A.1 in \cite{EggEng2005}, $\|u(a_0,0)_{yy} - u(a_0,0)_y\|_{L^2(D)} < C$ for some constant $C$, and since $u(a,\varphi)-u(a_0,0) \in W^{1,2}_2(D)$ for any $(a,\varphi) \in \mathcal{D}(F)$, it follows that 
\begin{multline*}
  \|u(a,\varphi)_{yy} - u(a,\varphi)_y\|_{L^2(D)} = \|u(a,\varphi)_{yy} - u(a,\varphi)_y \pm (u(a_0,0)_{yy} - u(a_0,0)_y)\|_{L^2(D)}\\
  \leq \|(u(a,\varphi)-u(a_0,0))_{yy} - (u(a,\varphi)-u(a_0,0))_y\|_{L^2(D)} +\| u(a_0,0)_{yy} - u(a_0,0)_y\|_{L^2(D)}.
 \end{multline*}
 By Equation~\eqref{eq:aux_estimate} and Corollary~A.1 in \cite{EggEng2005},
 \begin{multline*}
  \|(u(a,\varphi)-u(a_0,0))_{yy} - (u(a,\varphi)-u(a_0,0))_y\|_{L^2(D)}\leq \|u(a_0,\varphi)-u(a_0,0)\|_{W^{1,2}_2(D)} \\ \leq \frac{C}{1-K}\|\varphi\|_{L^1(\R)}\|u(a_0,0)_{yy} - u(a_0,0)_y\|_{L^2(D)}.
 \end{multline*}
 Since $\|\varphi\|_{L^1(\R)} \leq K/C$, it follows that
 $$
 \|u(a,\varphi)_{yy} - u(a,\varphi)_y\|_{L^2(D)} \leq \displaystyle\frac{C}{1-K},
 $$
 for any $(a,\varphi)$ given by $\mathcal{D}(F)$.
\end{proof}
We can now state the following:
\begin{pr}
 The map $F : \mathcal{D}(F) \rightarrow W^{1,2}_2(D)$ is continuous.
  \label{prop:continuity}
\end{pr}
\begin{proof}
 Let the sequence $\{(\tilde a,\varphi_{-,n},\varphi_{+,n})\}_{n\in\N}$ in $\mathcal{D}(F)$ converge to $(\tilde a,\varphi_{-},\varphi_{+})$. We must show that $\|F(\tilde a,\varphi_{-,n},\varphi_{+,n}) - F(\tilde a,\varphi_{-},\varphi_{+})\|\rightarrow 0$.
 Define 
 $$w_n := F(\tilde a,\varphi_{-,n},\varphi_{+,n}) - F(\tilde a,\varphi_{-},\varphi_{+}) = u(a_n,\varphi_n) - u(a,\varphi).$$ 
 By the linearity of the PIDE problem \eqref{eq:pide2}-\eqref{eq:pide2b}, $w_n$ is the solution of the PIDE problem \eqref{eq:aux_pide} with $a$ and $\varphi$ replaced by $a_n$ and $\varphi_n$, respectively, homogeneous boundary conditions and 
 $$f_n = (a_n-a)(u(a,\varphi)_{yy} - u(a,\varphi)_y) - I_{\varphi_n-\varphi}(u(a,\varphi)_{yy} - u(a,\varphi)_y).$$
 So, by the estimate \eqref{eq:aux_estimate} and Young's inequality for convolutions,
 \begin{multline*}
  \|w_n\|_{W^{1,2}_2(D)}\leq \displaystyle\frac{C}{1-K}\|(a_n-a)(u(a,\varphi)_{yy} - u(a,\varphi)_y) - I_{\varphi_n-\varphi}(u(a,\varphi)_{yy} - u(a,\varphi)_y)\|_{L^2(D)}\\
  \leq \displaystyle\frac{C}{1-K}\left(\|(a_n-a)(u(a,\varphi)_{yy} - u(a,\varphi)_y)\|_{L^2(D)} +\right.\\
  \left. \|\varphi_n-\varphi\|_{L^1(\R)}\|u(a,\varphi)_{yy} - u(a,\varphi)_y)\|_{L^2(D)}\right)
 \end{multline*}
By the Sobolev embedding (see Theorem~7.75 in \cite{iorio}), 
it follows that
\begin{multline*}
\|(a_n-a)(u(a,\varphi)_{yy} - u(a,\varphi)_y)\|_{L^2(D)} \leq \|a_n-a\|_{L^\infty(D)}\|u(a,\varphi)_{yy} - u(a,\varphi)_y\|_{L^2(D)}\\
\leq c\|a_n-a\|_{H^{1+\varepsilon}(D)}\|u(a,\varphi)_{yy} - u(a,\varphi)_y\|_{L^2(D)}
\end{multline*}
The above estimate and Equation~\eqref{eq:unif_estimate} imply that 
\begin{multline*}
 \|(a_n-a)(u_{yy} - u_y)\|_{L^2(D)} + \|\varphi_n-\varphi\|_{L^1(\R)}\|u_{yy} - u_y\|_{L^2(D)}\\ \leq \displaystyle\frac{\tilde C}{1-K}\left(\|a_n-a\|_{H^{1+\varepsilon}(D)} + \|\varphi_n-\varphi\|_{L^1(\R)}\right)
\end{multline*}
Summarizing,
$$
\|w_n\|_{W^{1,2}_2(D)} \leq \left(\displaystyle\frac{\tilde C}{1-K}\right)^2\left(\|a_n-a\|_{H^{1+\varepsilon}(D)} + \|\varphi_n-\varphi\|_{L^1(\R)}\right),
$$
and the assertion follows.
\end{proof}

\begin{pr}
 The map $F : \mathcal{D}(F) \rightarrow W^{1,2}_2(D)$ is weakly continuous and compact.
  \label{prop:compactness}
\end{pr}
\begin{proof}
 Let the sequence $\{(\tilde{a}_n,\varphi_{-,n},\varphi_{+,n})\}_{n\in\N}$ in $\mathcal{D}(F)$ converge weakly to $(\tilde a,\varphi_{-},\varphi_{+})$. Proceeding as in the proof of Proposition~\ref{prop:continuity} 
 define
 $$w_n := F(\tilde{a}_n,\varphi_{-,n},\varphi_{+,n}) - F(\tilde a,\varphi_{-},\varphi_{+}) = u(a_n,\varphi_n) - u(a,\varphi).$$ 
So it satisfies the PIDE~\eqref{eq:aux_pide} with $a$ and $\varphi$ replaced by $a_n$ and $\varphi_n$, respectively, and homogeneous boundary condition. Furthermore, it   
satisfies
\begin{multline}
\|w_n\|_{W^{1,2}_2(D)} \leq \displaystyle\frac{C}{1-K}\left(\|(a_n-a)(u(a,\varphi)_{yy} - u(a,\varphi)_y)\|_{L^2(D)}\right. \\ 
+\left. \|I_{\varphi_n-\varphi}(u(a,\varphi)_{yy} - u(a,\varphi)_y)\|_{L^2(D)}\right). \label{sum2}
\end{multline}
We shall prove that each of the two terms on the RHS of Equation~\eqref{sum2} goes to zero as $n\rightarrow \infty$.
In either case, we decompose the set $D$ as the disjoint union $D = D_M\cup D_M^c$, where
$$D_M =  [0,T]\times [-M,M],$$
with $M>0$. 
Concerning the first term on the RHS of Equation~\eqref{sum2}, we have by Sobolev's embedding that 
\begin{multline}
 \|(a_n-a)(u_{yy} - u_y)\|_{L^2(D)} = \|(a_n-a)(u_{yy} - u_y)\|_{L^2(D_M)} + \|(a_n-a)(u_{yy} - u_y)\|_{L^2(D_M^c)}\\
 \leq \|a_n-a\|_{H^{1+\varepsilon/2}(D_M)}\|u_{yy} - u_y\|_{L^2(D_M)} + \|a_n-a\|_{H^{1+\varepsilon/2}(D_M^c)}\|u_{yy} - u_y\|_{L^2(D_M^c)} \label{sum3}
\end{multline}
By the compact immersion of $H^{1+\varepsilon}(D_M)$ into $H^{1+\varepsilon/2}(D_M)$ we have that weakly convergent
sequences of $H^{1+\varepsilon}(D_M)$ are sent into norm convergent ones in  $H^{1+\varepsilon/2}(D_M)$
(Proposition~IV.4.4 in \cite{Tay2011}). 
Thus, $\|a_n-a\|_{H^{1+\varepsilon/2}(D_M)}\rightarrow 0$.  Now, we recall that 
$\|u_{yy} - u_y\|_{L^2(D_M^c)}\rightarrow 0$ as $M\rightarrow +\infty$. 
To see that the RHS of Inequality~\eqref{sum3} goes to zero, note that, given $\eta>0$, for a sufficiently large $M$, $\|a_n-a\|_{H^{1+\varepsilon/2}(D_M^c)}\|u_{yy} - u_y\|_{L^2(D_M^c)}<\eta/2$, since $\|a_n-a\|_{H^{1+\varepsilon/2}(D_M^c)}$ is dominated by $\|a_n-a\|_{H^{1+\varepsilon/2}(D)}$, which is uniformly bounded. In addition, 
$\|u_{yy} - u_y\|_{L^2(D_M)}$ is bounded by $\|u_{yy} - u_y\|_{L^2(D)}$, which is finite. Thus, 
for all  sufficiently large $n\in \N$, and with the same $M$ of the previous estimate, $\|a_n-a\|_{H^{1+\varepsilon/2}(D_M)}\|u_{yy} - u_y\|_{L^2(D_M)} < \eta/2$. 

Concerning the convergence of the  second term in Equation~\eqref{sum2}, 
by Jensen's inequality, we have that  
\begin{multline*}
\|I_{\varphi_n-\varphi}(u(a,\varphi)_{yy} - u(a,\varphi)_y)\|^2_{L^2(D)}\\ \leq \|\varphi_n-\varphi\|_{L^1(\R)}\displaystyle\int_{D}\int_{\R}|\varphi_n(x-y)-\varphi(x-y)|(u_{yy}(\tau,y) - u_y(\tau,y))^2dy d\tau dx.
\end{multline*}
So, breaking it into the following two integrals we get 
\begin{multline*}
\displaystyle\int_{D}\int_{\R}|\varphi_n(x-y)-\varphi(x-y)|(u_{yy}(\tau,y) - u_y(\tau,y))^2dy d\tau dx \\ = 
\displaystyle\int_{D_M^c}\int_{\R}|\varphi_n(y)-\varphi(y)|(u_{yy}(\tau,x-y) - u_y(\tau,x-y))^2dy d\tau dx\\ +
\displaystyle\int_{D_M}\int_{\R}|\varphi_n(x-y)-\varphi(x-y)|(u_{yy}(\tau,y) - u_y(\tau,y))^2dy d\tau dx =: I_1 + I_2.
\end{multline*}
The integral $I_1$ goes to zero by the dominated convergence theorem as $M\rightarrow \infty$ (Theorem~1.50 in \cite{AdaFou2003}). 
By Fubini's Theorem, it follows that
$$
I_2 = 
\displaystyle\int_{D}(u_{yy}(\tau,y) - u_y(\tau,y))^2\int_{|x|\leq M}|\varphi_n(x-y)-\varphi(x-y)|dxd\tau dy
$$
For almost every $y\in \R$, the Rellich-Kondrachov theorem (Part II of Theorem~6.3 in \cite{AdaFou2003}) implies that $\int_{|x|\leq M}|\varphi_n(x-y)-\varphi(x-y)|dx$ goes to zero. Just recall that $\varphi|_{(-\infty,0)}\in W^{2,1}(-\infty,0)$ and $\varphi|_{(0,+\infty)}\in W^{2,1}(0,+\infty)$. 
By the estimate
$$I_2 \leq \|\varphi_n-\varphi\|_{L^1(\R)}\|u_{yy}-u_y\|^2_{L^2(D)} \leq \frac{2KC}{(1-K)^2} \mbox{, }$$
we can apply the dominated convergence theorem to get that $I_2$ goes to zero as $n\rightarrow \infty$, for each fixed $M$.  
Therefore, $\|w_n\|_{W^{1,2}_2(D)}\rightarrow 0$ and the assertion follows.
\end{proof}
We formally define the derivative of $F$ and then we show that it is in fact the Frech\'et derivative of $F$. 
\begin{df}
 The derivative of $F$ at $(a,\varphi)$ in the direction $h = (h_1,h_2) \in X$, such that $(a+h_1,\varphi+h_2)\in \mathcal{D}(F)$, is the solution of the PIDE problem \eqref{eq:aux_pide} with homogeneous boundary condition and
 $$
 f = h_1(u(a,\varphi)_{yy} - u(a,\varphi)_y) - I_{h_2}(u(a,\varphi)_{yy} - u(a,\varphi)_y),
 $$
 where $u(a,\varphi)$ denotes the solution of the PIDE problem \eqref{eq:pide2}-\eqref{eq:pide2b}.
 Such derivative is denoted by $F^\prime(a,\varphi)h$ or $u^\prime(a,\varphi)h$, and is in $W^{1,2}_2(D)$.
\end{df}

\begin{rem}
 By the proof of Proposition~\ref{prop:existence}, for any $h_1 \in H^{1+\varepsilon}(D)$ and any $h_2 \in L^1(\R)$ the PIDE problem of the definition above still has a solution in $W^{1,2}_2(D)$. In addition, such PIDE problem is linear with respect to $h = (h_1,h_2) \in X$.
 So, for every $(a,\varphi) \in \mathcal{D}(F)$, $h \mapsto F^\prime(a,\varphi)h$ is a linear and bounded map from $X$ to $W^{1,2}_2(D)$, satisfying
 \begin{equation}
  \|F^\prime(a,\varphi)h\| \leq \left(\displaystyle\frac{C}{1-K}\right)^2\|h\|_X.
 \end{equation}

 \end{rem}
\begin{pr}
 The map $F : \mathcal{D}(F) \rightarrow W^{1,2}_2(D)$ is Frech\'et differentiable and satisfies
 \begin{multline}
  \|F(a+h_1,\varphi+h_2) - F(a,\varphi) - F^\prime(a,\varphi)h\|_{W^{1,2}_2(D)} \\
  \leq \displaystyle\frac{C}{1-K}\|h\|_X\|F(a+h_1,\varphi+h_2) - F(a,\varphi)\|_{W^{1,2}_2(D)},
  \label{eq:tangential_pide}
 \end{multline}
for any $(a,\varphi) \in \mathcal{D}(F)$ and any $h = (h_1,h_2)\in X$, such that $(a+h_1,\varphi+h_2) \in \mathcal{D}(F)$.
 \label{prop:frechet}
\end{pr}
\begin{proof}
 Let $(a,\varphi) \in \mathcal{D}(F)$ be fixed and $h = (h_1,h_2)\in X$ be such that $(a+h_1,\varphi+h_2)\in \mathcal{D}(F)$. Define
 $$
 w = F(a+h_1,\varphi+h_2) - F(a,\varphi) - F^\prime(a,\varphi)h.
 $$
 and 
 $v = F(a+h_1,\varphi+h_2) - F(a,\varphi)$. 
 By the linearity of the PIDE problems \eqref{eq:pide2}-\eqref{eq:pide2b} and \eqref{eq:aux_pide}, $w$ is the solution of the PIDE problem \eqref{eq:aux_pide}, with homogeneous boundary conditions and 
 $$f = -h_1(v_{yy} - v_y) + I_{h_2}(v_{yy} - v_y).$$
 So, $w$ satisfies the estimate
$$
 \|w\|_{W^{1,2}_2(D)}\leq \displaystyle\frac{C}{1-K}\|-h_1(v_{yy} - v_y) + I_{h_2}(v_{yy} - v_y)\|_{L^2(D)} \mbox{ .}
$$
By the triangle inequality, Young's inequality for the convolution and Sobolev's embedding's theorem, it follows that
\begin{multline*}
\|-h_1(v_{yy} - v_y) + I_{h_2}(v_{yy} - v_y)\|_{L^2(D)} \leq \|v_{yy}-v_y\|_{L^2(D)}\left(\|h_1\|_{H^{1+\varepsilon}(D)} + \|h_2\|_{L^1(\R)}\right)\\
 \leq \|v\|_{W^{1,2}_2(D)}\|h\|_{X},
\end{multline*}
and the asserted estimate holds.

The set $\mathcal{D}(F)$ has a nonempty interior, $h\mapsto F^\prime(a,\varphi)h$ is a bounded linear map from $X$ to $W^{1,2}_2(D)$, and the estimate \eqref{eq:tangential_pide} implies that 
$$
\displaystyle\lim_{\|h\|_X\rightarrow 0}\frac{\|F(a+h_1,\varphi+h_2) - F(a,\varphi) - F^\prime(a,\varphi)h\|_{W^{1,2}_2(D)}}{\|h\|_X}=0 \mbox{ .}
$$
Thus, $F$ is Frech\'et differentiable. 
\end{proof}


\section{Splitting Strategy and Regularization}\label{sec:splitting} 
In this section, under an abstract setting, we consider a Tikhonov-type regularization of the simultaneous calibration of two parameters from a set of observations. A splitting strategy is used to solve the resulting minimization problem. Results concerning the convergence of this approach to an approximate solution of the inverse problem are provided. They rely on certain assumptions which will be shown to hold for the calibration problem at hand of jump-diffusion local volatility models. 

\subsection{Tikhonov-type Regularization}\label{sec:tikhonov}
Firstly, let us introduce some basic notions of Tikhonov-type regularization.  This methodology has been used extensively for the solution of ill-posed inverse problems. See \cite{schervar} and \cite{ern} for more details.

Consider the map $F:\mathcal{D}(F)\subset X \rightarrow Y$ between two Banach spaces $X$ and $Y$. Given $\tilde y$ in the range of $F$, $\mathcal R(F)$, find some $x \in \mathcal{D}(F)$ solution of the equation:
\begin{equation}
 \tilde y = F(x).
 \label{eq:inverse_problem1}
\end{equation}
Since there may be more than one element in $\mathcal{D}(F)$ solving \eqref{eq:inverse_problem1}, it is common to search for a solution $x^\dagger$ that minimizes some convex functional  $f_{x_0}:\mathcal{D}(f_{x_0})\subset X\rightarrow \R_+$, which is related to some {\em a priori} information. 
So,  
$$
x^\dagger \in \argmin\left\{f_{x_0}(x) ~:~x\in\mathcal{D}(F)\,\mbox{ and }\, F(x) = \tilde y\right\},
$$
and is so-called a $f_{x_0}$-minimizing solution.

In general, it is not possible to have access to the data $\tilde y$ in $\mathcal R(F)$, but only some imperfect approximation $y^\delta \in Y$ satisfying
\begin{equation}
 \|P\tilde y-y^\delta\|_Y \leq \delta,
 \label{eq:data}
\end{equation}
where $\delta > 0$ is the noise level and $P:Y\rightarrow Y$ is a projection onto some subspace of $Y$, where $y^\delta$ is defined. For example, $P$ can define the observation of $y$ in some discrete mesh.

Since the inverse problem \eqref{eq:inverse_problem1} can be ill-posed, Tikhonov-type regularization is applied, i.e., we must find an element of $\mathcal{D}(F)$ that minimizes the (Tikhonov-type) functional:
\begin{equation}
 \mathcal F(x) = \phi(x) + \alpha f_{x_0}(x),
 \label{eq:tikhonov1}
\end{equation}
where
\begin{equation}
 \phi(x) = \|F(x) - y^\delta\|^p_Y
 \label{eq:data_misfit}
\end{equation}
is the {\em data misfit} or {\em merit function}, and $\alpha > 0$ is a constant so-called {\em regularization parameter}. The penalization $f_{x_0}$ is called the {\em regularization functional}. The minimizers of \eqref{eq:tikhonov1} in $\mathcal{D}:=\mathcal{D}(F)\cap\mathcal{D}(f_{x_0})$ are called Tikhonov minimizers or reconstructions, and are denoted by $x^{\delta}_{\alpha}$.

The framework of convex regularization will now be used. See~\cite{schervar} for more information. 
In what follows, we shall need:
\begin{ass}[Assumption~3.13 in \cite{schervar}] Let us assume that 

\begin{enumerate}
  \item The topologies  $T_X$ and $T_Y$ associated to $X$ and $Y$, respectively, are weaker than the corresponding norm topologies.
  \item The exponent in Equation~\eqref{eq:data_misfit} satisfies $p \geq 1$.
  \item The norm of $Y$ is sequentially lower semi-continuous with respect to $T_Y$.
  \item $f_{x_0}$ is convex and continuous with respect to $T_X$.
  \item The objective set satisfies $\mathcal{D}\not= \emptyset$, and $\mathcal D$ has a nonempty interior.
  \item For every $\alpha > 0$ and $M > 0$ the level set
  $$
  M_{\alpha}(M) := \left\{x\in \mathcal D~:~ \mathcal F(x) \leq M\right\}
  $$
  is sequentially pre-compact with respect to $T_X$.
  \item For every $\alpha > 0$ and $M > 0$ the level set $M_{\alpha}(M)$ is sequentially closed w.r.t. $T_X$ and the restriction of $F$ to $M_{\alpha}(M)$ is sequentially continuous w.r.t. $T_X$ and $T_Y$.
 \end{enumerate}
 \label{ass:2}
\end{ass}

By Assumption~\ref{ass:2}, the existence of stable Tikhonov minimizers is guaranteed by Theorems~3.22 and 3.23 in \cite{schervar}. If the inverse problem in \eqref{eq:inverse_problem1} has a solution, then, also based on Assumption~\ref{ass:2}, Theorem~3.25 in \cite{schervar} says that there exists an $f_{x_0}$-minimizing solution of \eqref{eq:inverse_problem1} and Theorem~3.26 states the convergence of a 
sequence of Tikhonov minimizers to an $f_{x_0}$-minimizing solution whenever $\delta\rightarrow 0$ and $\alpha = \alpha(\delta)$ satisfies the limits:
\begin{equation}
 \lim_{\delta\rightarrow 0}\alpha(\delta) = 0 ~\mbox{ and }~ \lim_{\delta\rightarrow 0}\frac{\delta^p}{\alpha(\delta)}=0.
 \label{eq:limits_alpha}
\end{equation}

\subsection{A Splitting Strategy Algorithm}

The presence of jumps together with the diffusive parts motivates separating the regularization into two parts. In this section we shall now describe such approach in the general framework of convex regularization.

Let $X$ be given by $X:= W \times Z$, where $W$ and $Z$ are Banach spaces. 
Consider $T_W$ and $T_Z$ two topologies of $W$ and $Z$, respectively, which are assumed to be  weaker than the norm topologies of each of the corresponding spaces. So, $X$ will be endowed with two natural topologies: The norm 
$\|(w,z)\|_X = \|w\|_W+\|z\|_Z$ and the product topology $T_X:=T_W\times T_Z$ which is weaker than the norm topology. 
Consider again the operator $F:\mathcal{D}(F)\subset X \rightarrow Y$, where $F(x) = F(w,z)$.

The penalty term in \eqref{eq:tikhonov1} can be rewritten as
\begin{equation}
 \alpha f_{x_0}(x) = \alpha f_{x_0}(w,z) = \alpha \beta_1g_{w_0}(w) + \alpha \beta_2 h_{z_0}(z) = \alpha_1g_{w_0}(w) + \alpha_2 h_{z_0}(z),
 \label{eq:fx0}
\end{equation}
where $\alpha_j = \alpha\cdot\beta_j$ with $\beta_j \geq 0$, $j=1,2$, and the functionals $g_{w_0}$ and $h_{z_0}$ are convex and continuous w.r.t. $T_W$ and $T_Z$, respectively. So, the Tikhonov-type functional now reads:
\begin{equation}
 \mathcal F(w,z) = \phi(w,z) + \alpha_1g_{w_0}(w) + \alpha_2 h_{z_0}(z).
 \label{eq:tikhonov2}
\end{equation}

Let us assume that Assumption~\ref{ass:2} holds. Thus, if $\alpha_1,\alpha_2>0$, $\mathcal F(w,z)$ has minimizers in $\mathcal D$. Since the norm topology of $X$ and $T_X$ are defined by the products of the norm topologies of $W$ and $Z$, and $T_W$ and $T_Z$, respectively, the projection operators $P_W:(w,z)\mapsto w$ and $P_Z:(w,z)\mapsto z$ are continuous with respect to the norm topologies of $X$, $W$ and $Z$ and to $T_X$, $T_W$ and $T_Z$.

For each $z \in Z$, define the operator $F_z:P_W(\mathcal D) \subset W \rightarrow Y$ as $F_z(w) = F(w,z)$, the Tikhonov-type functional $\mathcal F_z(w) = \mathcal F(w,z)$, and the set $\mathcal D_z = P_W(\mathcal D)\times\{z\}$. Similarly $F_w$, $\mathcal F_w$ and $\mathcal D_w$ are defined. 

Assume also that Items~5, 6 and 7 in Assumption~\ref{ass:2} remain valid whenever $\mathcal D$ is replaced by $P_W(\mathcal D)$ or $P_Z(\mathcal D)$, $F$ by $F_w$ or $F_z$ and $\mathcal F$ by $\mathcal F_w$ or $\mathcal F_z$. In this case, Theorems~3.22 and 3.23 in \cite{schervar} guarantee the existence of stable Tikhonov minimizers of $\mathcal F_w$ and $\mathcal F_z$, for each $w\in P_W(\mathcal D)$ and $z\in P_Z(\mathcal D)$.

Our approach is to split the iteration so that at each step the jump and the diffusive component are updated successively. More precisely, 
For any $w\in P_W(\mathcal D)$ (or $z\in P_Z(\mathcal D)$), set $w^0 = w$ ($z^0 = z$) and consider the iterations with $n\in\N$:
\begin{eqnarray}
   z^n \in \argmin\left\{\mathcal F_{w^{n-1}}(z) ~:~ z\in P_Z(\mathcal D)\right\}\nonumber\\
   w^n \in \argmin\left\{\mathcal F_{z^n}(w) ~:~ w\in P_W(\mathcal D)\right\}.
   \label{algorithm1}
\end{eqnarray}
Repeat the iterations until some termination criteria.

If the algorithm starts with $z$ instead of $w$, the order of the two iterations must be reversed.

\begin{df}
 A stationary point of the functional $\mathcal F$ is some point $\hat x = (\hat w, \hat z) \in \mathcal D$, such that
 \[
 \hat w\in \argmin\{\mathcal F_{\hat z}(w)~:~ w \in P_W(\mathcal D)\}
 ~\mbox{ and }~
\hat z\in \argmin\{\mathcal F_{\hat w}(z)~:~ z \in P_Z(\mathcal D)\}.
 \]
\end{df}

In what follows we shall assume the continuity of $\mathcal F$ with respect to $T_X$. This holds, for example, if $F$ and $f_{x_0}$ are $T_X$-continuous. This hypothesis is necessary in the proof of the following proposition.

\begin{pr}
 For every initializing pair $(w,z) \in \mathcal{D}$, any convergent subsequence produced by the algorithm of Equation~\eqref{algorithm1} converges to some stationary point of  $\mathcal F$.
 \label{prop:splitting1}
\end{pr}
 \begin{proof}
Consider the sequence $\{(w^n,z^n)\}_{n\in\N}$ defined by the iterations in \eqref{algorithm1}. By construction, the sequence $\{\mathcal F (w^n,z^n)\}_{n\in\N}$ is non-increasing and bounded, and thus it converges. In addition, $\{(w^n,z^n)\}_{n\in\N}$ is a subset of some level set $\mathcal{M}_\alpha(M)$, which is $T_X$-pre-compact by Item~6 in Assumption~\ref{ass:2}. 
For every cluster point $(\overline w, \overline z)$ of $\{(w^n,z^n)\}_{n\in\N}$, $\mathcal F(\overline w,\overline z) \leq \mathcal F(w^n,z^n)$ for all $n\in\N$. 

Given $w \in P_W(\mathcal D)$, it follows that $\mathcal F(w,\overline z) = \lim_{k\rightarrow \infty} \mathcal F (w,z^{n_k})$ by the $T_X$-continuity of $\mathcal F$, since the subsequence $\{(w^{n_k},z^{n_k})\}_{k\in\N}$ converges to $(\overline w,\overline z)$ w.r.t $T_X$. So, for each $k\in\N$, 
\[\mathcal F(w,z^{n_k}) \geq \mathcal F(w^{n_{k}+1},z^{n_k}),\]
because $w^{n_k}+1$ is a minimizer of $\mathcal F_{z^{n_k}}$. Applying more steps of the algorithm of Equation~\eqref{algorithm1}, it follows that
$$
\mathcal F(w^{n_k+1},z^{n_k}) \geq \mathcal F(w^{n_k+1},z^{n_k+1})\geq \cdots \geq \mathcal F(w^{n_{k+1}},z^{n_{k+1}}).
$$
So, $\mathcal F(w,z^{n_k})\geq \mathcal F(w^{n_{k+1}},z^{n_{k+1}})$.
In addition, for every $k\in\N$,
\[\mathcal F(w^{n_{k}},z^{n_k}) \geq \mathcal F(\overline w,\overline z).\]
Hence, $\overline w$ is a minimizer of $\mathcal F_{\overline z}$. 
To see that $\overline z$ is a minimizer of $\mathcal F_{\overline w}$, note that, for any $z \in P_Z(\mathcal{D}$,
$
\mathcal{F}(\overline w,z) = \lim_{k\rightarrow \infty}\mathcal F(w^{n_k},z).
$
Since $z^{n_k}$ is a minimizer of $\mathcal F_{w^{n_k}}$, it follows that, 
$\mathcal F(w^{n_k},z)\geq \mathcal F(w^{n_k},z^{n_k})$. By the fact that $\mathcal F(w^{n_k},z^{n_k})\geq F(\overline w, \overline z)$, for every $k\in\N$, the assertion follows.
\end{proof}
Denote the stationary point obtained with Algorithm~\eqref{algorithm1} by $(\overline w^\delta_{\alpha},\overline z^\delta_{\alpha})$. Note that a stationary point need not to be a Tikhonov minimizer, since, in principle, it can be a saddle point. However, we shall see in Proposition~\ref{pr:convergence} that such stationary point is indeed an approximation of the inverse problem solution.  

\begin{rem}
Recall the definition of the sub-differential of a convex function $f: \mathcal{D}(f) \subset X \rightarrow \R$ at the point $\overline x \in \mathcal{D}(f)$, with $X$ a Banach space, which is the set $\partial f(x)$ of elements $x^*$ in the dual space $X^*$ satisfying
 $$
 f(x) - f(\overline x) - \langle x^*,x-\overline x\rangle \geq 0 ~\forall x \in \mathcal{D}(f).
 $$
If $w\mapsto\phi(w,z)$ and $z\mapsto\phi(w,z)$ are Frech\'et differentiable, it follows that, for each $w$ and $z$,
$$
\partial \mathcal F_z(w)  = \left\{\frac{\partial }{\partial w}\phi( w, z)\right\} + \alpha_1 \partial g_{w_0}( w)
\quad\mbox{and}\quad
\partial \mathcal F_w(z)  = \left\{\frac{\partial }{\partial z}\phi( w, z)\right\} + \alpha_2 \partial h_{z_0}( z).
$$
Moreover, $\phi$ is also Frech\'et differentiable and
$$
\partial \mathcal F(w,z) = \left\{\left(\frac{\partial }{\partial w}\phi(w,z),\frac{\partial }{\partial z}\phi(w,z)\right)\right\} + \{\alpha_1 \partial g_{w_0}(w)\} \times\{\alpha_2 \partial h_{z_0}(z)\}.
$$
\label{rem:1}
\end{rem}
For a proof of Remark~\ref{rem:1}, see Item~(c) of Exercise~8.8 and Proposition~10.5 in \cite{RocWet2009}. 
So, if $0 \in \partial \mathcal F_w(z)$ and $0 \in \partial \mathcal F_z(w)$, then $0 \in \partial \mathcal F(w,z)$. 

Let $(\overline w, \overline z)$ denote the stationary point obtained with the algorithm of Equation~\eqref{algorithm1}, this means that $\overline w$ is a local minimum of $\mathcal F_{\overline z}$ and $\overline z$ is a local minimum of $\mathcal F_{\overline w}$. By Theorem~10.1 in \cite{RocWet2009}, $0 \in \partial \mathcal F_{\overline w}(\overline z)$ and $0 \in \partial \mathcal F_{\overline z}(\overline w)$. So, $0 \in \partial \mathcal F(\overline w, \overline z)$. If, in addition, $\mathcal{F}$ is convex, then, $(\overline w, \overline z)$ is a Tikhonov minimizer.


\begin{df}
  A stationary point $(\overline w,\overline z)$ with data $y^\delta$ is stable, if for every sequence $\{y_k\}_{k\in\N} \subset Y$ such that $y_k\rightarrow y^\delta$ in norm, then, $\{(\overline w^k,\overline z^k)\} \subset \mathcal D$, the sequence of solutions obtained with the algorithm of Equation~\eqref{algorithm1} considering the data $y_k$ for each $k\in\N$ has a $T_X$-convergent subsequence. In addition, the limit of every $T_X$-convergent subsequence $\{(\overline w^{k_l},\overline z^{k_l})\}$ is a stationary point of the Tikhonov functional $\mathcal F$ with data $y^\delta$.
  \label{def:stability}
\end{df}

\begin{pr}
The stationary point obtained by the algorithm of Equation~\eqref{algorithm1} is stable.
\label{prop:splitting2}
\end{pr}
\begin{proof}
Consider the sequences $\{y_k\}_{k\in\N} \subset Y$ and $\{(\overline w^k,\overline z^k)\} \subset \mathcal D$ as in Definition~\ref{def:stability}. Firstly, it is necessary to prove that $\{(\overline w^k,\overline z^k)\} \subset \mathcal D$ has a convergent subsequence. By Lemma~3.21 in \cite{schervar},
$$
\mathcal F(\overline w^k,\overline z^k;y^\delta) 
\leq 2^{p-1}\mathcal F(\overline w^k,\overline z^k;y_k) + 2^{p-1}\|y_k-y^\delta\|^p.
$$
The sequence $\{y_k\}_{k\in\N}$ converges to $y^\delta$, so $\|y_k-y^\delta\|^p$ is uniformly bounded in $k$. In addition, if we assume further that, for each $y_k$, the algorithm of Equation~\eqref{algorithm1} is initialized with the same $(w^0,z^0)$, it follows that
$$
\mathcal F(\overline w^k,\overline z^k;y_k) \leq \mathcal F( w^0,z^0;y_k),
$$
and applying Lemma~3.21 in \cite{schervar} again, 
$$\mathcal F( w^0,z^0;y_k) \leq 2^{p-1}\mathcal F( w^0,z^0;y^\delta)+2^{p-1}\|y_k-y^\delta\|^p.$$
By the estimates above,
$$
\mathcal F(\overline w^k,\overline z^k;y^\delta) \leq 4^{p-1}\mathcal F( w^0,z^0;y^\delta) + (4^{p-1}+2^{p-1})\|y_k-y^\delta\|^p,
$$
which implies that $\{(w^k,z^k)\}_{k\in\N}$ is a subset of some level set of $\mathcal F(\cdot,\cdot\cdot;y^\delta)$. Item~6 in Assumption~\ref{ass:2} implies that such level set is $T_X$-pre-compact, and the assertion follows.

Suppose with no loss of generality that $\{(\overline w^k,\overline z^k)\}$ converges to $(\tilde w, \tilde z)$, w.r.t. $T_X$. For every $w \in P_W(\mathcal D)$, since $\overline w^k$  is in $\argmin \mathcal F_{\overline z^k;y_k}(w)$,
\[
 \mathcal F(\tilde w,\tilde z;y^\delta) \leq \liminf_{k\rightarrow \infty}\mathcal F(\overline w^k,\overline z^k;y_k)
\leq \lim_{k\rightarrow \infty}\mathcal F(w,\overline z^k,y_k) = \mathcal F(w,\tilde z;y^\delta).
\]
So, $\tilde w$ is in $\argmin \mathcal F_{\tilde z;y^\delta}(w)$. Similarly, it follows that $\tilde z$ is in $\argmin \mathcal F_{\tilde w;y^\delta}(z)$ and the assertion follows.
\end{proof}

Since the stationary point obtained by the algorithm in Equation~\eqref{algorithm1} is determined w.r.t $y^\delta$ and the regularization parameters $\alpha_1$ and $\alpha_2$, let us denote it by $(\overline w^\delta_{\alpha_1,\alpha_2},\overline z^\delta_{\alpha_1,\alpha_2})$. 
Let us also denote by $x^\dagger = (w^\dagger,z^\dagger)$ an $f_{x_0}$-minimizing solution of the inverse problem in \eqref{eq:inverse_problem1}, and by $\tilde y$ the noiseless data in \eqref{eq:inverse_problem1}.

\paragraph{Tangential Cone Condition} Now, we show that the tangential cone condition is a sufficient condition for the splitting strategy algorithm of Equation~\eqref{algorithm1} to converge to some approximation of an $f_0$-minimizing solution of the inverse problem \eqref{eq:inverse_problem1}.

\begin{ass}
 Let the operator $F$ be Fr\'echet differentiable on each variable $w$ and $z$, so it is Fr\'echet differentiable and its Fr\'echet derivative satisfies
$$
F^\prime(x) = (\partial_w F(w,z), \partial_z F(w,z)).
$$
 In addition, there exist positive constants $r>0$ and $0 \leq \eta < 1/2$ such that, if $x,\tilde x$ are in the ball $B(x^*;r)$,  centered at $x^*$ with radius $r$, then the {\em tangential cone condition} is satisfied:
 $$
 \|F(\tilde x) - F(x) - F^\prime(x)(\tilde x -x)\| \leq \eta\|F(\tilde x) - F(x)\|.
 $$
 \label{ass:3}
\end{ass}
So, we can state the following result:
\begin{pr}
 Let Assumptions \ref{ass:2} and \ref{ass:3} hold. If the initializing pair and $x^\dagger$ are inside the ball $B(x^*;r)$ and $\lambda > (1+\eta)/(1-\eta)$ is fixed, then, for any pair of regularization parameters $(\alpha_1,\alpha_2)$ with sufficiently small entries,  there exists some finite $n = n(\alpha_1,\alpha_2)$ such that the iterates of the splitting algorithm satisfy
 \begin{equation}
  \|F(w^n,z^n) - y^\delta\|_Y \geq \lambda \delta > \|F(w^{n+1},z^{n+1}) - y^\delta\|_Y.
  \label{eq:discrepancy}
 \end{equation}
 \label{pr:saddle_point}
\end{pr}

\begin{proof}
 By Proposition~\ref{prop:splitting1}, the splitting algorithm converges to $(\overline w^\delta_{\alpha_1,\alpha_2},\overline z^\delta_{\alpha_1,\alpha_2})$, a stationary point of the functional in \eqref{eq:tikhonov2}. Since the operator $F$ is Fr\'echet differentiable, by Remark~\ref{rem:1}, zero is in the sub-differential of $\mathcal F$ at $(\overline w^\delta_{\alpha_1,\alpha_2},\overline z^\delta_{\alpha_1,\alpha_2})$. In other words, there exists $\gamma \in \partial g_{w_0}(w)$ and $\beta\in \partial h_{z_0}(z)$, such that
 \begin{multline*}
 0 = \left(\frac{\partial }{\partial w}\phi(\overline w^\delta_{\alpha_1,\alpha_2},\overline z^\delta_{\alpha_1,\alpha_2}),\frac{\partial }{\partial z}\phi(\overline w^\delta_{\alpha_1,\alpha_2},\overline z^\delta_{\alpha_1,\alpha_2})\right) + (\alpha_1 \gamma, \alpha_2 \beta) =\\ F^\prime(\overline x^\delta_{\alpha_1,\alpha_2})^*J(F(\overline x^\delta_{\alpha_1,\alpha_2}) - y^\delta) + (\alpha_1 \gamma, \alpha_2 \beta),
 \end{multline*}
 where $J:Y\rightarrow Y^*$ is the duality map, and $\overline x^\delta_{\alpha_1,\alpha_2} = (\overline w^\delta_{\alpha_1,\alpha_2},\overline z^\delta_{\alpha_1,\alpha_2})$. See \cite{MarRie2014} and Chapter~II in \cite{Cio1990} for more details on duality maps. 
 Applying $\overline x^\delta_{\alpha_1,\alpha_2} - x^\dagger$ on both sides of the above equality, we have:
 $$
 \alpha_1\langle \gamma,w^\dagger - \overline w^\delta_{\alpha_1,\alpha_2}\rangle + \alpha_2\langle \beta, z^\dagger - \overline z^\delta_{\alpha_1,\alpha_2}\rangle = \langle J(F(\overline x^\delta_{\alpha_1,\alpha_2})-y^\delta), F^\prime(\overline x^\delta_{\alpha_1,\alpha_2})(\overline x^\delta_{\alpha_1,\alpha_2} - x^\dagger)\rangle.
 $$
 Note that,
 \begin{multline*}
  \langle J(F(\overline x^\delta_{\alpha_1,\alpha_2})-y^\delta), F^\prime(\overline x^\delta_{\alpha_1,\alpha_2})(\overline x^\delta_{\alpha_1,\alpha_2} - x^\dagger)\rangle  =\\
  \|F(\overline x^\delta_{\alpha_1,\alpha_2})-y^\delta\|^p - \langle J(y^\delta-F(\overline x^\delta_{\alpha_1,\alpha_2})), y^\delta - F(\overline x^\delta_{\alpha_1,\alpha_2}) - F^\prime(\overline x^\delta_{\alpha_1,\alpha_2})(x^\dagger -\overline x^\delta_{\alpha_1,\alpha_2})\rangle \geq\\
  \|F(\overline x^\delta_{\alpha_1,\alpha_2})-y^\delta\|^p - \|F(\overline x^\delta_{\alpha_1,\alpha_2})-y^\delta\|^{p-1}\|y^\delta - F(\overline x^\delta_{\alpha_1,\alpha_2}) - F^\prime(\overline x^\delta_{\alpha_1,\alpha_2})(x^\dagger-\overline x^\delta_{\alpha_1,\alpha_2})\| \geq \\
  \|F(\overline x^\delta_{\alpha_1,\alpha_2})-y^\delta\|^p +\\- \|F(\overline x^\delta_{\alpha_1,\alpha_2})-y^\delta)\|^{p-1}\left(\delta + \|F(x^\dagger) - F(\overline x^\delta_{\alpha_1,\alpha_2}) - F^\prime(\overline x^\delta_{\alpha_1,\alpha_2})(x^\dagger - \overline x^\delta_{\alpha_1,\alpha_2})\|\right)\geq \\
  \|F(\overline x^\delta_{\alpha_1,\alpha_2})-y^\delta\|^p - \|F(\overline x^\delta_{\alpha_1,\alpha_2})-y^\delta\|^{p-1}\left(\delta + \eta\|F(x^\dagger) - F(\overline x^\delta_{\alpha_1,\alpha_2})\|\right)\geq \\
  \|F(\overline x^\delta_{\alpha_1,\alpha_2})-y^\delta\|^p - \|F(\overline x^\delta_{\alpha_1,\alpha_2})-y^\delta\|^{p-1}\left((1+\eta)\delta + \eta\|y^\delta - F(\overline x^\delta_{\alpha_1,\alpha_2})\|\right).
 \end{multline*}
Let us assume, by contradiction, that there is no $\alpha_1,\alpha_2 > 0$ such that $\phi(w^\delta_{\alpha_1,\alpha_2},\overline z^\delta_{\alpha_1,\alpha_2}) < \lambda^p\delta^p$. So, by the above estimates and assuming that $\lambda > (1+\eta)/(1-\eta)$, it follows
\begin{multline*}
\alpha_1\langle \gamma,w^\dagger - \overline w^\delta_{\alpha_1,\alpha_2}\rangle + \alpha_2\langle \beta, z^\dagger - \overline z^\delta_{\alpha_1,\alpha_2}\rangle \geq\\ \|F(\overline x^\delta_{\alpha_1,\alpha_2})-y^\delta\|^p - \|F(\overline x^\delta_{\alpha_1,\alpha_2})-y^\delta)\|^{p-1}\left((1+\eta)\delta + \eta\|y^\delta - F(\overline x^\delta_{\alpha_1,\alpha_2})\|\right) \geq \\
\|F(\overline x^\delta_{\alpha_1,\alpha_2})-y^\delta\|^p\left(1 - \eta - \displaystyle\frac{1+\eta}{\lambda}\right) \geq (\lambda\delta)^p\left(1 - \eta - \displaystyle\frac{1+\eta}{\lambda}\right).
\end{multline*}
Since $g_{w_0}$ and $h_{z_0}$ are convex,
$$
\alpha_1 (g_{w_0}(w^\dagger)-g_{w_0}(w^\delta_{\alpha_1,\alpha_2})) + \alpha_2(h_{z_0}(z^\dagger)-h_{z_0}(z^\delta_{\alpha_1,\alpha_2})) \geq \alpha_1\langle \gamma,\overline w^\dagger - w^\delta_{\alpha_1,\alpha_2}\rangle + \alpha_2\langle \beta, z^\dagger - \overline z^\delta_{\alpha_1,\alpha_2}\rangle.
$$
Summarizing,
\begin{equation}
 \alpha_1 (g_{w_0}(w^\dagger)-g_{w_0}(w^\delta_{\alpha_1,\alpha_2})) + \alpha_2(h_{z_0}(z^\dagger)-h_{z_0}(z^\delta_{\alpha_1,\alpha_2})) \geq (\lambda\delta)^p\left(1 - \eta - \displaystyle\frac{1+\eta}{\lambda}\right),
 \label{eq:4}
\end{equation}
and the right-hand side of the inequality~\eqref{eq:4} is positive. Since $\alpha_1,\alpha_2>0$, by the estimate above, $g_{w_0}(w^\dagger)-g_{w_0}(w^\delta_{\alpha_1,\alpha_2})\geq 0$ and $h_{z_0}(z^\dagger)-h_{z_0}(z^\delta_{\alpha_1,\alpha_2}) \geq 0$. 
If $K = \max\{g_{w_0}(w^\dagger),h_{z_0}(z^\dagger)\}$, then,
$$
2(\alpha_1 + \alpha_2)K \geq \alpha_1 (g_{w_0}(w^\dagger)-g_{w_0}(w^\delta_{\alpha_1,\alpha_2})) + \alpha_2(h_{z_0}(z^\dagger)-h_{z_0}(z^\delta_{\alpha_1,\alpha_2})).
$$
Hence, we can find $\alpha_1,\alpha_2>0$ such that the left-hand side of \eqref{eq:4} becomes smaller than the right-hand side, which is a contradiction.  Therefore, there must exist $\alpha_1^+,\alpha_2^+>0$ such that $$\phi(w^\delta_{\alpha_1^+,\alpha_2^+},\overline z^\delta_{\alpha_1^+,\alpha_2^+}) < (\lambda\delta)^p \mbox{ ,}$$ 
for each fixed $\lambda>(1+\eta)/(1-\eta)$. By the continuity of $\phi(\cdot,\cdot\cdot)$ with respect to the norm topology of $X$ and to the topology $T_X$ the existence of some finite iterate number $n$ holds. To see that this also holds for any sufficiently small regularization parameters $\alpha_1,\alpha_2$, just note that, by the same arguments above, there is no sequence $\{\alpha_1^n,\alpha_2^n\}_{n\in\N}$, of regularization parameters with $\alpha_1^n,\alpha_2^n\searrow 0$, such that $\phi(\overline w^\delta_{\alpha_1^n,\alpha_2^n},\overline z^\delta_{\alpha_1^n,\alpha_2^n}) \geq \lambda^p \delta^p$ for every $n\in\N$. So, there must be $\alpha_1^+,\alpha_2^+>0$, such that, for any $\alpha_1\in (0,\alpha_1^+)$ and $\alpha_2\in(0,\alpha_2^+)$. It follows that $\phi(\overline w^\delta_{\alpha_1,\alpha_2},\overline z^\delta_{\alpha_1,\alpha_2}) < \lambda^p \delta^p$.
\end{proof}

As a corollary of the proof above, we have the following estimate:
\begin{multline}
  (1-\eta)\|F(\overline w_{\alpha_1,\alpha_2}^\delta,\overline z_{\alpha_1,\alpha_2}^\delta)-y^\delta\|^{p} - (1+\eta)\delta\|F(\overline w_{\alpha_1,\alpha_2}^\delta,\overline z_{\alpha_1,\alpha_2}^\delta)-y^\delta\|^{p-1}\\ \leq \alpha_1\left[g_{w_{0}}(w^\dagger)-g_{w_{0}}(\overline w_{\alpha_1,\alpha_2}^\delta)\right] +\alpha_2\left[h_{z_0}(z^\dagger) -h_{z_0}(\overline w_{\alpha_1,\alpha_2}^\delta)\right].
  \label{eq:aux_ineq1}
 \end{multline}

The following proposition states that the algorithm of Equation~\eqref{algorithm1} produces a stable approximation of the solution of the inverse problem in \eqref{eq:inverse_problem1}.

\begin{pr}
 If Assumptions~\ref{ass:2} and \ref{ass:3} hold and the regularization parameters satisfy
\begin{equation}
\lim_{\delta\rightarrow 0}\alpha_j(\delta)=\delta~ \mbox{ and }~ \lim_{\delta\rightarrow 0}\frac{\delta^p}{\alpha_j(\delta)}=0, ~j=1,2,
\label{eq:alpha_delta}
\end{equation}
then, every sequence of solutions obtained by the algorithm of Equation~\eqref{algorithm1}, satisfying the discrepancy \eqref{eq:discrepancy}, when $\delta\searrow 0$, has a 
 $T_X$-convergent subsequence converging 
 to some $f_{x_0}$-minimizing solution of the Inverse Problem~\eqref{eq:inverse_problem1}, with $f_{x_0}$ as in Equation~\eqref{eq:fx0}. 
 \label{pr:convergence}
\end{pr}
\begin{proof}
Consider $\{\delta_k\}_{k\in\N}$, such that $\delta_k\searrow 0$, for each $k\in \N$, choose $\alpha_1 = \alpha_1(\delta_k)$ and $\alpha_2=\alpha_2(\delta_k)>0$ such that the discrepancy principle in Equation~\eqref{eq:discrepancy} and the estimates in Equation~\eqref{eq:alpha_delta} hold with data $y^{\delta_k}$ and noise level $\delta_k$. 
Consider also the sequence $\{(\overline w^{\delta_k}_{\alpha_1,\alpha_2},\overline z^{\delta_k}_{\alpha_1,\alpha_2})\}_{k\in\N}$ of the corresponding stationary points generated by the algorithm of Equation~\eqref{algorithm1}.
%

We need to show that this sequence has a $T_X$-convergent subsequence. Assume that the algorithm in Equation~\eqref{algorithm1} initializes always with the same pair $(w^0,z^0)$. So, for every $k$, 
$$\mathcal F(\overline w^{\delta_k}_{\alpha_1,\alpha_2},\overline z^{\delta_k}_{\alpha_1,\alpha_2};y^{\delta_k})\leq \lambda^p\delta^p.$$
In addition, by the relation in Equation~\eqref{eq:fx0} and the estimate in \eqref{eq:aux_ineq1}, 
\begin{equation}
 \limsup_{k\rightarrow\infty} \left[\beta_1 g_{w_0}(\overline w^{\delta_k}_{\alpha_1,\alpha_2}) + \beta_2 h_{w_0}(\overline z^{\delta_k}_{\alpha_1,\alpha_2})\right]\leq \beta_1 g_{w_0}(w^\dagger)  +\beta_2 h_{z_0}(z^\dagger).
\label{eq:estimate_gh}
\end{equation}
So, taking $\alpha^+ = \max_{k\in\N}\max\{\alpha_1(\delta_k),\alpha_2(\delta_k)\}$, and since
$$
\|F(\overline w^{\delta_k}_{\alpha_1,\alpha_2},\overline z^{\delta_k}_{\alpha_1,\alpha_2}) -  \tilde y\| \leq \|F(\overline w^{\delta_k}_{\alpha_1,\alpha_2},\overline z^{\delta_k}_{\alpha_1,\alpha_2}) -  y^{\delta_k}\| + \delta_k \leq (\lambda+1)\delta_k,
$$
it follows that,
\begin{multline*}
\limsup_{k\rightarrow \infty}\left[\|F(\overline w^{\delta_k}_{\alpha_1,\alpha_2},\overline z^{\delta_k}_{\alpha_1,\alpha_2}) -  \tilde y\|^p + \alpha^+g_{w_0}(\overline w^{\delta_k}_{\alpha_1,\alpha_2}) + \alpha^+h_{z_0}(\overline z^{\delta_k}_{\alpha_1,\alpha_2})\right]\\
\leq \alpha^+g_{w_0}(w^\dagger) + \alpha^+h_{z_0}(z^\dagger),
\end{multline*}
i.e., there exists a constant $K>0$ such that the sequence $\{(\overline w^{\delta_k}_{\alpha_1,\alpha_2},\overline z^{\delta_k}_{\alpha_1,\alpha_2})\}_{k\in\N}$ is in the level set $M_{\alpha^+}(K)$, which is pre-compact w.r.t. $T_X$. Hence, it has a $T_X$-convergent subsequence, which is denoted again by 
%
%
$\{(\overline w^{\delta_k}_{\alpha_1,\alpha_2},\overline z^{\delta_k}_{\alpha_1,\alpha_2})\}_{k\in\N}$, and converging to $(\tilde w,\tilde z)$, w.r.t. $T_X$. Since, for each  $k\in \N$, $\|F(\overline w^{\delta_k}_{\alpha_1,\alpha_2},\overline z^{\delta_k}_{\alpha_1,\alpha_2}) - y^{\delta_k}\|_Y \leq \lambda \delta_k$, by the weakly lower semi-continuity of $\phi$, 
\[
 \|F(\tilde w,\tilde z) - \tilde y\|^p \leq \displaystyle\liminf_{k\rightarrow 0}\|F(\overline w^{\delta_k}_{\alpha_1,\alpha_2},\overline z^{\delta_k}_{\alpha_1,\alpha_2}) - y^{\delta_k}\|^p
 \leq \lim_{k\rightarrow 0}\lambda \delta_k= 0.
\]
This means that $(\tilde w,\tilde z)$ is a solution of the Inverse Problem \eqref{eq:inverse_problem1}. Note that, by the estimate in Equation~\eqref{eq:estimate_gh},
$$\beta_1 g_{w_0}(\tilde w) + \beta_2 h_{z_0}(\tilde z)\leq \beta_1 g_{w_0}(w^\dagger) + \beta_2 h_{z_0}(z^\dagger).$$
So, $(\tilde w,\tilde z)$ is an $f_{x_0}$-minimizing solution.
\end{proof}

The following proposition states the convergence of {\em inexact solutions} to some solution of the inverse problem in Equation~\eqref{eq:inverse_problem1}. By inexact solution we mean the iterate $(w^{n+1},z^{n+1})$ satisfying the discrepancy in Equation~\eqref{eq:discrepancy}.

\begin{pr}
 Let the hypotheses of Proposition~\ref{pr:saddle_point} be satisfied. Assume further that the functionals $g_{w_0}(w)$ and $h_{z_0}(z)$ are uniformly bounded for $(w,z) \in \mathcal{D}$. Then, when $\delta\searrow 0$, every sequence of inexact solutions satisfying the discrepancy in Equation~\eqref{eq:discrepancy} has a $T_X$-convergent subsequence converging to a solution of the inverse problem in Equation~\eqref{eq:inverse_problem1}.
\end{pr}
\begin{proof}
As in the proof of Proposition~\eqref{pr:convergence}, let us consider $\{\delta_k\}_{k\in\N}$, such that $\delta_k\searrow 0$, for each $k\in \N$, choose $\alpha_1 = \alpha_1(\delta_k)$ and $\alpha_2=\alpha_2(\delta_k)>0$ such that the discrepancy principle in Equation~\eqref{eq:discrepancy} is satisfied and assume that $\max\{\alpha_1(\delta_k),\alpha_2(\delta_k)\}\leq \alpha^+$ for some finite constant $\alpha^+$. Consider also the iterates $(w^{n+1,\delta_k},z^{n+1,\delta_k})$ corresponding to $\alpha_1,\alpha_2$ and satisfying the discrepancy in Equation~\eqref{eq:discrepancy}.

Since $\|F(w^{n+1,\delta_k},z^{n+1,\delta_k}) - y^{\delta_k}\| \leq \lambda\delta_k$ and by Lemma~3.21 in \cite{schervar},
\begin{multline*}
 \mathcal F(w^{n+1,\delta_k},z^{n+1,\delta_k};\tilde y,\alpha^+)\leq 2^{p-1}\mathcal F(w^{n+1,\delta_k},z^{n+1,\delta_k};y^{\delta_k},\alpha^+) + 2^{p-1}\delta_k^p\\
 \leq 2^{p-1}(\lambda^p+1)\delta_k^p + \alpha^+g_{w_0}(w^{n+1,\delta_k}) + \alpha^+h_{z_0}(z^{n+1,\delta_k}).
\end{multline*}
Since $g_{w_0}(w)$ and $h_{z_0}(z)$ are uniformly bounded for $(w,z) \in \mathcal{D}$, there exists some constant $K>0$ such that $(w^{n+1,\delta_k},z^{n+1,\delta_k})\in M_{\alpha^+}(K)$, where $y^\delta$ is replaced by $\tilde y$ in the Tikhonov-type functional. Since $M_{\alpha^+}(K)$ is pre-compact w.r.t. $T_X$, the sequence of iterates $\{(w^{n+1,\delta_k},z^{n+1,\delta_k})\}_{k\in\N}$ has a $T_X$-convergent subsequence, which is also denoted by $\{(w^{n+1,\delta_k},z^{n+1,\delta_k})\}_{k\in\N}$ and converges to $(\tilde w,\tilde z)$ w.r.t. $T_X$. So, by the $T_X$- continuity of $F$ and the norm of $Y$, it follows that
$$
\|F(\tilde w,\tilde z)-\tilde y\| \leq \liminf_{k\rightarrow \infty}\|F(w^{n+1,\delta_k},z^{n+1,\delta_k})-y^\delta_k\| \leq \lim_{k\rightarrow \infty}\lambda\delta_k = 0,
$$
and the assertion follows.
\end{proof}

\begin{rem}
 It is not difficult to prove that, under the hypotheses of Proposition~\ref{pr:convergence} 
 there exists a 
 sequence of finite iterates or inexact solutions that converges w.r.t. $T_X$ to some $f_{x_0}$-minimizing solution of the inverse problem in Equation~\eqref{eq:inverse_problem1}, when $\delta\searrow 0$. Let us consider $\{\delta_k\}_{k\in\N}$, such that $\delta_k\searrow 0$ and assume that for each $k\in\N$, the solution $(\overline w^{\delta_k}_{\alpha_1,\alpha_2},\overline z^{\delta_k}_{\alpha_1,\alpha_2})$ provided by Algorithm~\ref{algorithm1} satisfies the discrepancy in Equation~\eqref{eq:discrepancy}. Also, for each $k\in\N$, find a subsequence of iterates converging w.r.t. $T_X$ to $(\overline w^{\delta_k}_{\alpha_1,\alpha_2},\overline z^{\delta_k}_{\alpha_1,\alpha_2})$ and select one iterate that also satisfies the discrepancy and is close to $(\overline w^{\delta_k}_{\alpha_1,\alpha_2},\overline z^{\delta_k}_{\alpha_1,\alpha_2})$ w.r.t. $T_X$, and gets arbitrarily closer as $k$ increases. By Proposition~\ref{pr:convergence}, the sequence $\{(\overline w^{\delta_k}_{\alpha_1,\alpha_2},\overline z^{\delta_k}_{\alpha_1,\alpha_2})\}$ has a $T_X$-convergent subsequence, converging to $(\tilde w,\tilde z)$, a $f_{x_0}$-minimizing solution of the inverse problem in Equation~\eqref{eq:inverse_problem1}. It is easy to see that the corresponding subsequence of iterates also converges to $(\tilde w,\tilde z)$ w.r.t. $T_X$.
\end{rem}

\section{The Calibration}\label{sec:calibration} 
This section is devoted to the solution of the calibration from quoted European vanilla option prices of the local volatility surface $\left\{ a(\tau,y) | (\tau,y)\in D \right\}$ 
and the double exponential tail $\varphi$, by the splitting technique presented in Section~\ref{sec:splitting}. From the double-exponential tail, we estimate the jump-size distribution $\nu$. 

\subsection{Calibration of Local Volatility Surface and Double Exponential Tail}\label{sec:calibexptail}
The inverse problem can be stated as: {\em 
Given a set of European call option prices $\tu$, such that, $\tu - u(a_0,0)$ is in the range of $F$, $\mathcal R(F)$, find $(a^\dagger,\varphi^\dagger)$ in $\mathcal D(F)$ satisfying the equation}
\begin{equation}
 \tu = u(a^\dagger,\varphi^\dagger),
 \label{eq:inverseproblem}
\end{equation}
{\em where $u$ is the solution of the PIDE problem in \eqref{eq:pide2}-\eqref{eq:pide2b}, using the integral representation \eqref{eq:integral_part}.} 

In practice, it is only possible to observe noisy option data given in a sparse mesh of strikes. Such data is denoted by $u^\delta$, where
\begin{equation}
 \|\tu - u^\delta\| \leq \delta,
\end{equation}
and $\delta>0$ is the noise level.

To use the results from Section~\ref{sec:splitting} in this context, we introduce the following notation:
\begin{multline*}
w := \tilde a = a-a_0,~ z := (\varphi_-,\varphi_+),~w_0 := a_0, z_0 := 0,\\
\tilde y := \tu - u(a_0,0),~\mbox{and}~ y^\delta := u^\delta - Pu(a_0,0),
\end{multline*}
where $P$ projects the solution of \eqref{eq:pide2}-\eqref{eq:pide2b} onto the sparse mesh where $u^\delta$ is given. 
Since $X = W\times Z$, 
$$
W := H^{1+\varepsilon}(D),\quad\mbox{and}\quad Z := W^{2,1}(-\infty,0)\times W^{2,1}(0,+\infty). 
$$

Let $T_X$ and $T_Y$ be the weak topologies of $X$ and $Y=W^{1,2}_2(D)$, respectively. By assuming that $g_{a_0} = g_{w_0}$ and $h_{\varphi_0} = h_{z_0}$ are convex, proper and weakly lower semi-continuous functionals, Propositions~\ref{prop:continuity}-\ref{prop:frechet} imply that  Assumptions~\ref{ass:2}-\ref{ass:3} hold true. Note that, the tangential cone condition in Assumption~\ref{ass:3} is an easy consequence of the Inequality~\eqref{eq:tangential_pide} in Proposition~\ref{prop:frechet}. Hence, given the data $u^\delta$, the splitting algorithm applied to the simultaneous calibration of $a$ and $\varphi$ converges to some approximation of the true solution of the inverse problem \eqref{eq:inverseproblem}, if the latter exists. 

Since the inclusion of $W^{1,2}_2(D)$ into $L^2(D)$ is continuous, the existence and stability of solutions given by the splitting algorithm as well as its convergence to the true solution also hold whenever $Y=W^{1,2}_2(D)$ is replaced $Y=L^2(D)$ in the Tikhonov regularization functional in \eqref{eq:tikhonov2} and \eqref{eq:data_misfit}.

A possible choice of the penalization term to fulfill the weak pre-compactness of the level sets of the Tikhonov functional \eqref{eq:tikhonov2} is $g_{a_0}(a) = \|a-a_0\|^2_{H^{1+\varepsilon}(D)}$ for the variable $a$ and for the variable $\varphi$, $h_{\varphi_0}(\varphi)$ is 
\begin{multline*}
h_{\varphi_0}(\varphi) = KL(\varphi_+|\varphi_{+,0}) + KL(\varphi^\prime_+|\varphi^\prime_{+,0}) + KL(\varphi^{\prime\prime}_+|\varphi^{\prime\prime}_{+,0})\\ + KL(\varphi_-|\varphi_{-,0}) + KL(\varphi^\prime_-|\varphi^\prime_{-,0}) + KL(\varphi^{\prime\prime}_-|\varphi^{\prime\prime}_{-,0})
\end{multline*}
where the $KL$ stands for the Kullback-Leibler divergence
$$
 KL(\varphi_+|\varphi_{+,0}) = \displaystyle\int_{0}^{+\infty}\left[\varphi_+\ln\left(\frac{\varphi_+}{\varphi_{+,0}}\right)+ (\varphi_{+,0}-\varphi_+)\right]dx,
$$
with $\varphi_0> 0$ given. 
In this case, $g_{a_0}$ and $h_{\varphi_0}$ are convex, weakly continuous and coercive. In addition, the level sets of the Kullback-Leibler divergence
$$
\{\varphi \in L^1(\R) ~:~ KL(\varphi|\varphi_0) \leq C\}
$$
are weakly pre-compact in $ L^1(\R)$. See Lemma 3.4 in \cite{resa}.

\subsection{Calibration of Jump-Size Distribution from Double Exponential Tail}

One possible way, but not recommended, to obtain the jump-size distribution $\nu$ is by differentiating once the double exponential tail $\varphi$, since, $\nu$ is such that $\varphi_{+}:=\varphi|_{(0,+\infty)} \in W^{2,1}(0,+\infty)$ and $\varphi_{-}:=\varphi|_{(-\infty,0)}\in W^{2,1}(-\infty,0)$. By Sobolev's embedding (see Theorem~4.12 in \cite{AdaFou2003}), $\varphi^\prime_{\pm}$ are continuous functions and
$$
\varphi^\prime(z) = \left\{
\begin{array}{ll}
\text{e}^z\displaystyle\int_{-\infty}^{z}\nu(dx) = \text{e}^z\nu((-\infty,z]),& z<0\\
-\text{e}^z\displaystyle\int^{+\infty}_{z}\nu(dx)= -\text{e}^z\nu([z,+\infty)),& z>0.
\end{array}
\right.
$$
So, $z\mapsto \nu((-\infty,z])$ and $z\mapsto \nu([z,+\infty))$ are continuous functions.

By Proposition~5.2 in \cite{KinMay2011}, $\nu$ can be represented as
$$
\nu(dx) = \left\{
\begin{array}{ll}
\displaystyle\frac{1+x^2}{x^2}\mu_-(dx),& x<0\\
\displaystyle\frac{1+x^2}{x^2(1+x\text{e}^x)}\mu_+(dx),& x>0,
\end{array}
\right.
$$
where $\mu_+$ and $\mu_-$ are finite measures, defined in $(0,+\infty)$ and $(-\infty,0)$, respectively. This implies that, $z\mapsto \mu_-((-\infty,z])$ and $z\mapsto \mu_+([z,+\infty))$ are continuous functions, which implies that they are absolutely continuous with respect to the Lebesgue measure. See Lemma~III.4.13 in \cite{DunSch1958}. So, there exist integrable functions $h_{\pm}$, such that
$$
h_-(x)dx = \mu_-(dx) \quad \mbox{ and }\quad h_+(x)dx = \mu_+(dx).
$$
Define $h \in L^1(\R)$, such that $h|_{(-\infty,0)} = h_-$ and $h|_{(0,+\infty)} = h_+$.

\begin{lem}
 The map $h \in L^1(\R) \longmapsto \varphi \in L^2(\R)$ is compact.
\end{lem}
\begin{proof}
Let the sequence $\{(h_{-,n},h_{+,n})\}_{n\in\N}$ converge weakly to some $(h_{-},h_{+})$ in $ L^1(-\infty,0)\times L^1(0,+\infty)$. Define
$$
\varphi_n(z) = \left\{
\begin{array}{ll}
\displaystyle\int_{-\infty}^z(\text{e}^z-\text{e}^x)\frac{1+x^2}{x^2}h_{-,n}(x)dx,& z<0\\
\\
\displaystyle\int^{+\infty}_z(\text{e}^x-\text{e}^z)\frac{1+x^2}{x^2(1+x\text{e}^x)}h_{+,n}(x)dx,& z>0,
\end{array}
\right.
$$
for each $n\in \N$ and $\varphi = \varphi(\mu_{-},\mu_{+})$ in the same way. It is easy to see that $\varphi_+ = \varphi|_{(0,+\infty)} \in W^{1,1}(0,+\infty)$ and $\varphi_- = \varphi|_{(-\infty,0)} \in W^{1,1}(-\infty,0)$. Note also that, by hypothesis, $\varphi_{+,n} \in W^{2,1}(0,+\infty)$ and $\varphi_{-,n}\in W^{2,1}(-\infty,0)$. So, by Sobolev's embedding (see Theorem~4.12 in \cite{AdaFou2003}), $\varphi_{+,n}, \varphi_+ \in L^2(0,+\infty)$ and $\varphi_{-,n},\varphi_-\in L^2(-\infty,0)$. 

The estimate
$$
|\varphi_{-,n}(z) - \varphi_-(z)|^2 = \left|\displaystyle\int^z_{-\infty}(\text{e}^z-\text{e}^x)\frac{1+x^2}{x^2}(h_{-,n}-h_-)(x)dx\right|^2\rightarrow 0,
$$
holds almost everywhere in $(-\infty,0)$, since $(\text{e}^z-\text{e}^x)\frac{1+x^2}{x^2}$ is in $L^\infty(0,+\infty)$. Similarly, $|\varphi_{+,n}(z) - \varphi_+(z)|\rightarrow 0$ almost everywhere. By the monotone convergence theorem the assertion follows.
\end{proof}
Since the map that associates $h$ to $\varphi$ is compact, it follows that the corresponding inverse problem is ill-posed. So, the procedure of obtaining $h$ by differentiating $\varphi$ is not stable.

The inverse problem of finding the jump-size distribution from the double exponential tail is: {\em Given the output of the splitting algorithm $\tilde \varphi \in L^2(\R)$, find $(h_-,h_+) \in \mathcal{V}_-\times \mathcal{V}_+$ satisfying}
\begin{equation}
 \varphi(h_-,h_+) = \tilde\varphi,
 \label{eq:ip_phi}
\end{equation}
where
 $\mathcal{V}_-\times \mathcal{V}_+$ is some subset of $L^1(-\infty,0)\times L^1(0,+\infty)$.

If we apply Tikhonov-type regularization to such inverse problem, it can be rewritten as: find $(h_-,h_+) \in \mathcal{V}_-\times \mathcal{V}_+$ minimizing
$$
\mathcal G(h_-,h_+) = \| \varphi(h_-,h_+) - \tilde\varphi\|^2_{L^2(\R)} + \alpha f_{h_{-,0},h_{+,0}}(h_-,h_+),
$$
with $(h_{-,0},h_{+,0})$ in $L^1(-\infty,0)\times L^1(0,+\infty)$ given.

Let us assume that $ f_{\nu_{-,0},\nu_{+,0}}$ is weakly lower-semi-continuous and convex. If the level sets of 
$\mathcal G(h_-,h_+) $ are weakly compact in $\mathcal{V}_-\times \mathcal{V}_+$ (or $\mathcal{V}_-\times \mathcal{V}_+$ is weakly compact), then, as in Section~\ref{sec:tikhonov}, there exists stable minimizers of $\mathcal G(h_-,h_+)$ in $\mathcal{V}_-\times \mathcal{V}_+$.

Summing up, the Tikhonov-type regularization provides a stable approximation for the jump-size distribution $\nu$.
\subsection{Gradient Evaluation}\label{sec:gradient}
To implement a numerical gradient descent algorithm to minimize the Tikhonov-type functional with respect to each variable, as in \citet{AlbAscZub2016}, it is necessary to evaluate the directional derivatives of a numerical approximation of the data misfit function $\phi = \phi(a,\varphi)$. If $a^k$ and $\varphi^k$ denote the iterates of $a$ and $\varphi$ respectively in the gradient descent algorithm, evaluate
\begin{eqnarray}
 a^k = a^{k-1} + \theta_k \frac{\partial}{\partial a} \mathcal F(a^{k-1},\tilde\varphi)\\
 \varphi^k = \varphi^{k-1} + \beta_k \frac{\partial}{\partial \varphi} \mathcal F(\tilde a,\varphi^{k-1}),
\end{eqnarray}
until some tolerance is reached, with $\tilde a$ and $\tilde\varphi$ fixed. To perform this task, we shall present the evaluation of such derivatives in the continuous setting. 

Since
$$
\frac{\partial}{\partial a} \mathcal F(a,\varphi) = \frac{\partial}{\partial a} \phi(a,\varphi) + \alpha_1 \partial g_{a_0}(a) ~\text{  and  }~
\frac{\partial}{\partial \varphi} \mathcal F(a,\varphi) = \frac{\partial}{\partial \varphi} \phi(a,\varphi) + \alpha_2 \partial h_{\varphi_0}(\varphi),
$$
to evaluate $\frac{\partial}{\partial a} \phi$ and $\frac{\partial}{\partial \varphi} \phi$, recall that the directional derivative of $F$ at $(a,\varphi)$ in the direction $(h,\gamma)$, with $(a+h,\varphi+\gamma) \in \mathcal{D}(F)$, is denoted by $v$ and is the unique solution of the PIDE
\begin{multline}
  v_\tau(\tau,y) - a(\tau,y) \left(v_{yy}(\tau,y) - v_{y}(\tau,y)\right)  + rv_y(\tau,y)
 -I_\varphi(v_{yy} - v_y)(\tau,y)\\ =
 h\left(u_{yy}(\tau,y) - u_y(\tau,y)\right) 
+ I_{\gamma}\left(u_{yy} - u_y\right)(\tau,y),
 \label{eq:gradient1}
\end{multline}
with homogeneous boundary and initial conditions, where $u$ is the solution of the PIDE problem \eqref{eq:pide2}-\eqref{eq:pide2b}.

Note that,  $\frac{\partial}{\partial a} F(a,\varphi) h = v(h,\varphi+0)$, and $\frac{\partial}{\partial a} \phi(a,\varphi) = \frac{\partial}{\partial a} F(a,\varphi)^*P^*(Pu(a,\varphi)-u^\delta)$, so, for every $h\in Z$, 
$$
 \left\langle \frac{\partial}{\partial a} \phi,h\right\rangle = \left\langle \frac{\partial}{\partial a} F(a,\varphi)^*P^*(Pu(a,\varphi) - u^\delta),h\right\rangle = \left\langle Pu(a,\varphi) - u^\delta, P\frac{\partial}{\partial a} F(a,\varphi)h\right\rangle.
$$
Since $P\frac{\partial}{\partial a} F(a,\varphi)h = P\mathcal{L}M_{u_{yy}-u_y}h$, where $M_{u_{yy}-u_y}$ is the multiplication by $u_{yy}-u_y$ operator and $\mathcal{L}$ is the operator that maps the source $h(u_{yy} - u_y)$ onto the solution of the PIDE \eqref{eq:gradient1} with homogeneous boundary and initial conditions (with $\gamma = 0$), it follows that 
$$
\left\langle Pu(a,\varphi) - u^\delta, P\frac{\partial}{\partial a} F(a,\varphi)h\right\rangle = \langle M_{u_{yy}-u_y} \mathcal L^*P^*(Pu(a,\varphi) - u^\delta),h \rangle = \langle  (u_{yy}-u_y)w,h\rangle,
$$
where $w$ is the solution of the adjoint PIDE:
\begin{multline}
 w_\tau(\tau,y) + \left(a w\right)_{yy}(\tau,y) + (aw)_y(\tau,y) - rw_y(\tau,y) = \\ 
 \int_{\R}\varphi(x)\left(w_{yy}(\tau,x+y)+ w_y(\tau,x+y)\right)dx + P^*Pu(a,\nu) - P^*u^\delta,
\label{eq:adjoint}
\end{multline}
with homogeneous boundary and terminal conditions. 

In a similar way, we evaluate $\frac{\partial}{\partial \varphi}\phi$ and find
$$
\frac{\partial}{\partial \varphi}\phi(z) = [H^*_u w] (z) := \int_0^T\int_{\R}\left(u_{yy}(\tau,y-z) - u_y(\tau,y-z)\right)w(\tau,y)dyd\tau,
$$
where $u$ is the solution of the PIDE problem \eqref{eq:pide2}-\eqref{eq:pide2b}.

\section{A Numerical Scheme}\label{sec:numerics} 
Differently from \cite{ConVol2005a}, we consider directly the case where the activity of jumps can be infinite. This is because we use the representation \eqref{eq:integral_part} for the integral term in the PIDE problem \eqref{eq:pide2}-\eqref{eq:pide2b}.

Firstly, let us restrict the log-moneyness range where the PIDE problem \eqref{eq:pide2}-\eqref{eq:pide2b} is defined to $[y_{\min},y_{\max}]$, with $y_{\min} < 0 < y_{\max}$, and then, $D = [0,\tau_{\max}]\times[y_{\min},y_{\max}]$. Outside $[y_{\min},y_{\max}]$, the numerical solution assumes the value of the payoff function at these points.
 
Let $I,J \in \N$ be fixed. We consider the discretization $\tau_i = i \Delta \tau$, with $i = 0,1,2,...,I$, and $y_j = j \Delta y$, with $j = -J,-J+1,...,0,1,...,J$. Denote by $u^{i}_{j}:= u(\tau_i,y_j)$, $a^i_j := a(\tau_i,y_j)$,  $\beta:= \Delta \tau/\Delta y$ and $\eta = \Delta \tau/\Delta y^2$. Define also:
$$
\varphi_j = \left\{
\begin{array}{ll}
\displaystyle\int_{y_{\min}-\frac{\Delta y}{2}}^{y_j}(\text{e}^{y_j}-\text{e}^x)\nu(dx), & y_j < 0\\ 
\\
\displaystyle\int_{y_j}^{y_{\max}+\frac{\Delta y}{2}}(\text{e}^x-\text{e}^{y_j})\nu(dx), & y_j > 0,
\end{array}
\right.
$$
where these integrals are approximated by the trapezoidal rule.

The differential part of the PIDE problem \eqref{eq:pide2}-\eqref{eq:pide2b} is approximated by the Crank-Nicolson scheme and the integral operator by the trapezoidal rule, leading to:
\begin{multline}
u^{i}_{j} - \displaystyle\frac{1}{2}\eta a^{i}_{j}(u^{i}_{j+1} - 2u^{i}_{j} +u^{i}_{j-1}) + \frac{1}{4}\beta a^{i}_{j}(u^{i}_{j+1} - u^{i}_{j-1}) =\\ 
u^{i-1}_{j} + \displaystyle\frac{1}{2}\eta a^{i-1}_{j}(u^{i-1}_{j+1} - 2u^{i-1}_{j} +u^{i-1}_{j-1}) - \frac{1}{4}\beta a^{i-1}_{j}(u^{i-1}_{j+1} - u^{i-1}_{j-1}) + M^{i-1}_j,
\label{cns}
\end{multline}
where
$$
M^{i-1}_j = \sum_{k=-J}^{J}\varphi_k\text{e}^{y_k}\left[\beta(u^{i-1}_{j+1-k} - 2u^{i-1}_{j-k} +u^{i-1}_{j-1-k})-\frac{1}{2}\Delta \tau(u^{i-1}_{j+1-k} - u^{i-1}_{j-1-k})\right].
$$
In \citet{KinMayAlbEng2008} a Crank-Nicholson-type algorithm was also used to solve the so-called direct problem. There, the authors were interested in the calibration of the local speed function, which here is set constant and equal to $1$.

The numerical scheme for solving the adjoint PIDE \eqref{eq:adjoint} with homogeneous boundary and terminal conditions is quite similar to the one in Equation~\eqref{cns}. Following the same ideas presented in Section~ \ref{sec:gradient} we find the discrete version of the gradients of the data misfit function $\phi$.

\subsection{Numerical Validation}\label{sec:validation}

The purpose of this example is to illustrate the accuracy of the scheme in \eqref{cns} by comparing it with other techniques. 

Assume that $S_0=1$, $y_{\max} = 5$, $y_{\min} = -5$, $\tau_{\max} = 1$, $\Delta \tau = 0.005$, $\Delta y = 0.025$, $r = 0$ and the local volatility surface is constant with $a \equiv 0.0113$. Then, we evaluate European call prices in three different ways, the scheme of Section~\ref{sec:numerics}, the implicit-explicit scheme from \citet{ConVol2005a}, and the Fourier transform method from \citet{TanVol2009}, which is based on the pricing formula presented in \citet{CarMad1999}.

In the following synthetic examples the measure $\nu$ is assumed to be absolutely continuous w.r.t. the Lebesgue measure and given by
\begin{equation}
 \nu(dx) = \displaystyle\frac{0.1}{\sqrt{2\pi}}\text{e}^{-\frac{x^2}{2}}dx.
 \label{jumpsize}
\end{equation}

We also consider a functional local volatility surface, given by
\begin{equation}
\sigma(\tau,y) = \left\{
\begin{array}{ll}
\displaystyle\frac{2}{5}-\frac{4}{25}\text{e}^{-\tau/2}\cos\left(\displaystyle\frac{4\pi y}{5}\right),& \text{ if } -2/5 \leq y \leq 2/5\\
\\
2/5,& \text{ otherwise.}
\end{array} \right.
\label{vol}
\end{equation}
and compare the results given by the scheme \eqref{cns} with the one presented in \citet{ConVol2005a}.

To measure the accuracy, we consider implied volatilities instead of prices. Let us denote by:  
\begin{itemize}
 \item $\Sigma_{CN}$ the set of implied volatilities corresponding to the prices evaluated with the schemes from Equation~\eqref{cns}.
 \item $\Sigma_{CV}$ the set of implied volatilities corresponding to the prices evaluated with the schemes from \citet{ConVol2005a}.
 \item $\Sigma_{Fourier}$ the set of implied volatilities corresponding to the prices evaluated with the schemes from \citet{TanVol2009}.
\end{itemize}
We estimate the normalized $\ell_2-$distance between them as follows:
$$
\|\Sigma_{CN} - \Sigma_{CV}\|/\|\Sigma_{CV}\| ~\text{ or }~ \|\Sigma_{CN} - \Sigma_{Fourier}\|/\|\Sigma_{Fourier}\|,
$$
We also estimate the mean and standard deviation of the absolute relative error (abs. rel. error), which is evaluated at each node as follows:
$$
|\Sigma_{CN}(\tau_i,y_j) - \Sigma_{CV}(\tau_i,y_j)|/|\Sigma_{CV}(\tau_i,y_j)| ~\text{ or }~  |\Sigma_{CN}(\tau_i,y_j) - \Sigma_{Fourier}(\tau_i,y_j)|/|\Sigma_{Fourier}(\tau_i,y_j)|.
$$
Such results can be seen in Table~\ref{tab:example1}. A comparison between implied volatilities with constant and non-constant local volatility surface can be found in Figures~\ref{fig:impvol1a} and \ref{fig:impvol1b}, respectively.
\begin{table}[!ht]
\centering
\begin{tabular}{c c|c|c|c|}
\cline{3-5}
 & & \multirow{2}{*}{N.distance}&\multicolumn{2}{c|}{Abs. Rel. Error}\\
 \cline{4-5}
 & & & Mean & Std. Dev.\\
\hline
\multicolumn{1}{|c}{\multirow{2}{*}{$a\equiv 0.0113$}}&\multicolumn{1}{|c|}{CV} & 0.0064 & 0.0070 & 0.0072\\
\multicolumn{1}{|c}{}&\multicolumn{1}{|c|}{Fourier} & 0.0862 & 0.0923 & 0.0699\\
\hline
\multicolumn{1}{|c|}{Non-constant $a$}&\multicolumn{1}{|c|}{CV} & 0.0064 & 0.0064 & 0.0038\\
\hline
\end{tabular}
\caption{Normalized distance and absolute relative error.}
\label{tab:example1}
\end{table}

\begin{figure}[!ht]
  \centering
      \includegraphics[width=0.249\textwidth]{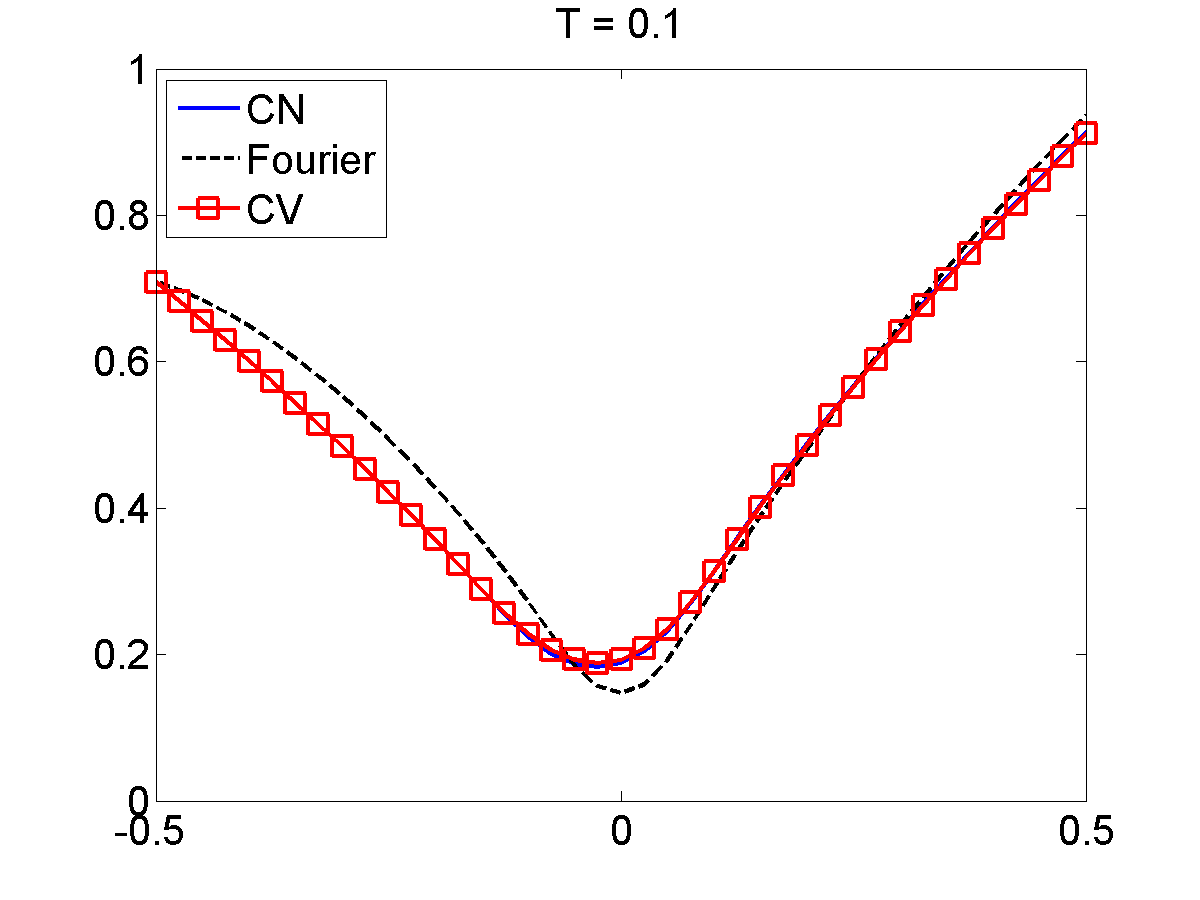}\hfill
      \includegraphics[width=0.249\textwidth]{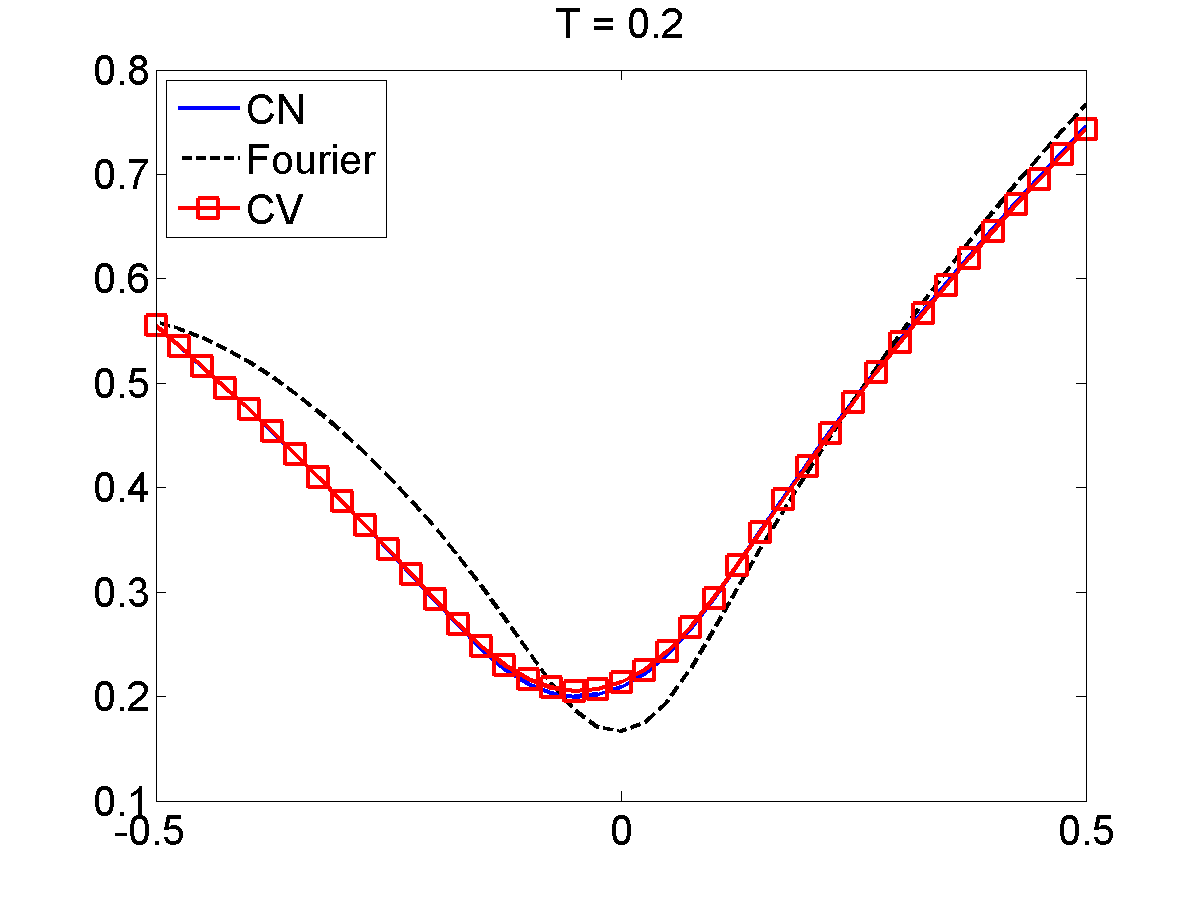}\hfill
      \includegraphics[width=0.249\textwidth]{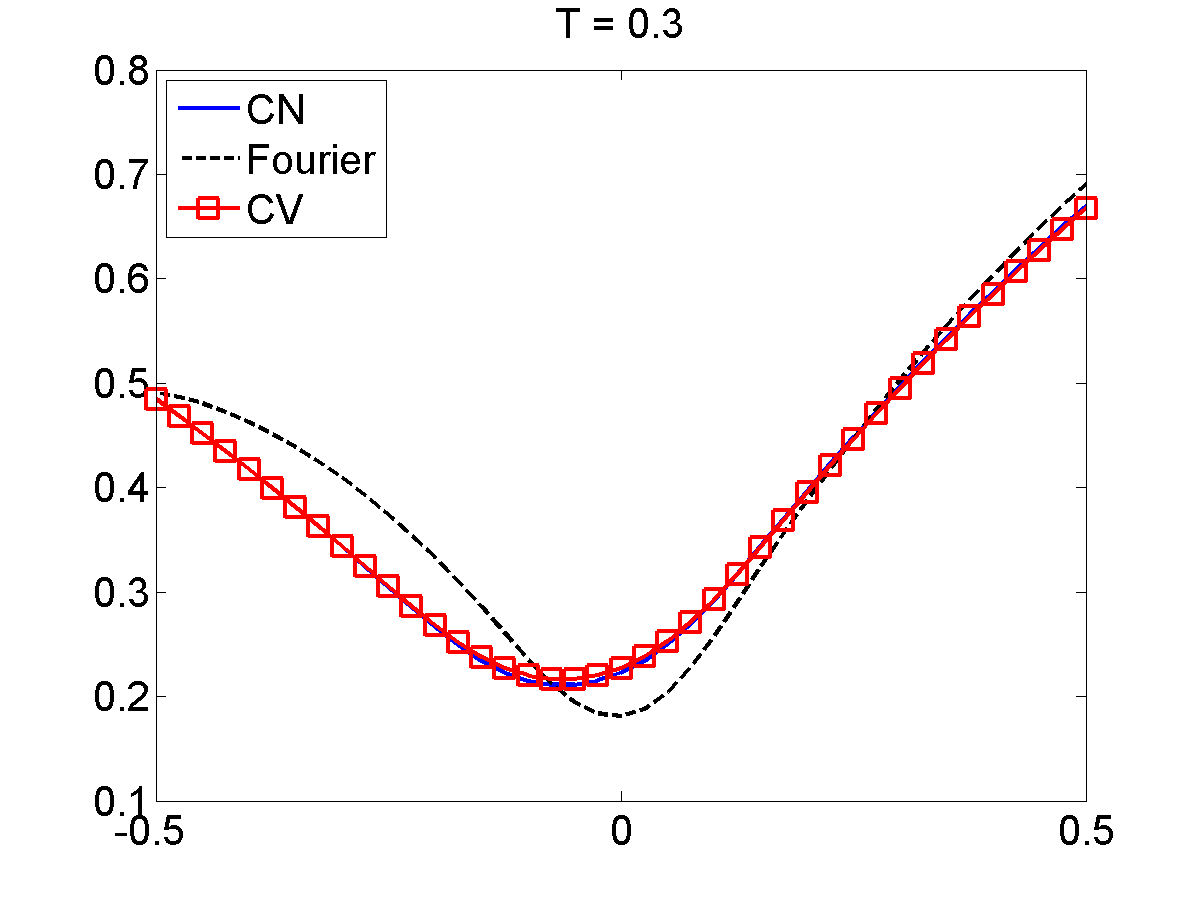}\hfill
      \includegraphics[width=0.249\textwidth]{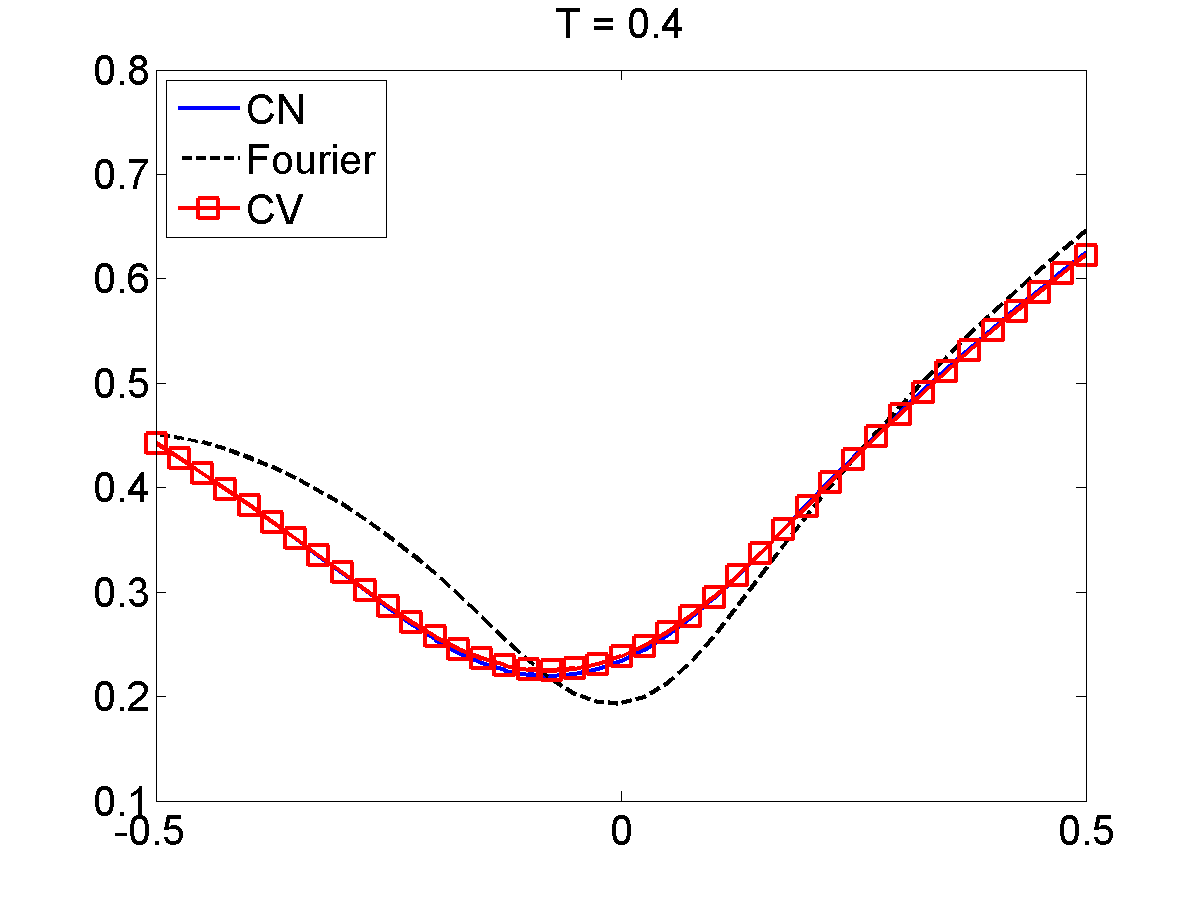}\hfill
      \includegraphics[width=0.249\textwidth]{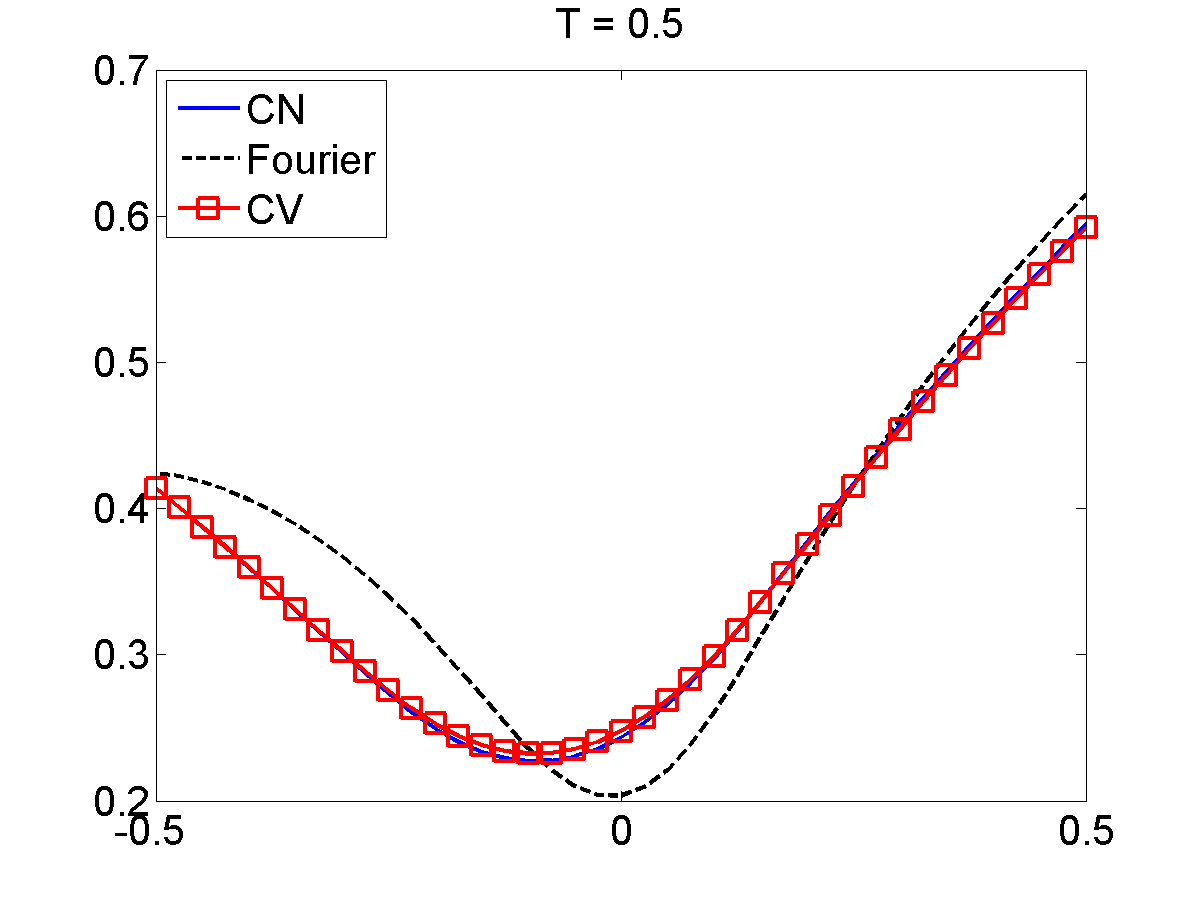}\hfill
      \includegraphics[width=0.249\textwidth]{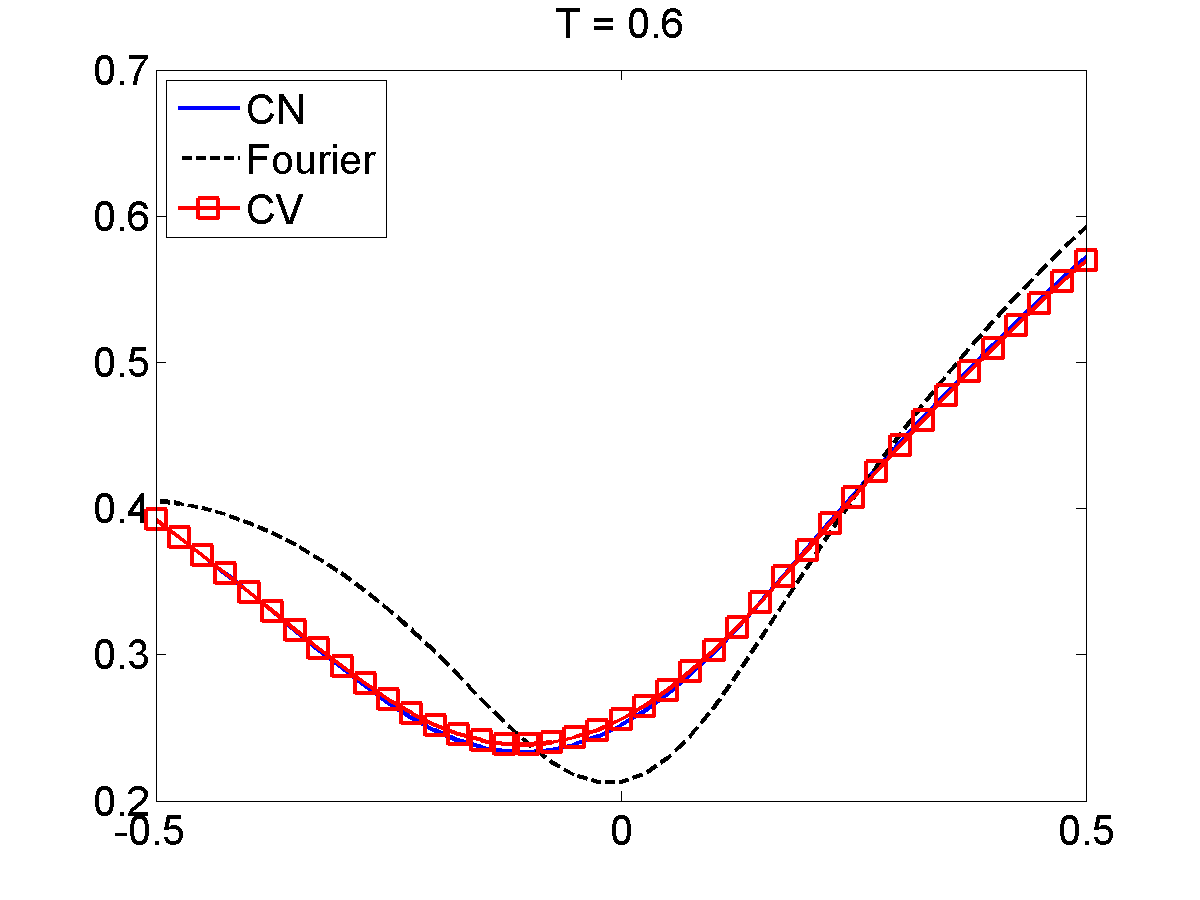}\hfill
      \includegraphics[width=0.249\textwidth]{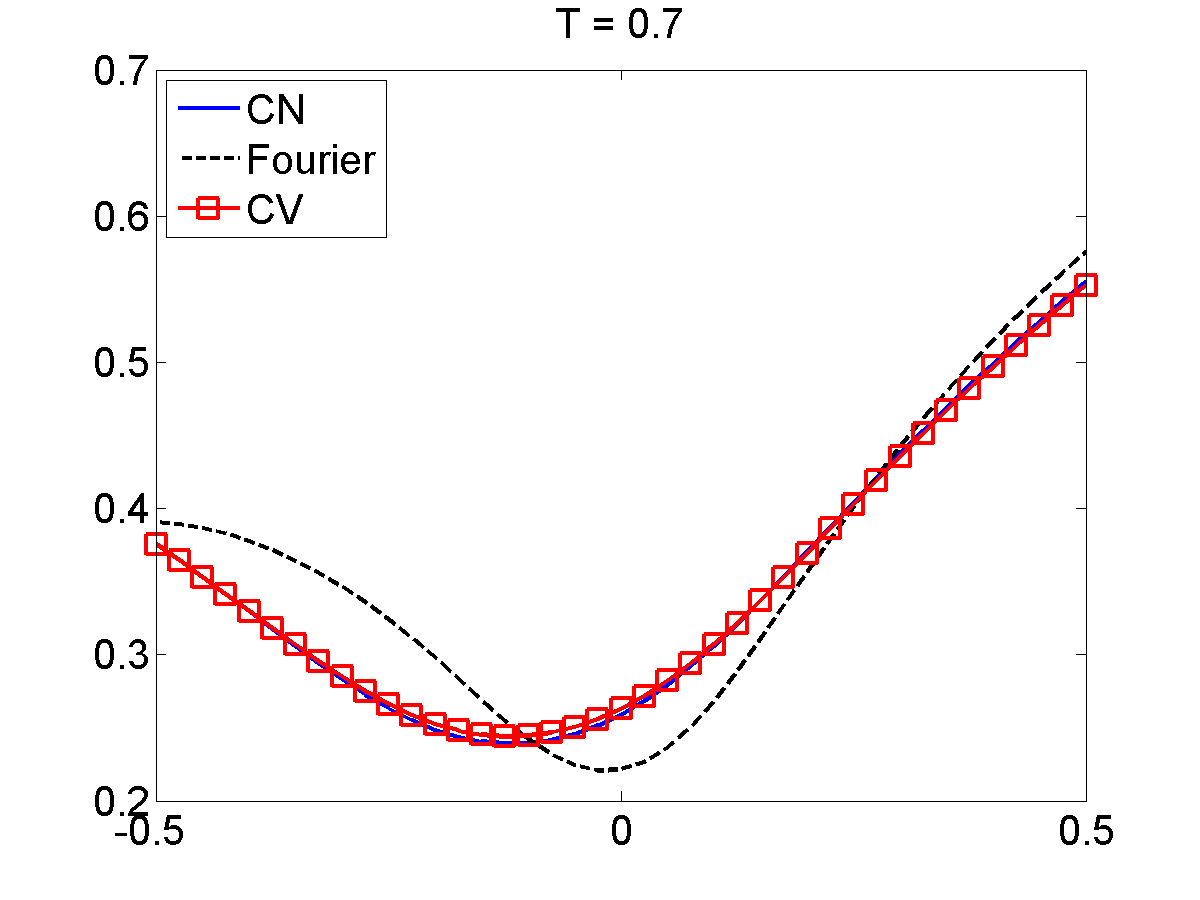}\hfill
      \includegraphics[width=0.249\textwidth]{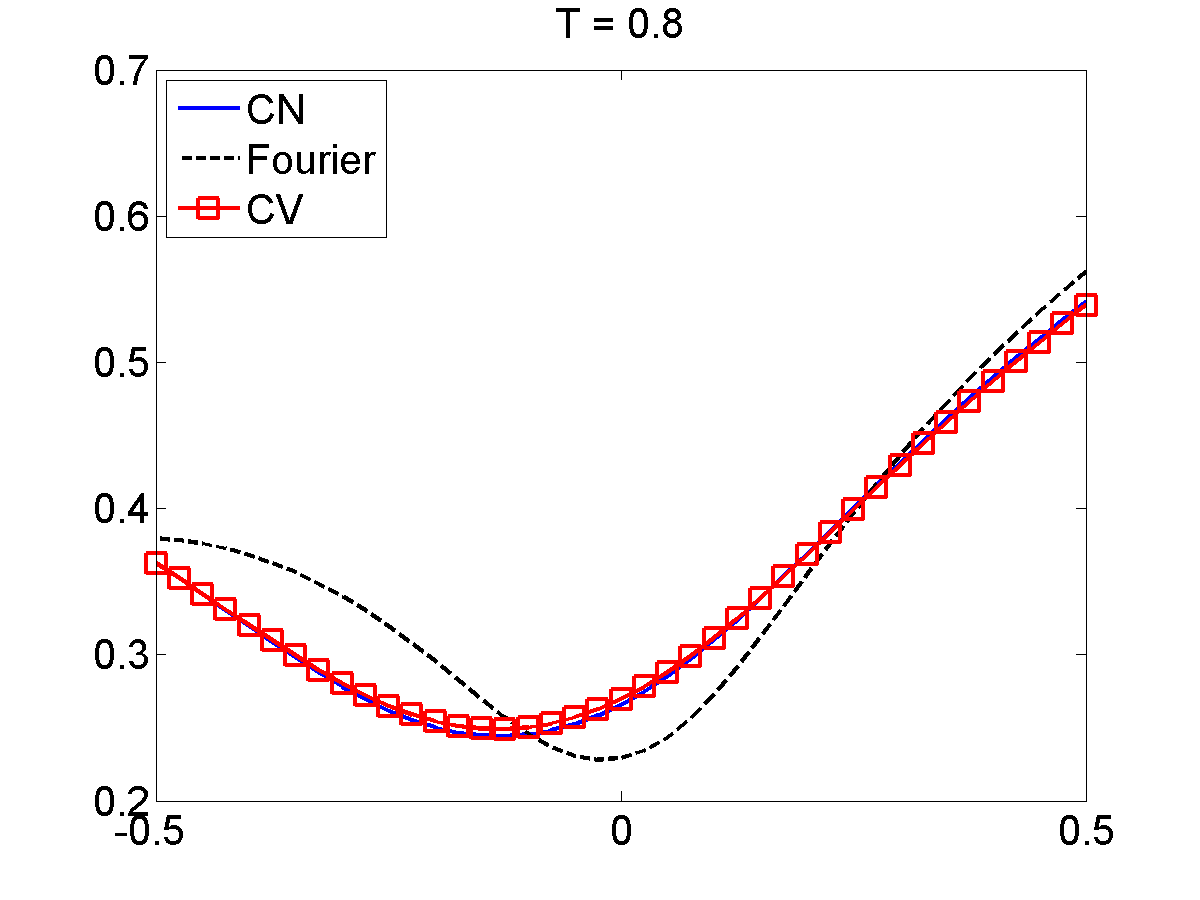}\hfill
      \includegraphics[width=0.249\textwidth]{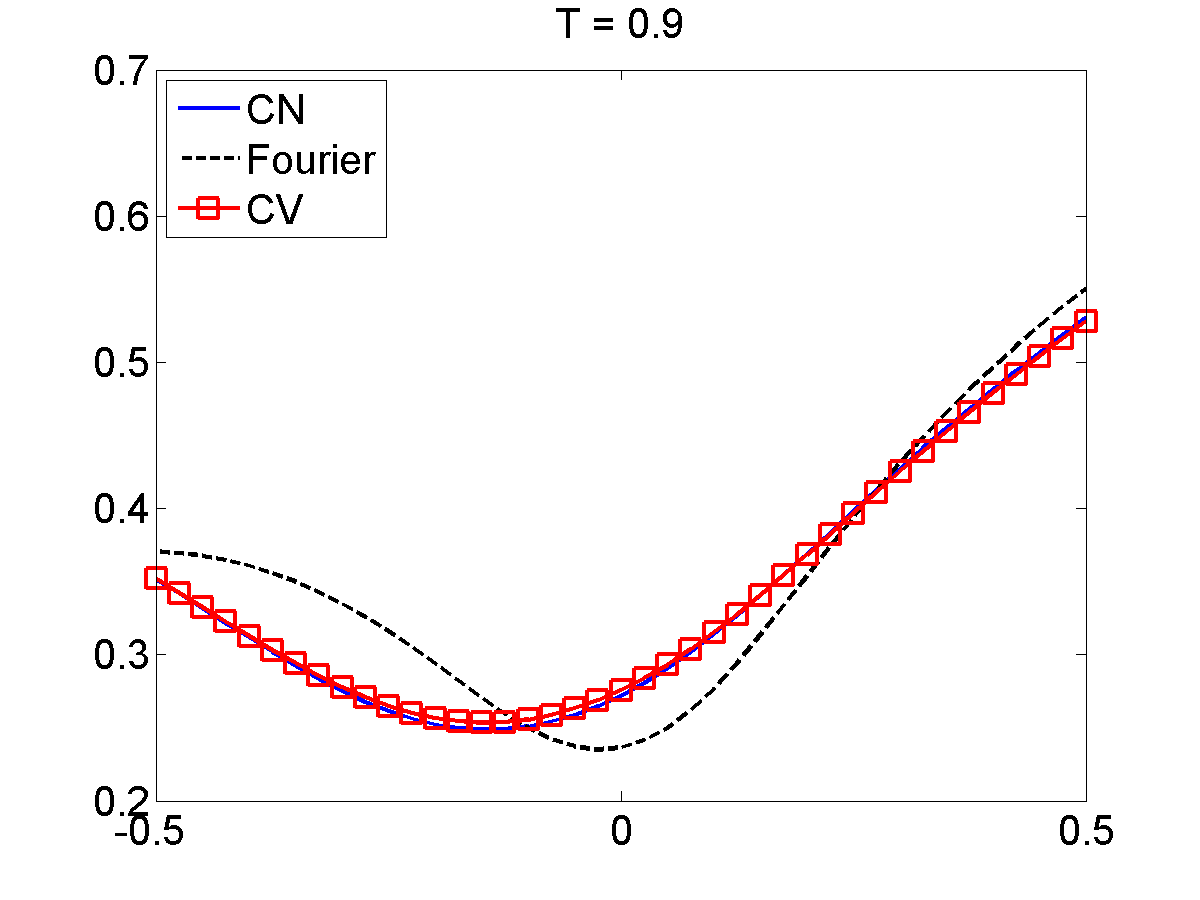}\hfill
      \includegraphics[width=0.249\textwidth]{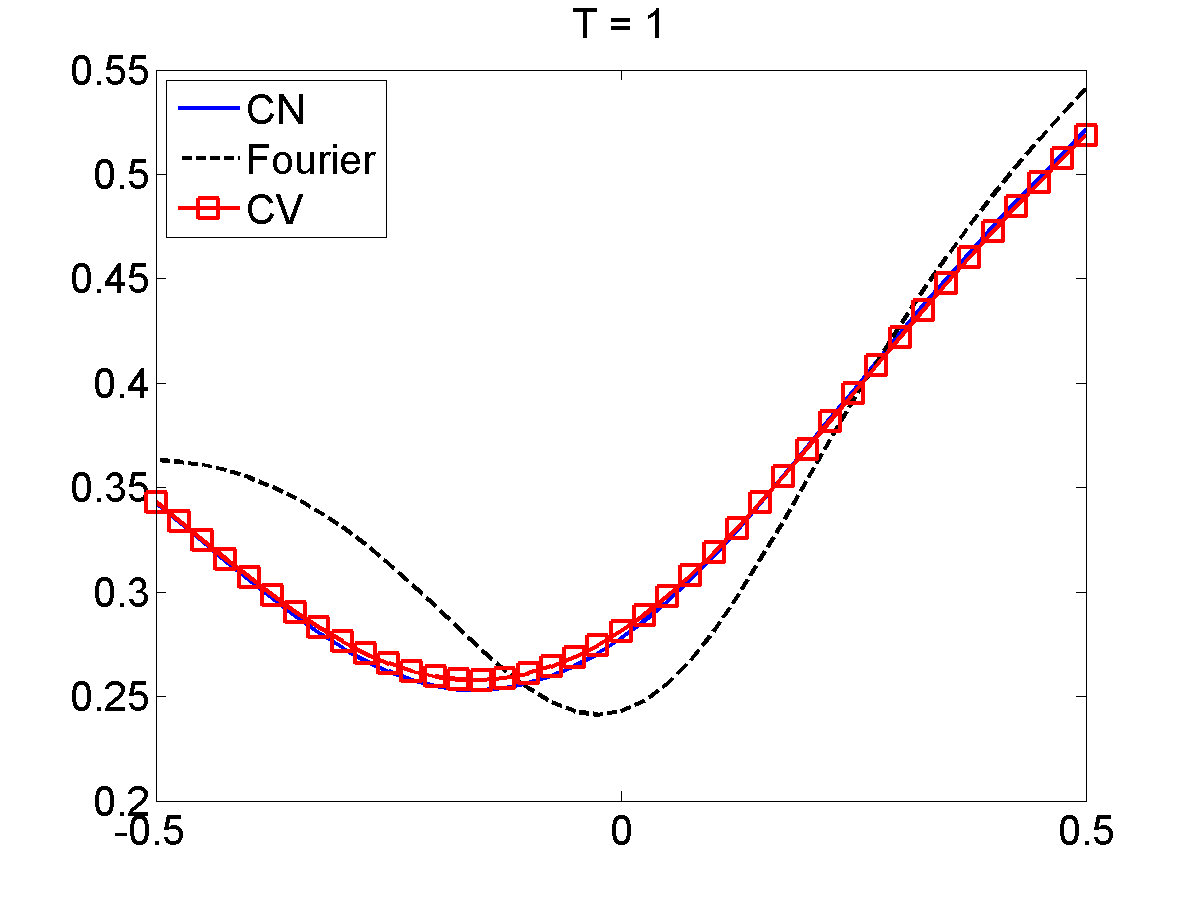}\hfill
  \caption{Implied volatility when local volatility surface is constant.}
  \label{fig:impvol1a}
\end{figure}

\begin{figure}[!ht]
  \centering
      \includegraphics[width=0.249\textwidth]{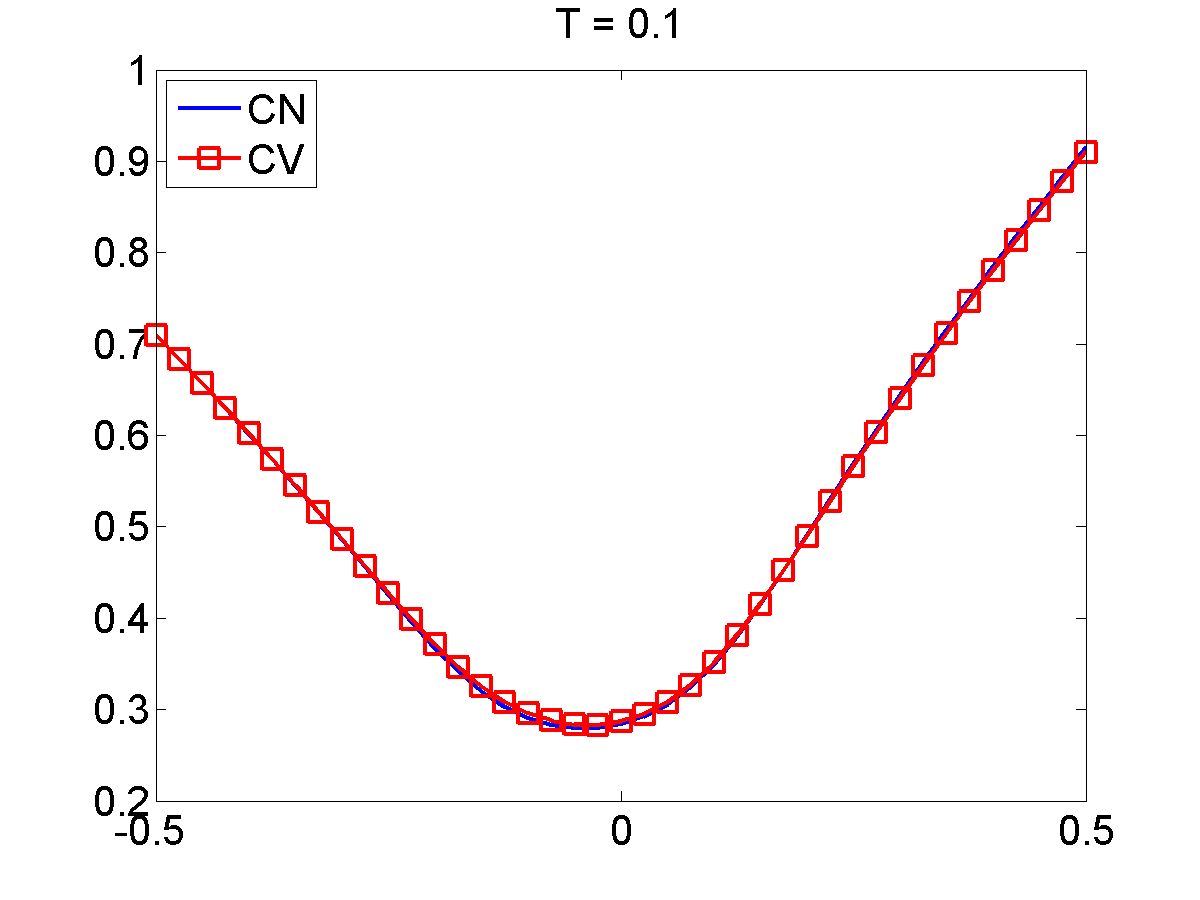}\hfill
      \includegraphics[width=0.249\textwidth]{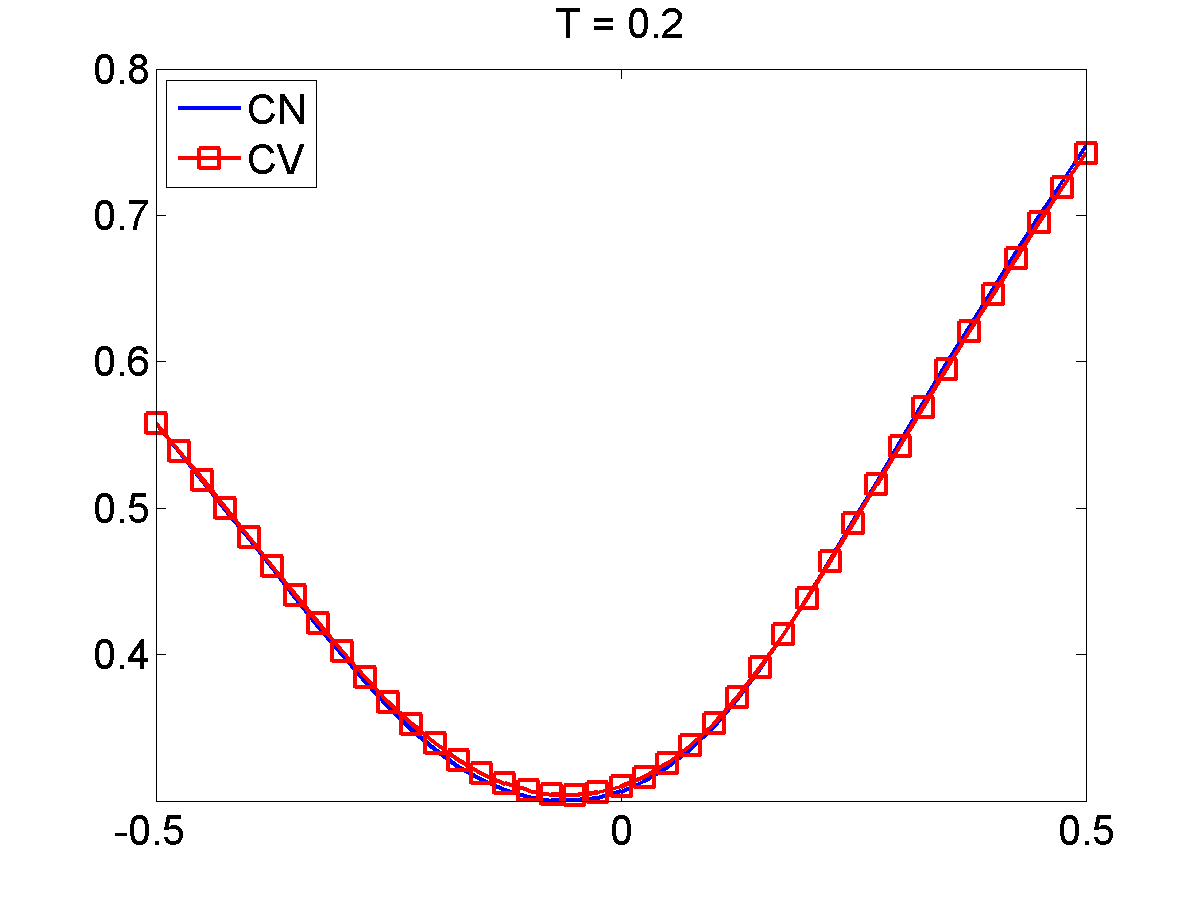}\hfill
      \includegraphics[width=0.249\textwidth]{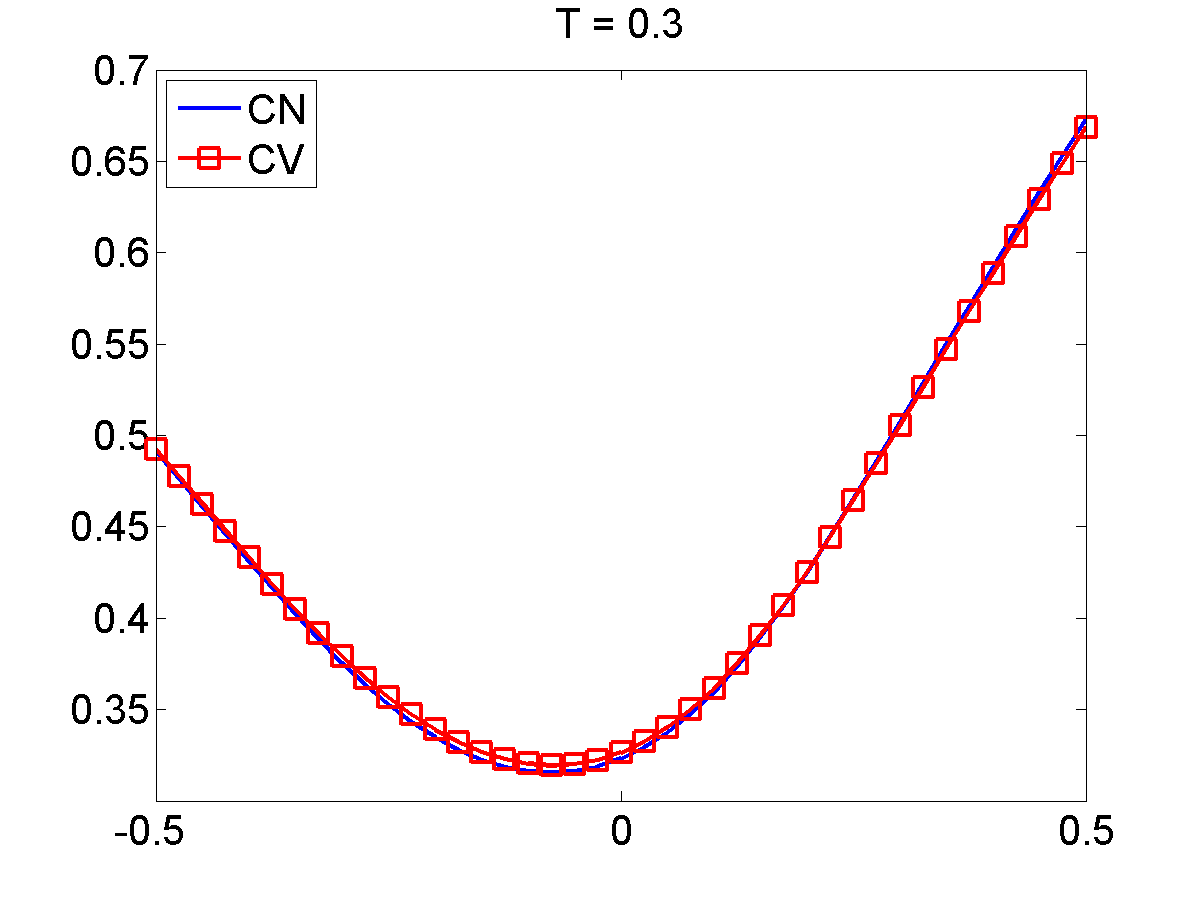}\hfill
      \includegraphics[width=0.249\textwidth]{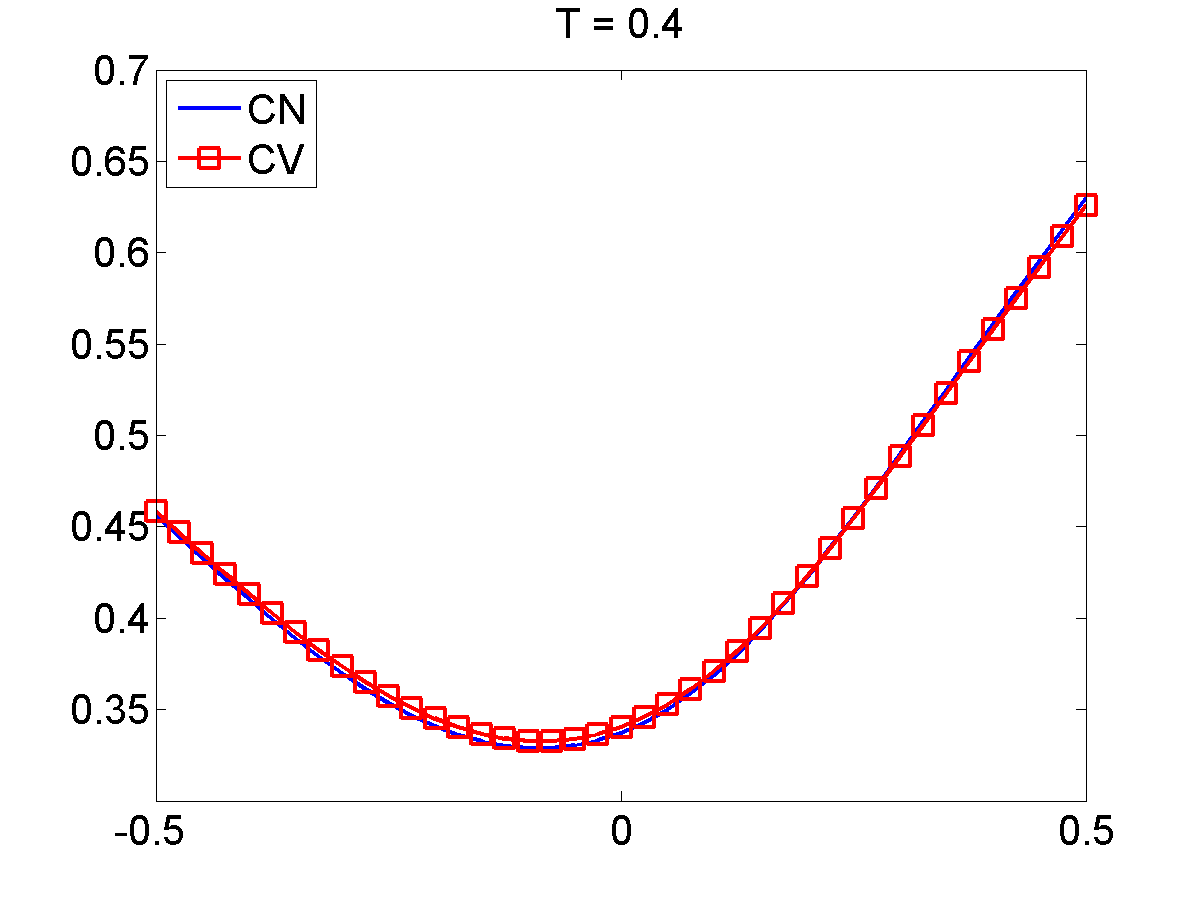}\hfill
      \includegraphics[width=0.249\textwidth]{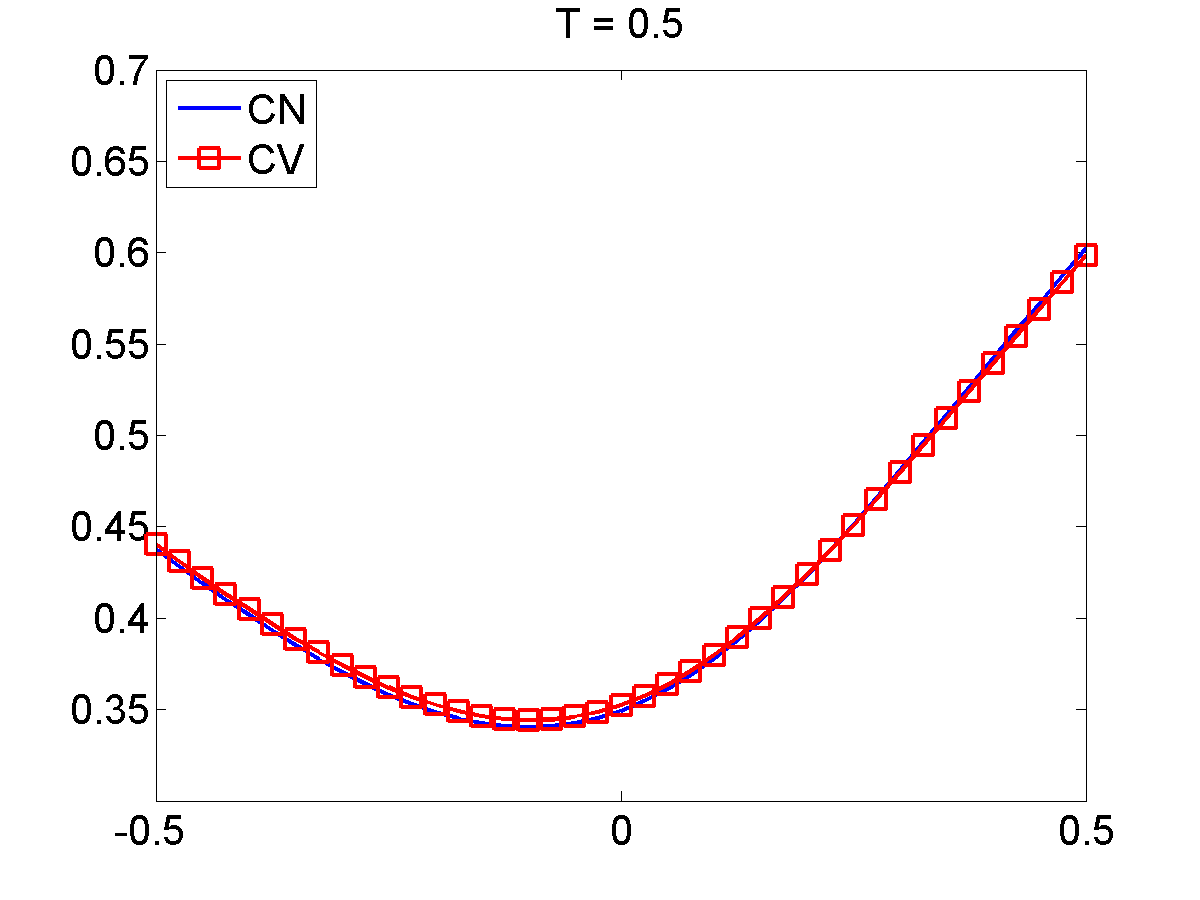}\hfill
      \includegraphics[width=0.249\textwidth]{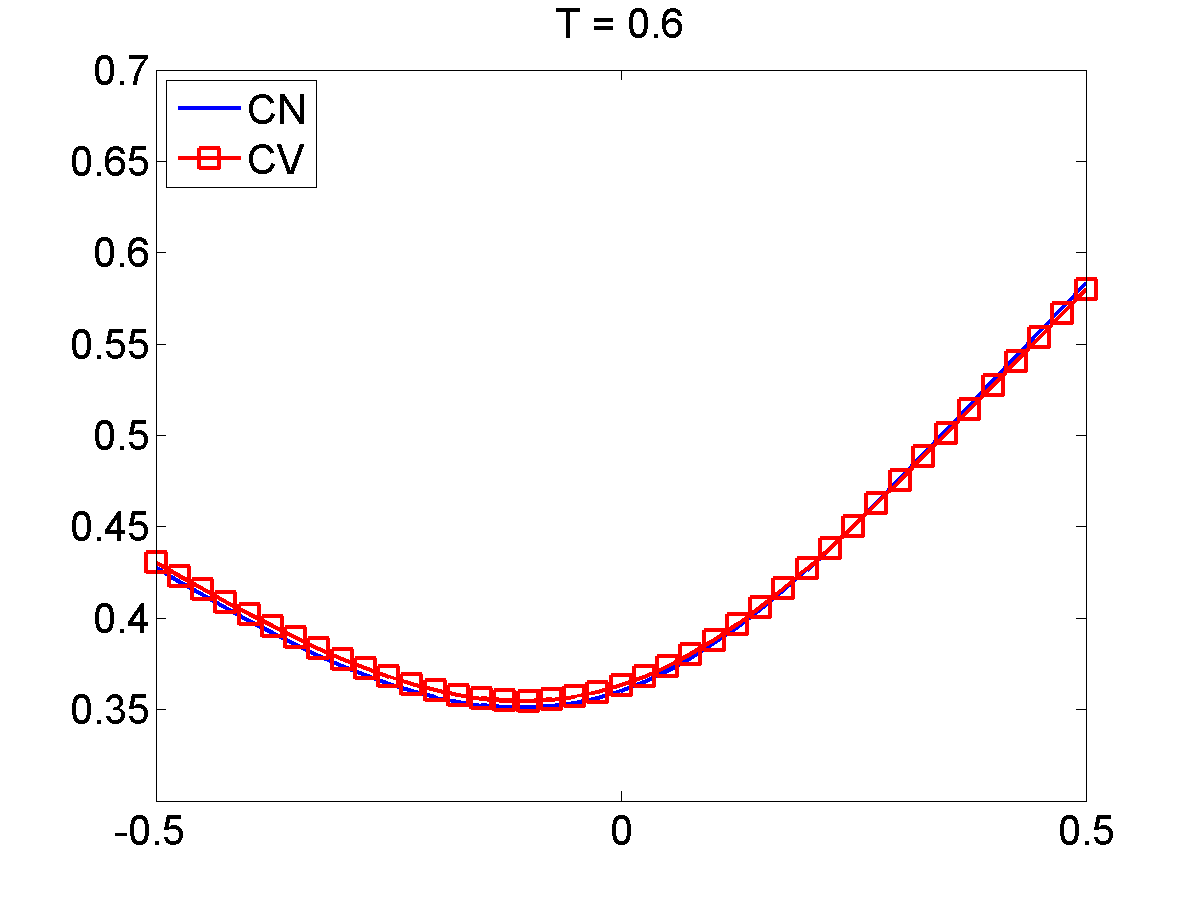}\hfill
      \includegraphics[width=0.249\textwidth]{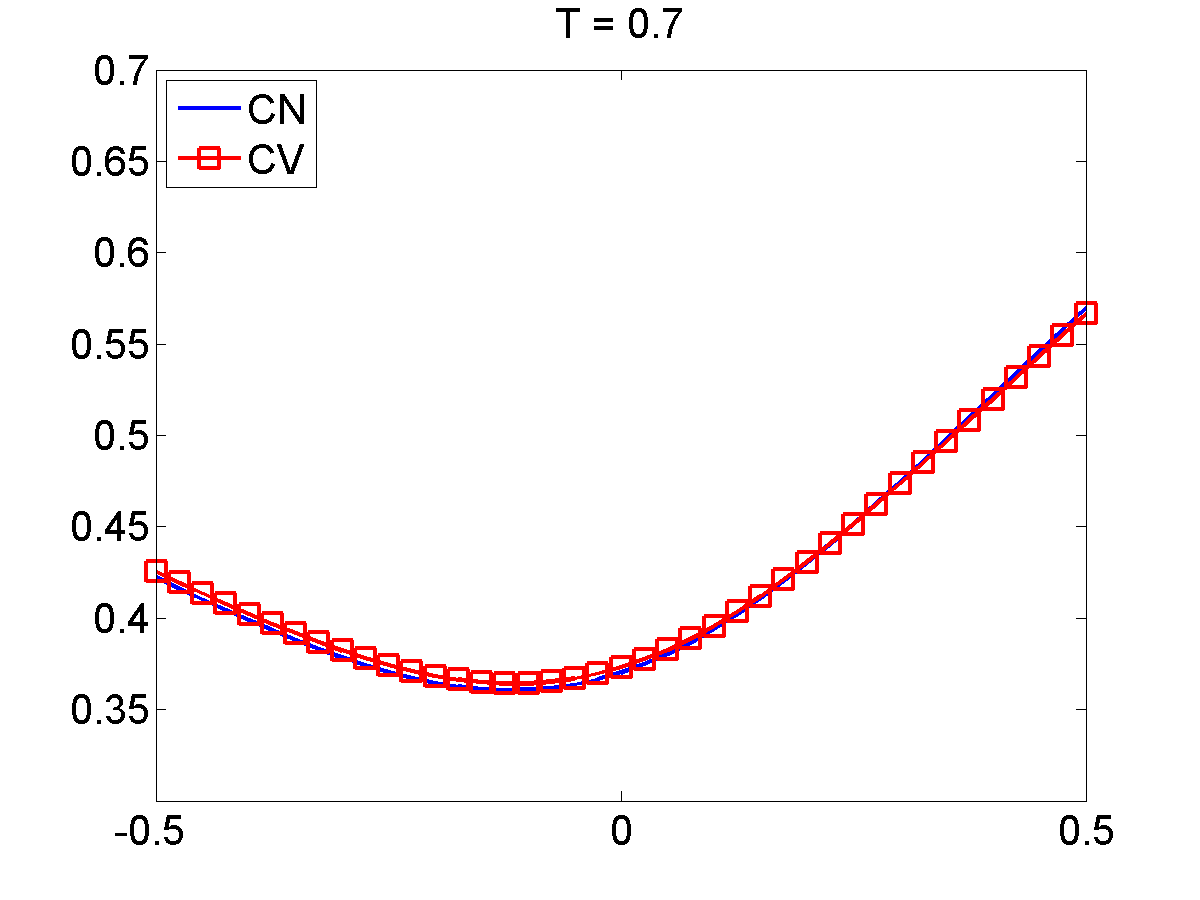}\hfill
      \includegraphics[width=0.249\textwidth]{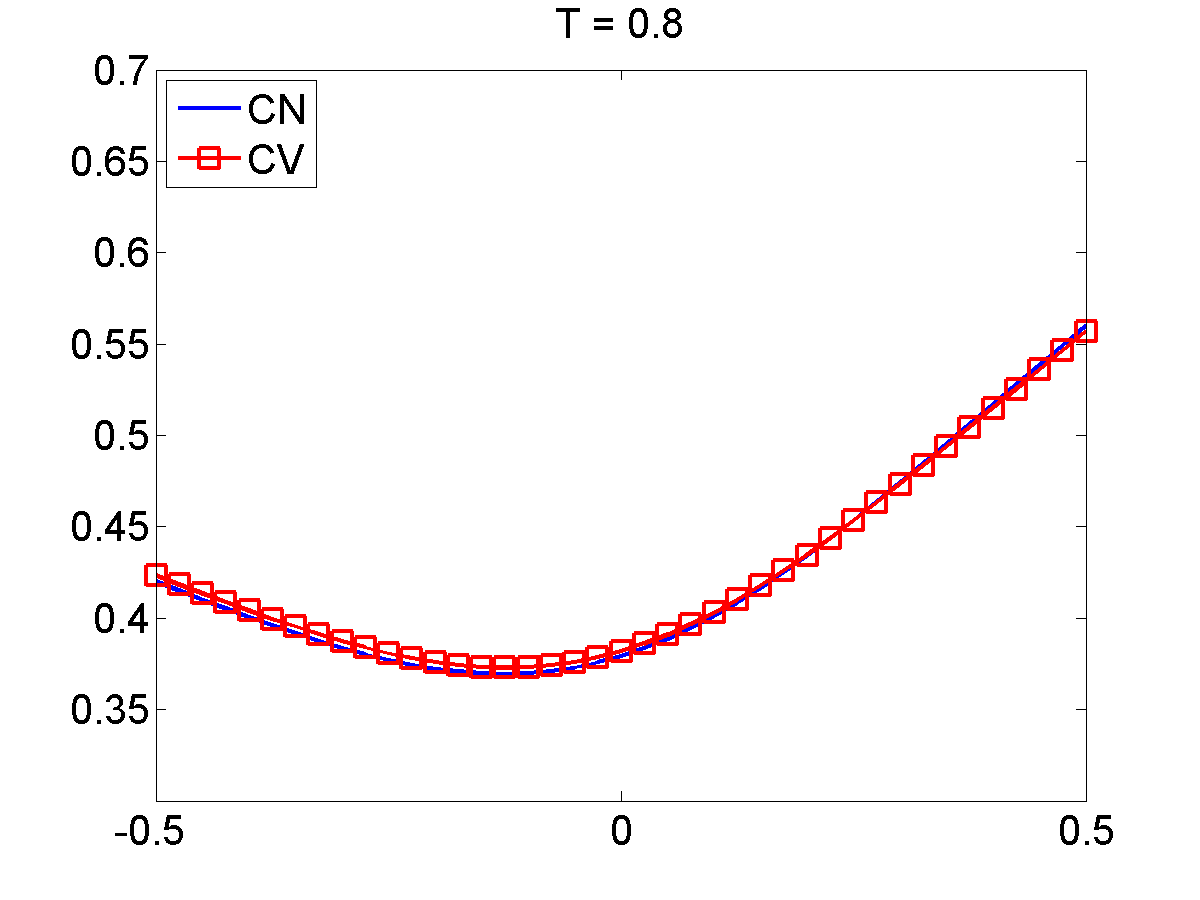}\hfill
      \includegraphics[width=0.249\textwidth]{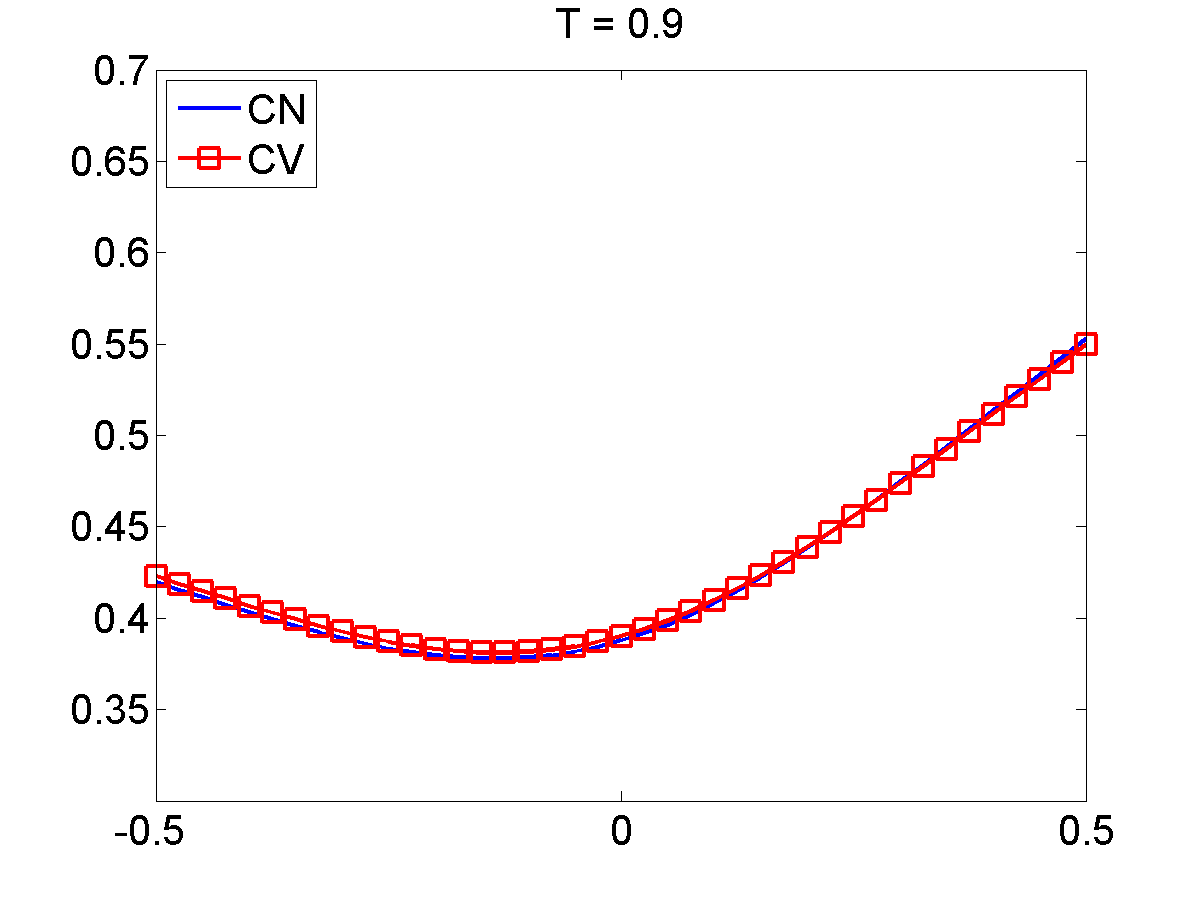}\hfill
      \includegraphics[width=0.249\textwidth]{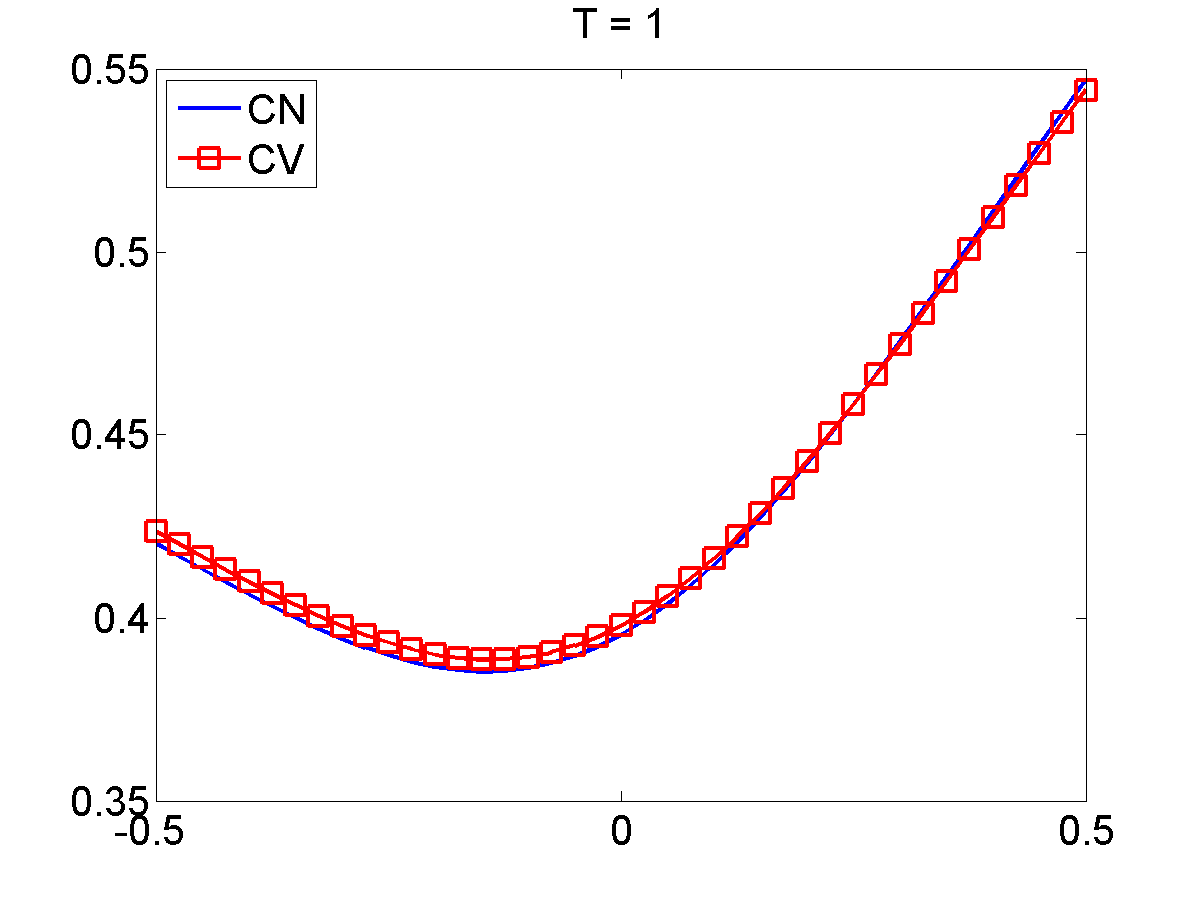}\hfill
  \caption{Implied volatility when local volatility surface is not constant.}
  \label{fig:impvol1b}
\end{figure}

As we can see, CN implied volatilities matched the CV ones  with constant and a non-constant local volatility surface. When comparing with implied volatilities corresponding to the Fourier prices, the adherence of CN volatilities was not exact, but the result was satisfactory, since the relative error and the normalized distance are below $10\%$ of the norm of $\Sigma_{Fourier}$. These results illustrate the accuracy of the present scheme.

\section{Numerical Examples}\label{sec:examples} 

We shall now perform a set of illustrative numerical examples. 

\subsection{Local Volatility Calibration}\label{sec:vol_estimation}

This example is aimed to illustrate that, if $\nu$ is known, it is possible to calibrate the local volatility surface, as in \cite{AndAnd2000}. The European call prices are generated by the difference scheme \eqref{cns}, with local volatility surface \eqref{vol} and jump-size density \eqref{jumpsize}, at the nodes $(\tau_i,y_j) = (i\cdot 0.1,j\cdot0.05)$, with $i=1,..,10$ and $j=-10,-9,...,0,1,...,10$. This is a sparse grid in comparison with the one where the direct problem is solved, see the beginning of Section~\ref{sec:validation}. 

Under a discrete setting, set in the functional \eqref{eq:tikhonov2} the parameters as $\alpha_2 = 0$, $\alpha_1 = 10^{-4}$, and define the penalty functional
\begin{equation}
 f_{a_0}(a) = \|a-a_0\|^2 + \|\partial_{\tau,\Delta} a\|^2 + 100\|\partial_{y,\Delta} a\|^2,
\end{equation}
where $\|\cdot\|$ denotes the $\ell_2$-norm and the operators $\partial_{\tau,\Delta}$ and $\partial_{y,\Delta}$ denote the forward finite difference approximation of the first derivatives w.r.t. $\tau$ and $y$, respectively. The choice of the weights in the penalty functional is  made heuristically and some hints about this choice can be found in \citet{AlbAscZub2016}.

The minimization problem is solved by a gradient-descent method, the step sizes are chosen by a rule inspired by the steepest decent method and the iterations cease whenever the normalized $\ell_2$-residual
$$
\|F(a)-u^\delta\|/\|u^\delta\|,
$$
is less than $0.01$. For more details, see \citet{AlbAscZub2016}. 

The mesh step sizes used to evaluate the local volatility $a$ are the same as those used to generate the data. So, we use the following rule to evaluate the local volatility surface in the whole domain
\[ a(\tau,y) = \begin{cases} a(\tau,-0.5) & {\rm if~} \tau > 0.1, \ y \leq -0.5, (\rm{deep~in~the~money}) \\ 
a(\tau,0.5) & {\rm if~} \tau > 0.1, \ y\geq 0.5, (\rm{deep~out~of~the~money}) \\
a(0.1,y) & {\rm if~} \tau \leq 0.1 , \end{cases} \]
combined with bilinear interpolation.

\begin{figure}[!ht]
  \centering
      \includegraphics[width=0.25\textwidth]{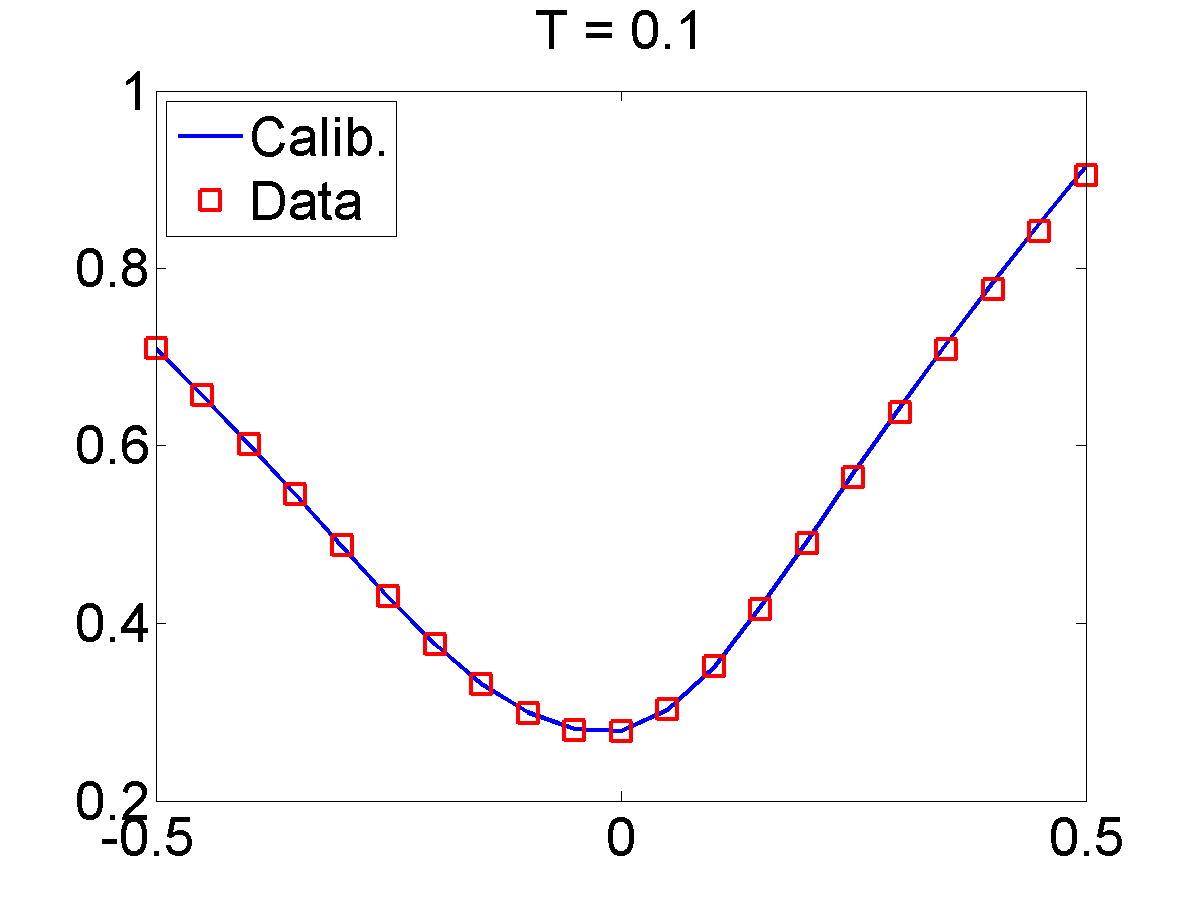}\hfill
      \includegraphics[width=0.25\textwidth]{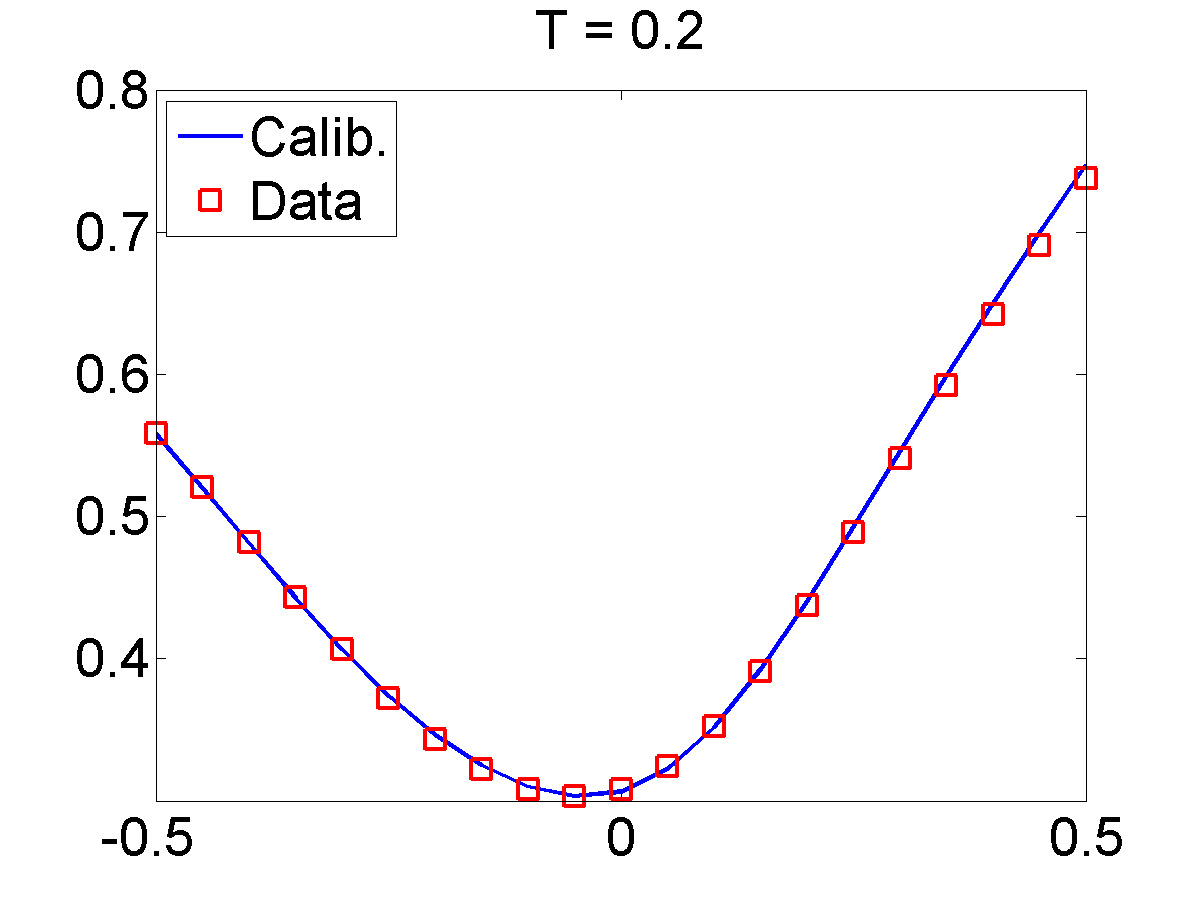}\hfill
      \includegraphics[width=0.25\textwidth]{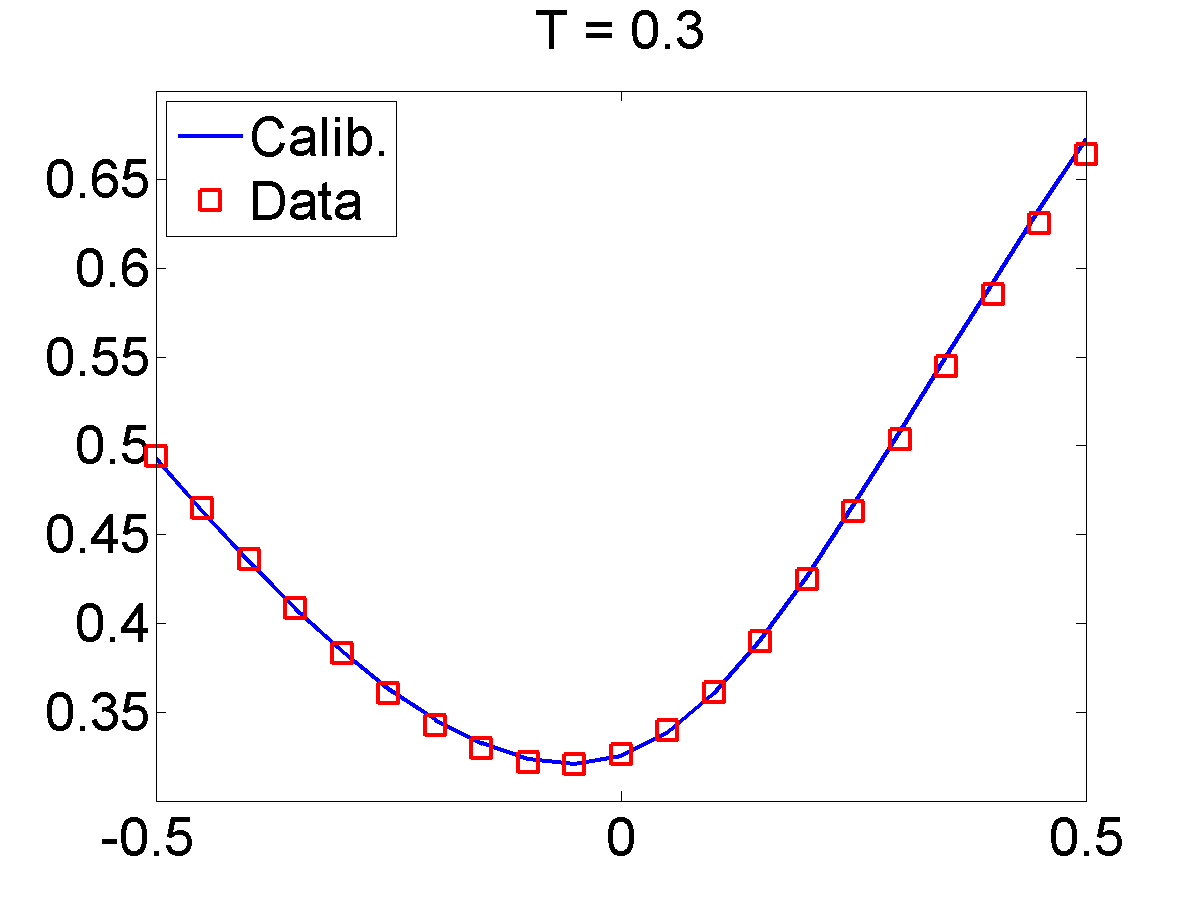}\hfill
      \includegraphics[width=0.25\textwidth]{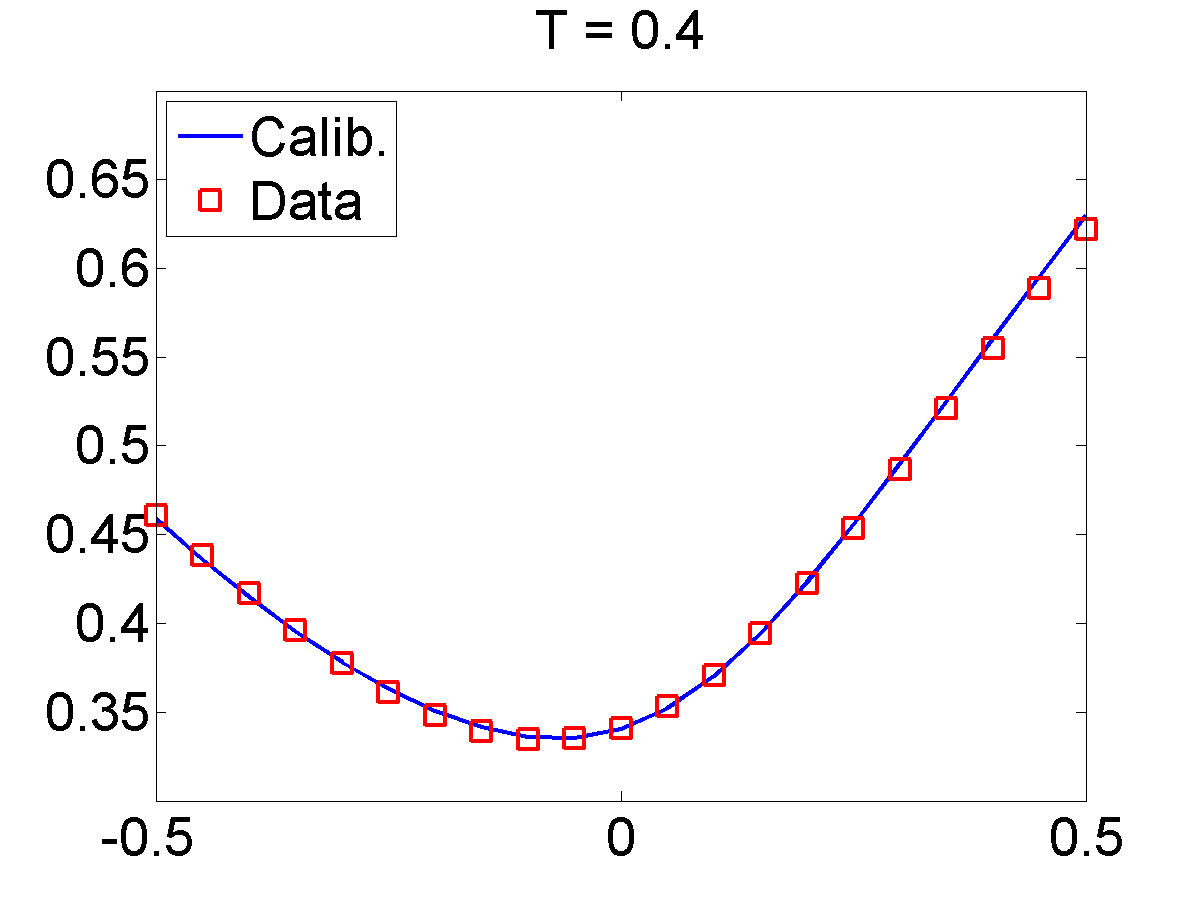}\hfill
      \includegraphics[width=0.25\textwidth]{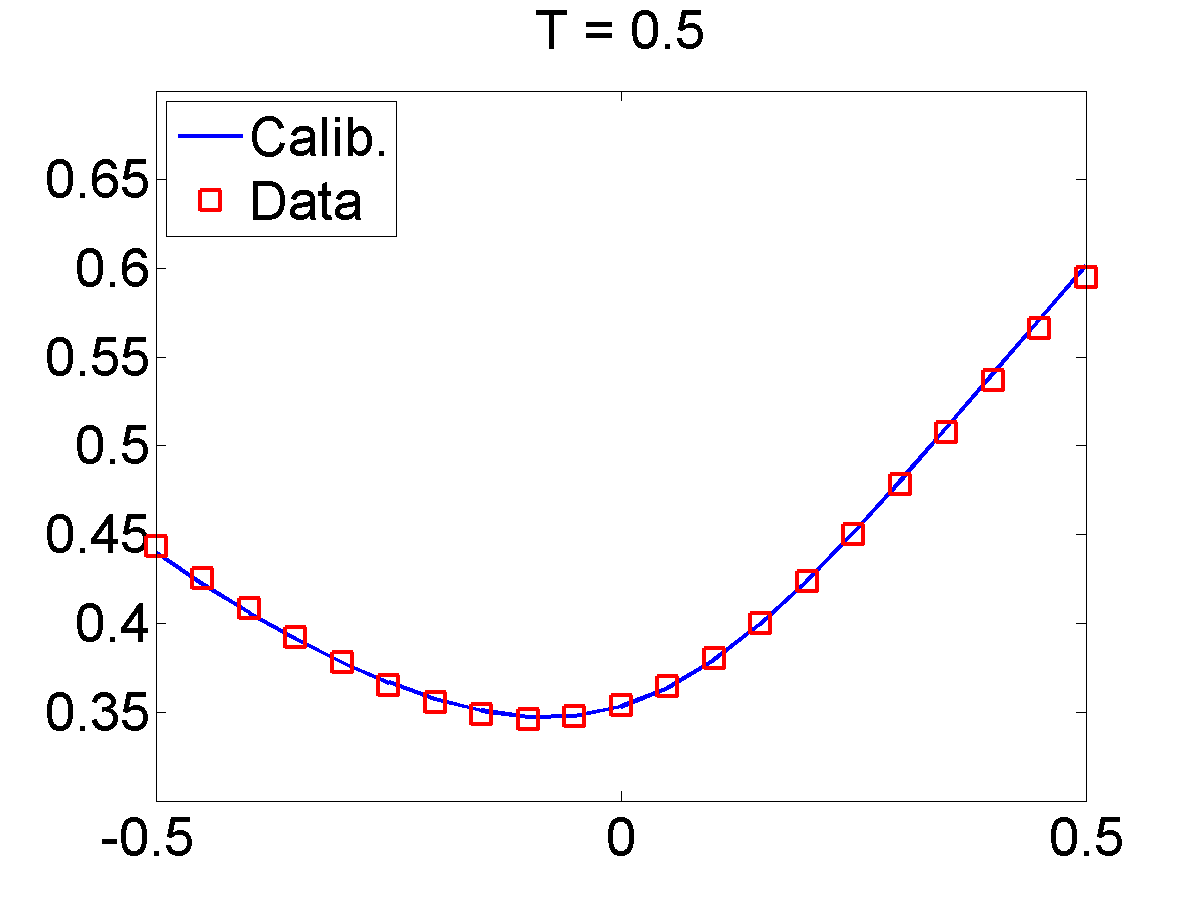}\hfill
      \includegraphics[width=0.25\textwidth]{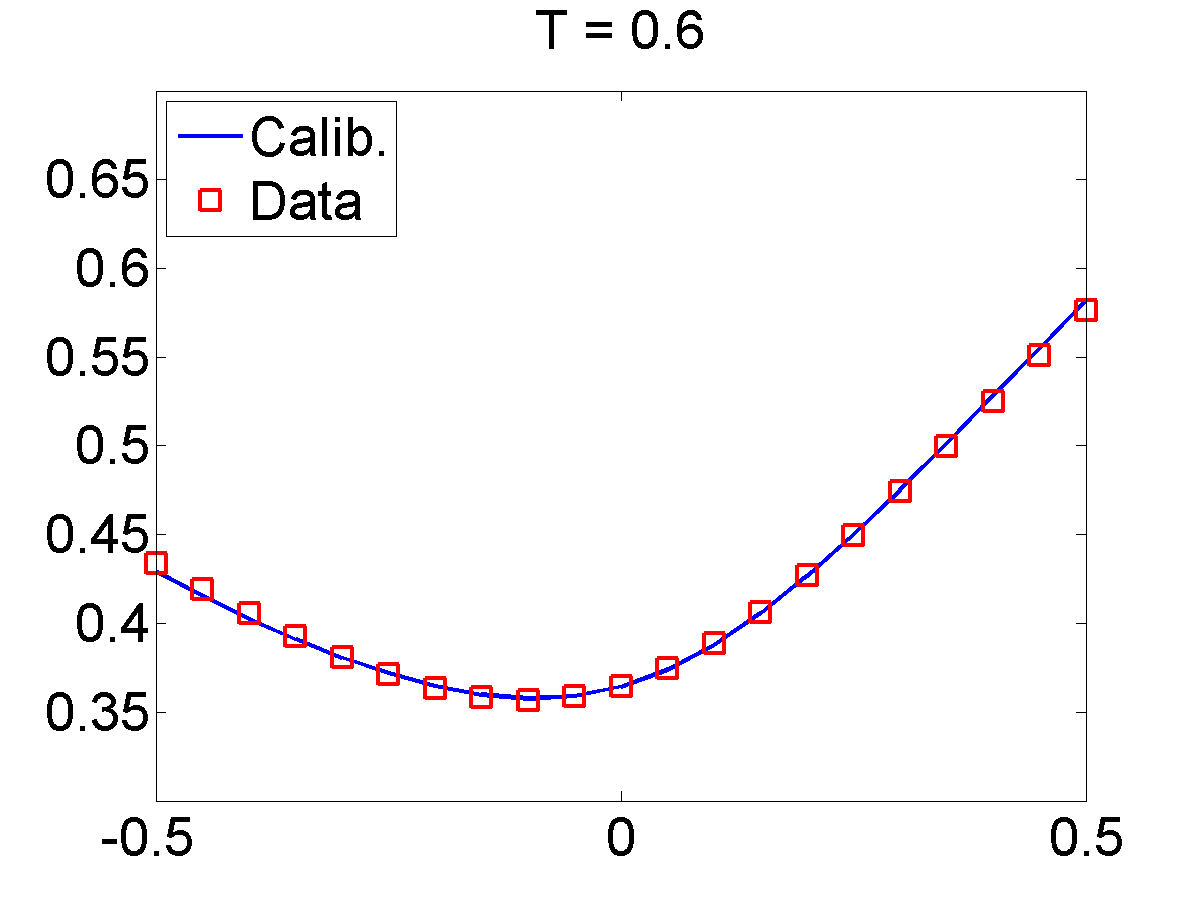}\hfill
      \includegraphics[width=0.25\textwidth]{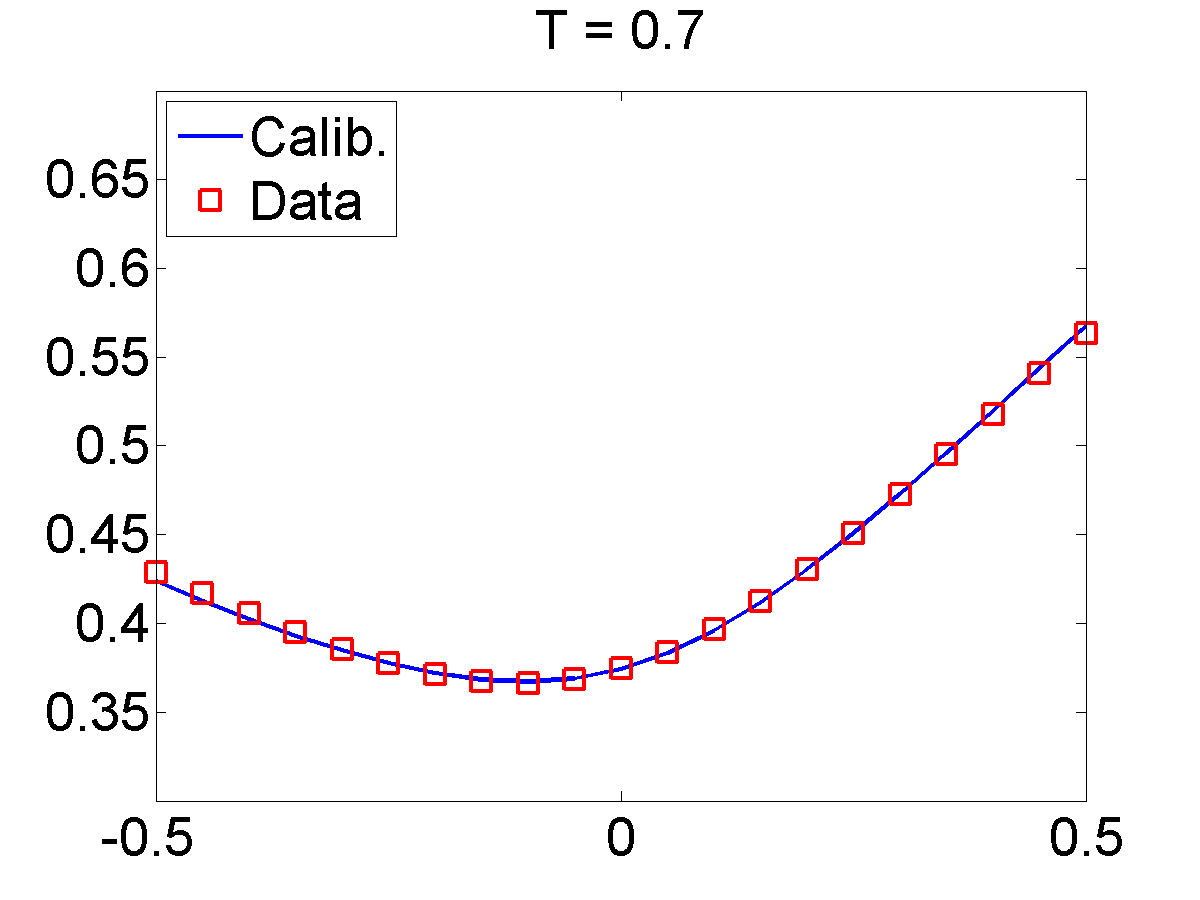}\hfill
      \includegraphics[width=0.25\textwidth]{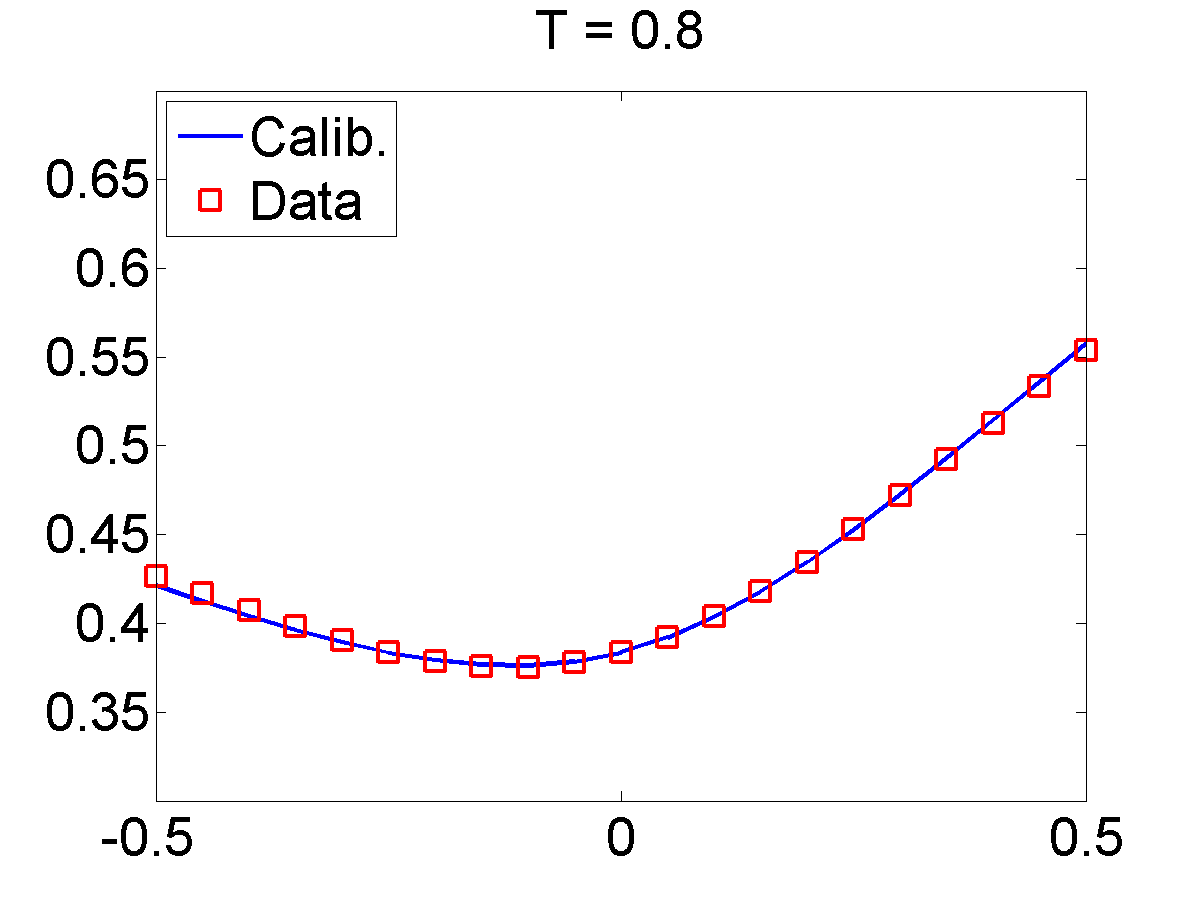}\hfill
      \includegraphics[width=0.25\textwidth]{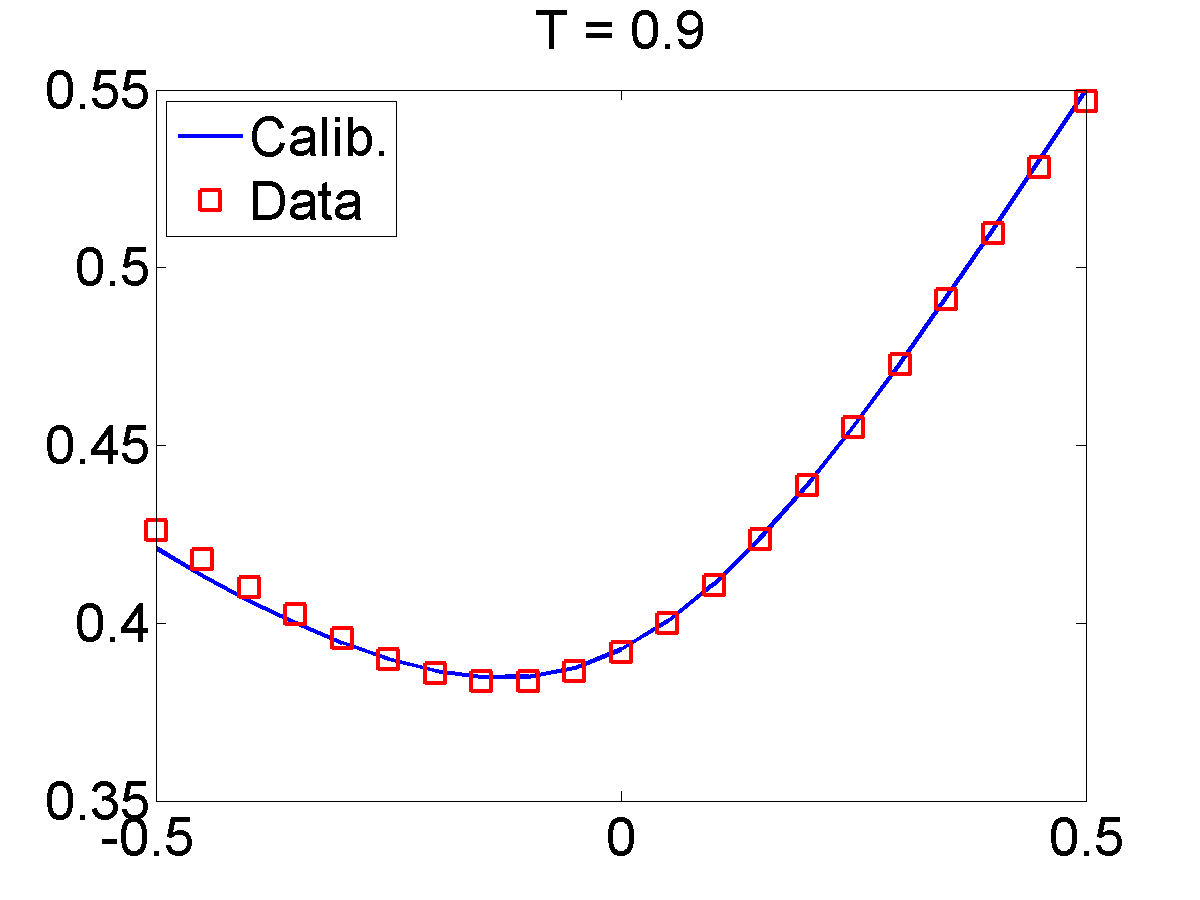}\hfill
      \includegraphics[width=0.25\textwidth]{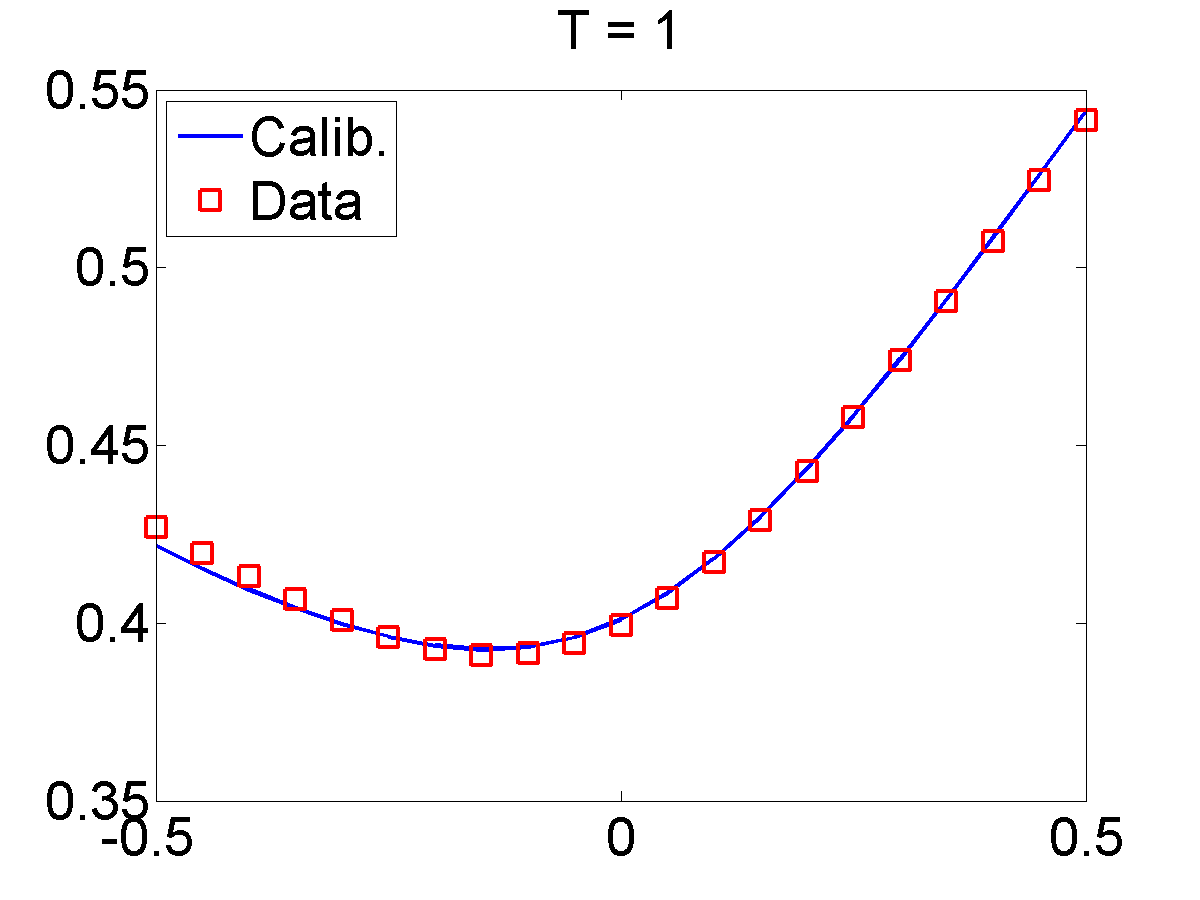}\hfill
  \caption{Implied volatilities of data and the calibrated local volatility.}
  \label{fig:impvol2}
\end{figure}

\begin{figure}[!ht]
  \centering
      \includegraphics[width=0.25\textwidth]{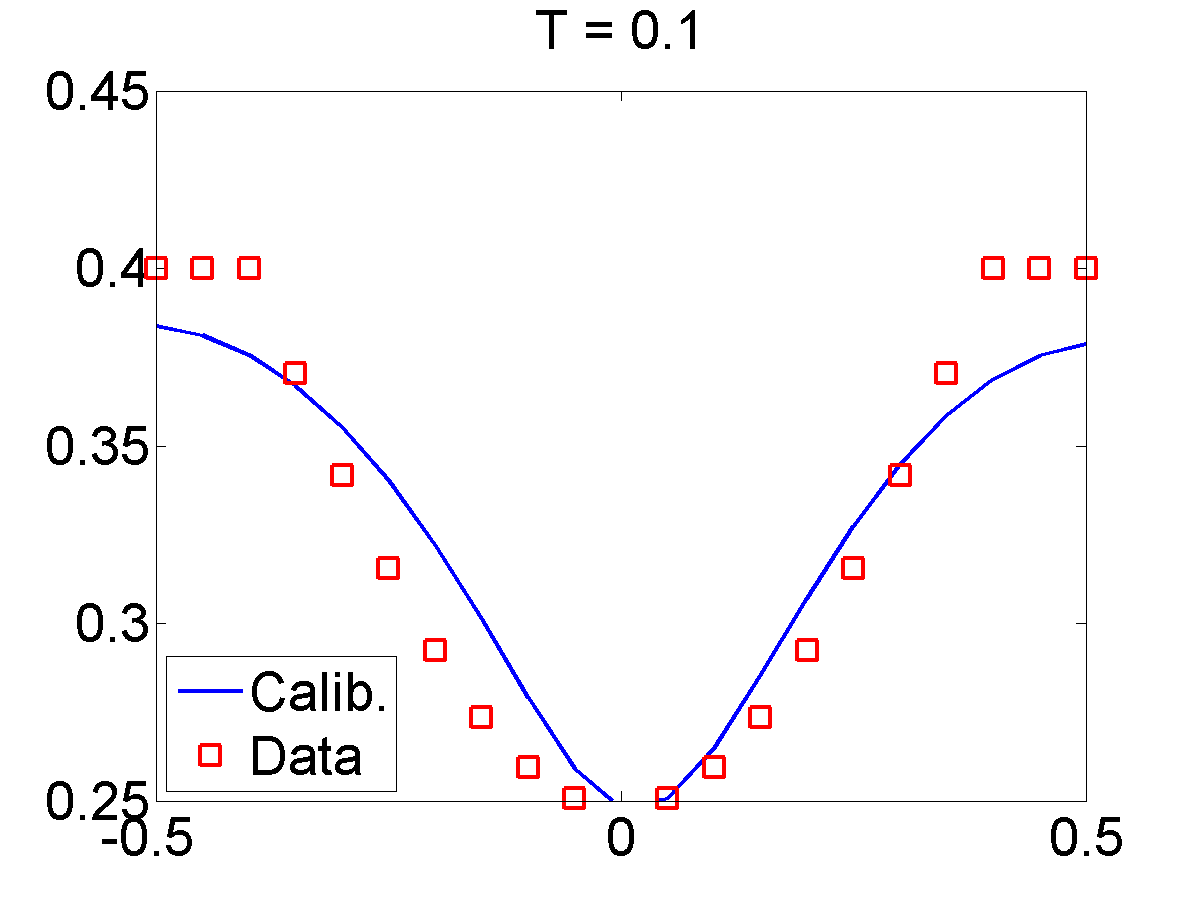}\hfill
      \includegraphics[width=0.25\textwidth]{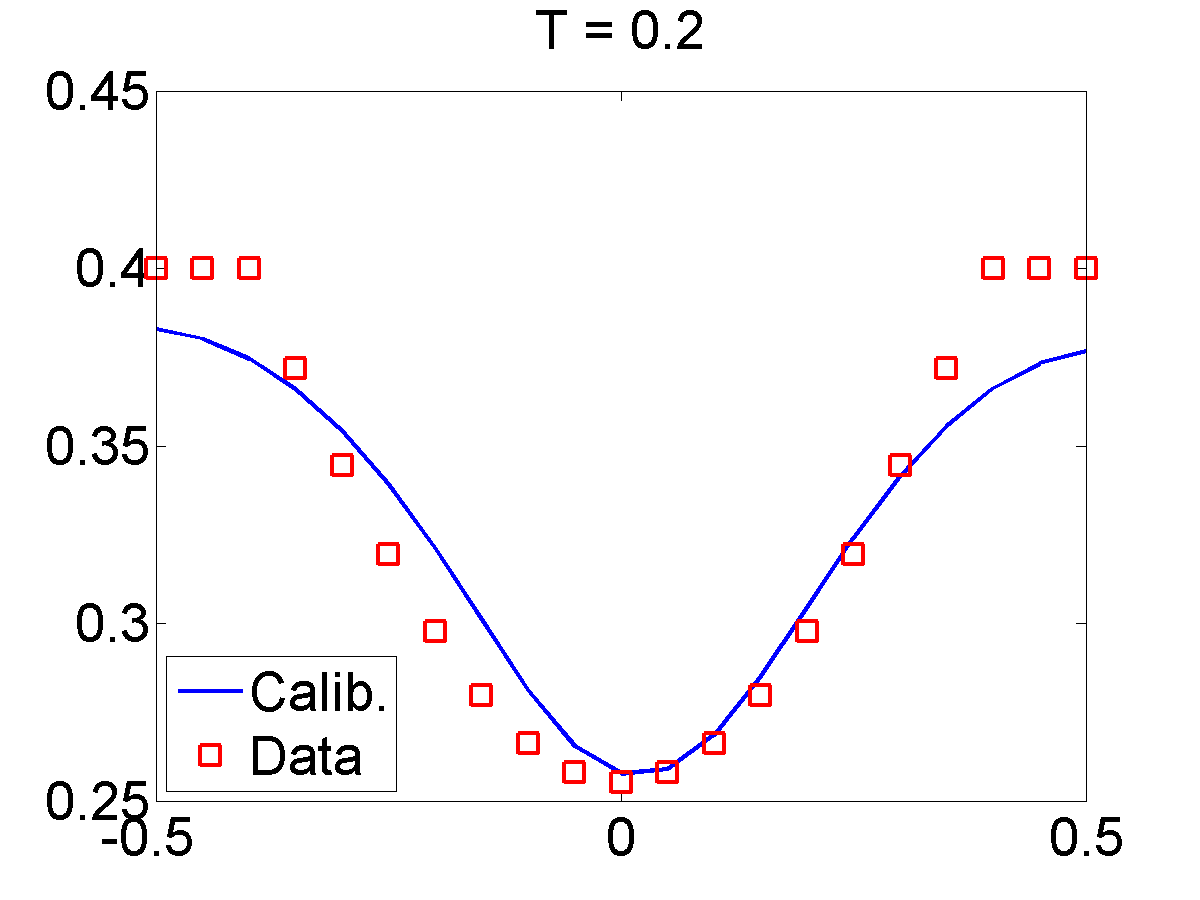}\hfill
      \includegraphics[width=0.25\textwidth]{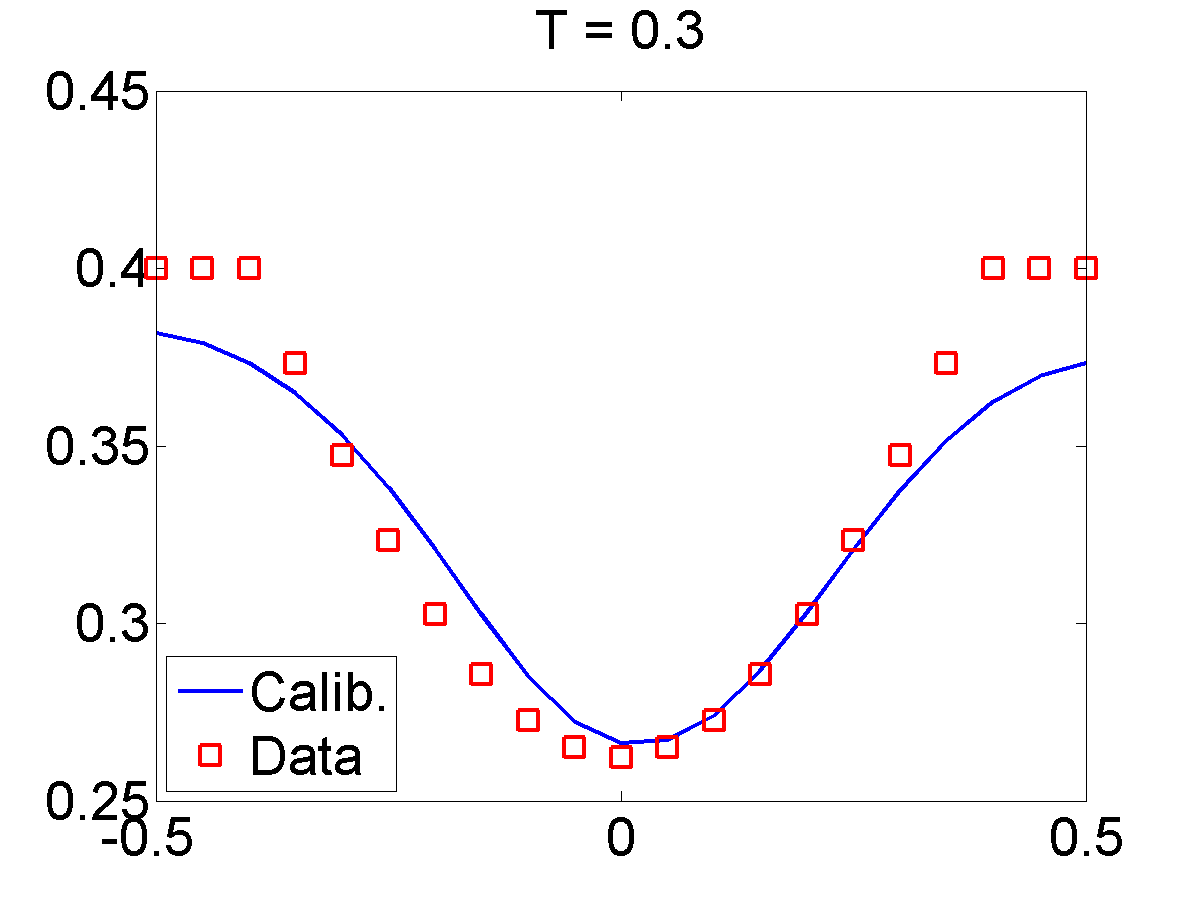}\hfill
      \includegraphics[width=0.25\textwidth]{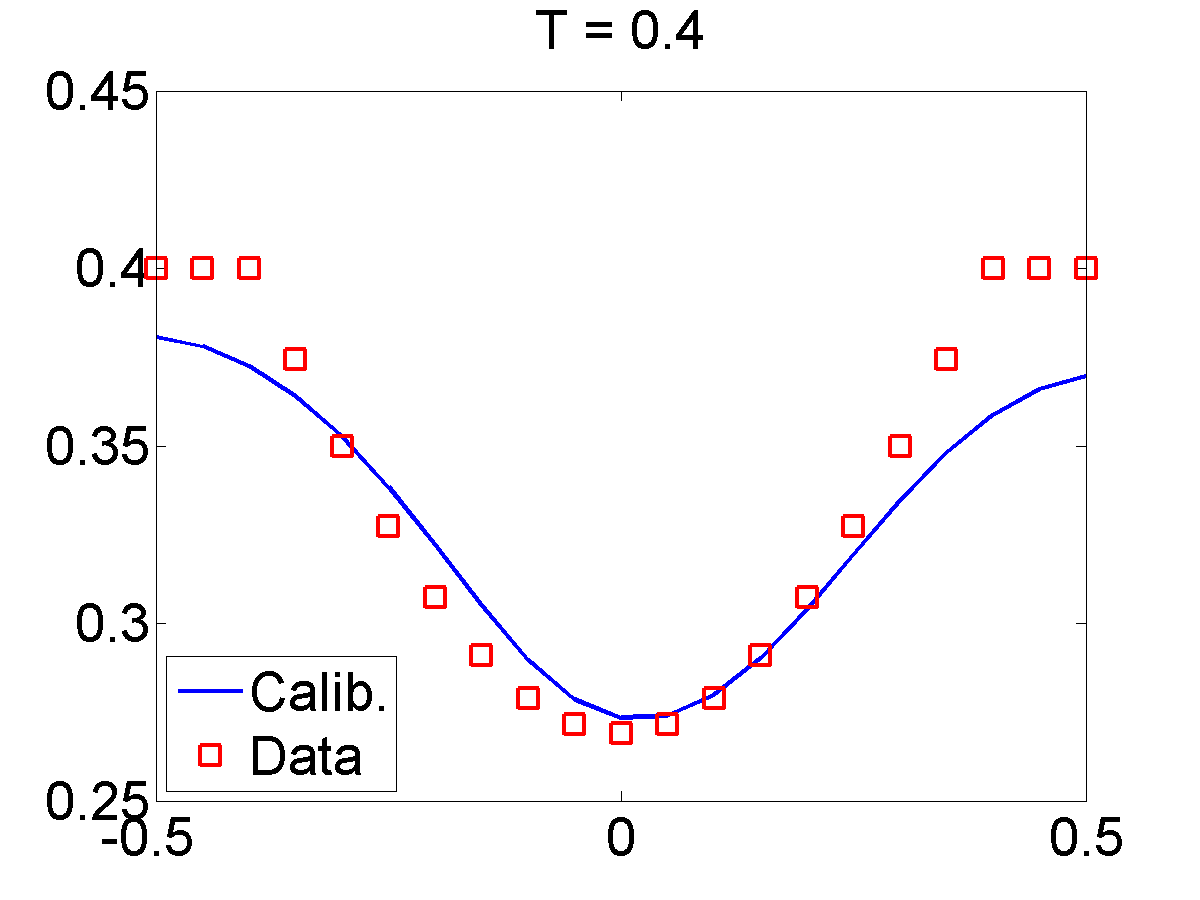}\hfill
      \includegraphics[width=0.25\textwidth]{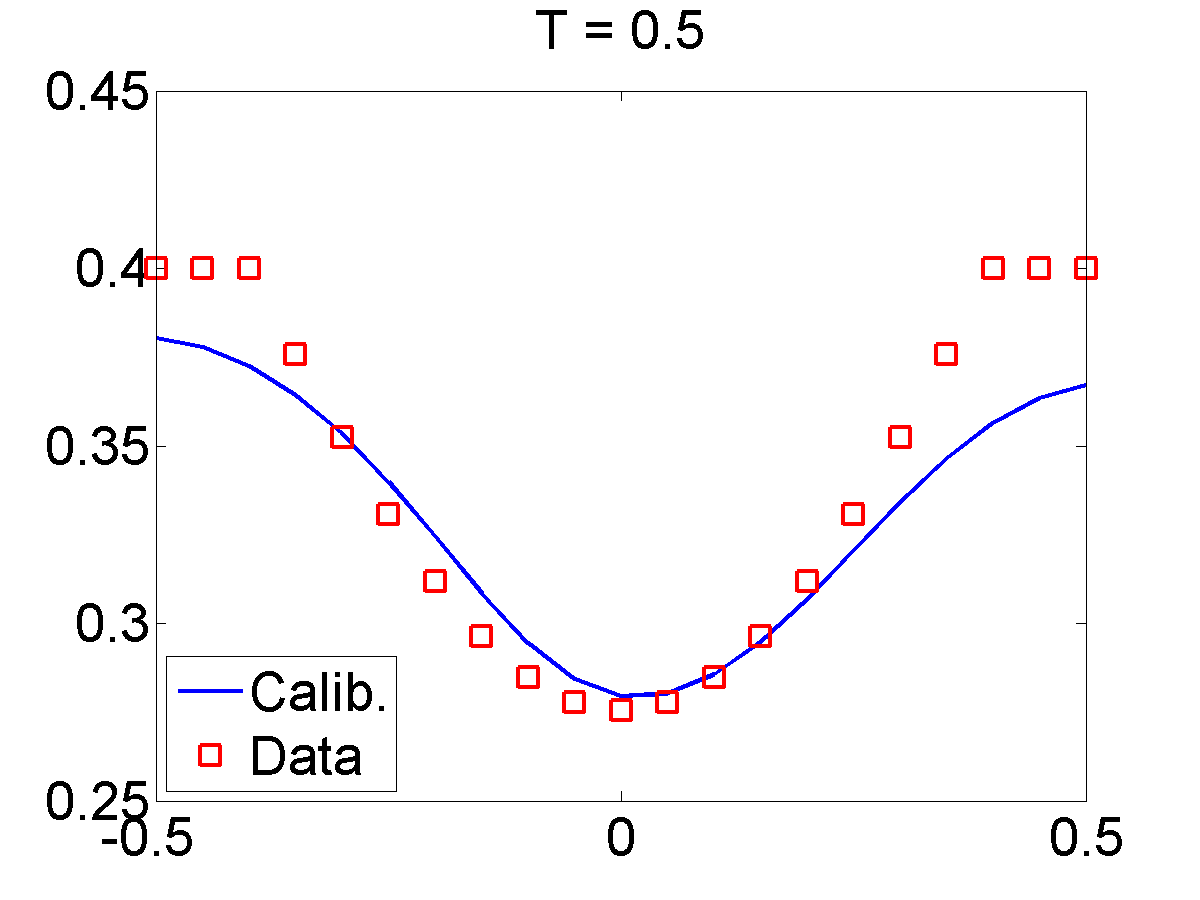}\hfill
      \includegraphics[width=0.25\textwidth]{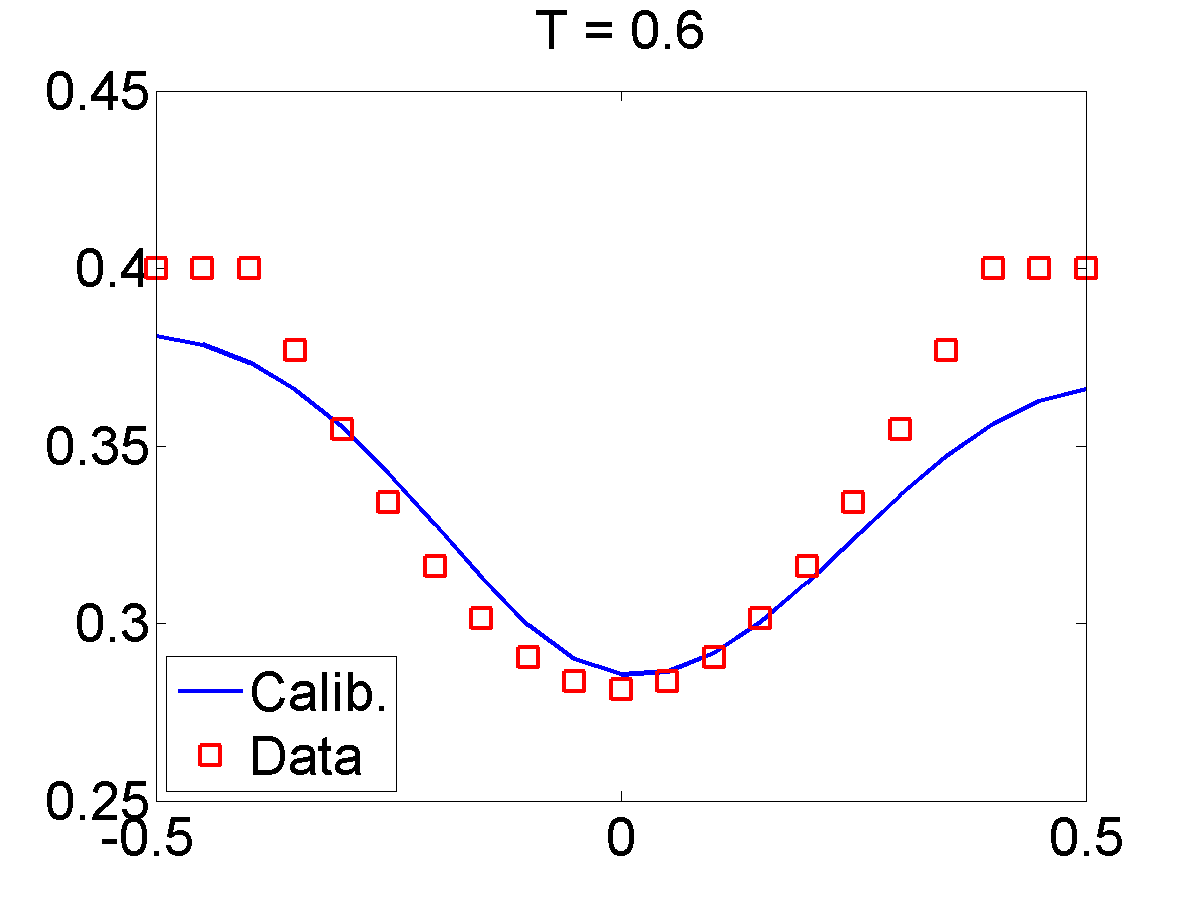}\hfill
      \includegraphics[width=0.25\textwidth]{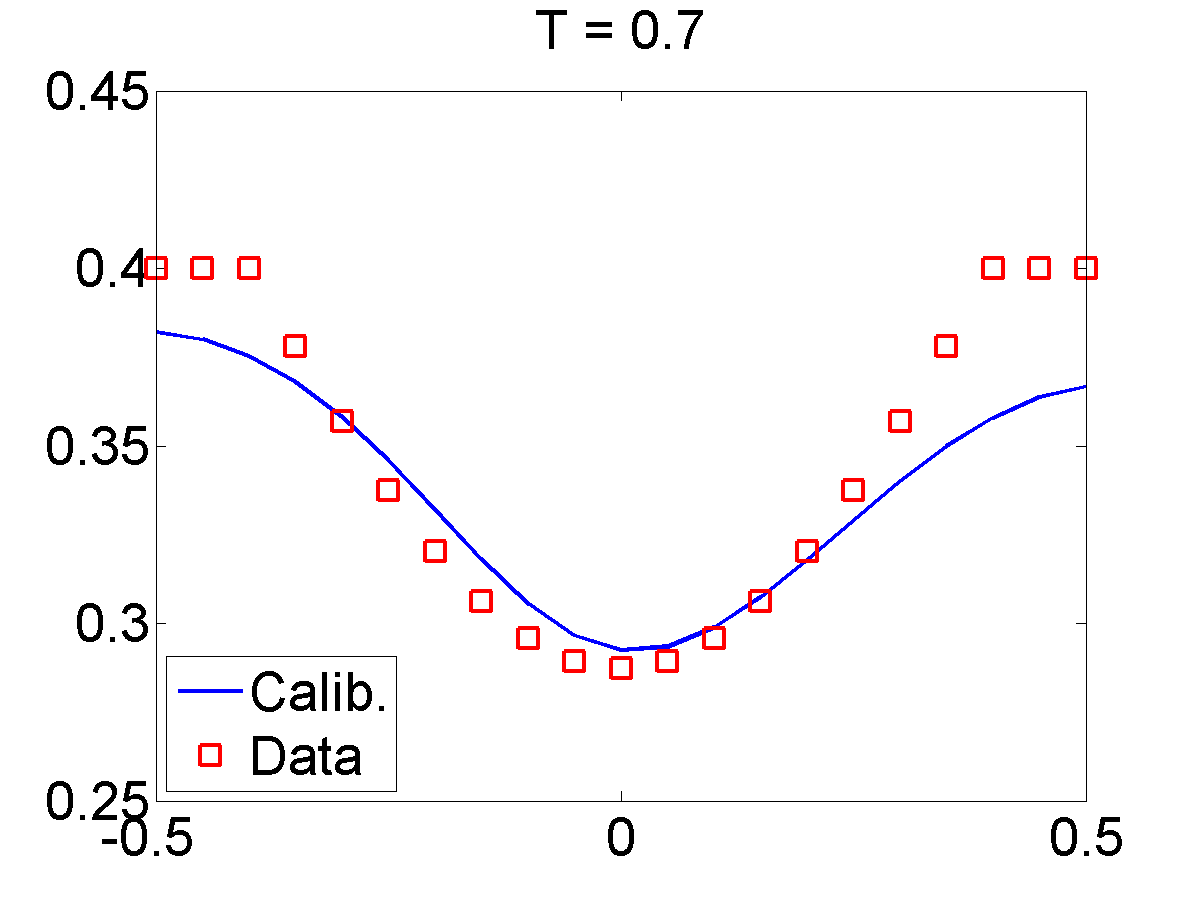}\hfill
      \includegraphics[width=0.25\textwidth]{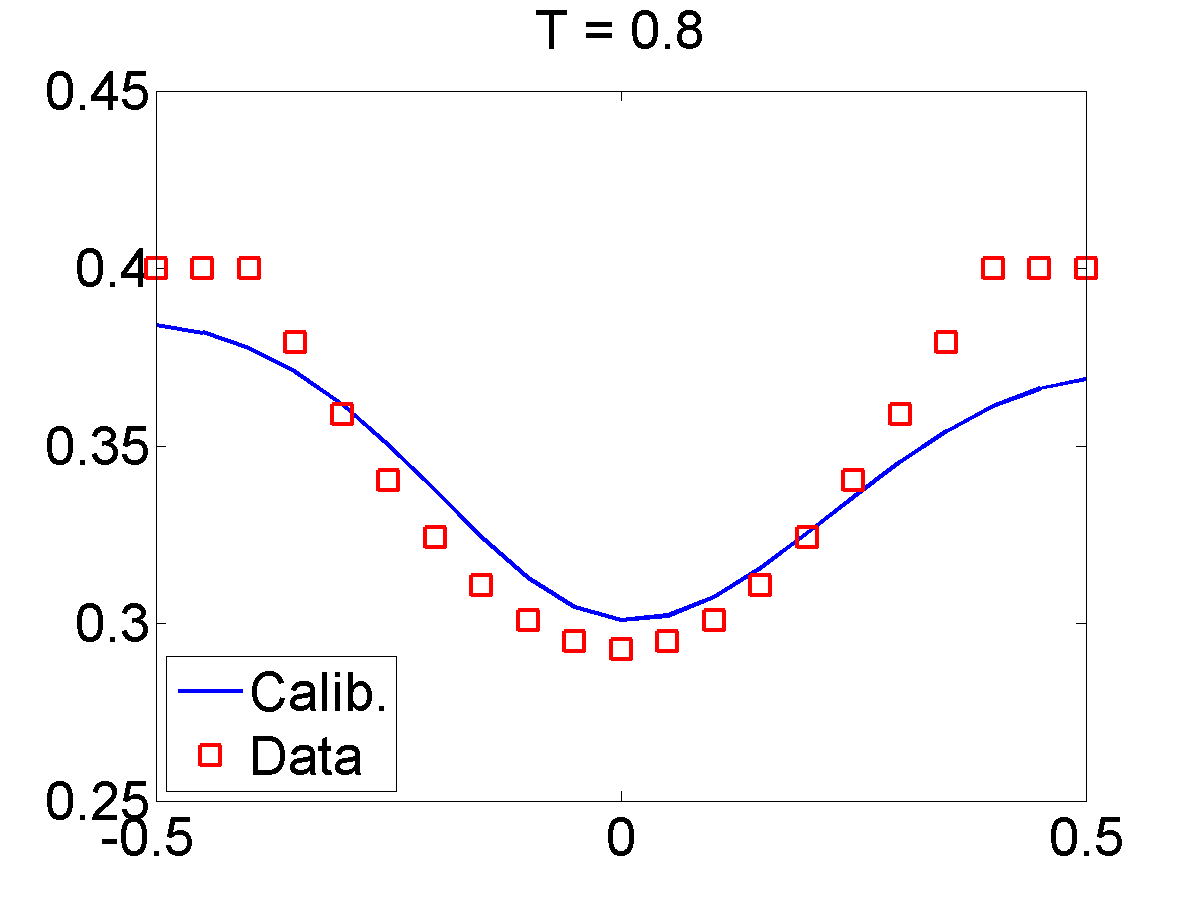}\hfill
      \includegraphics[width=0.25\textwidth]{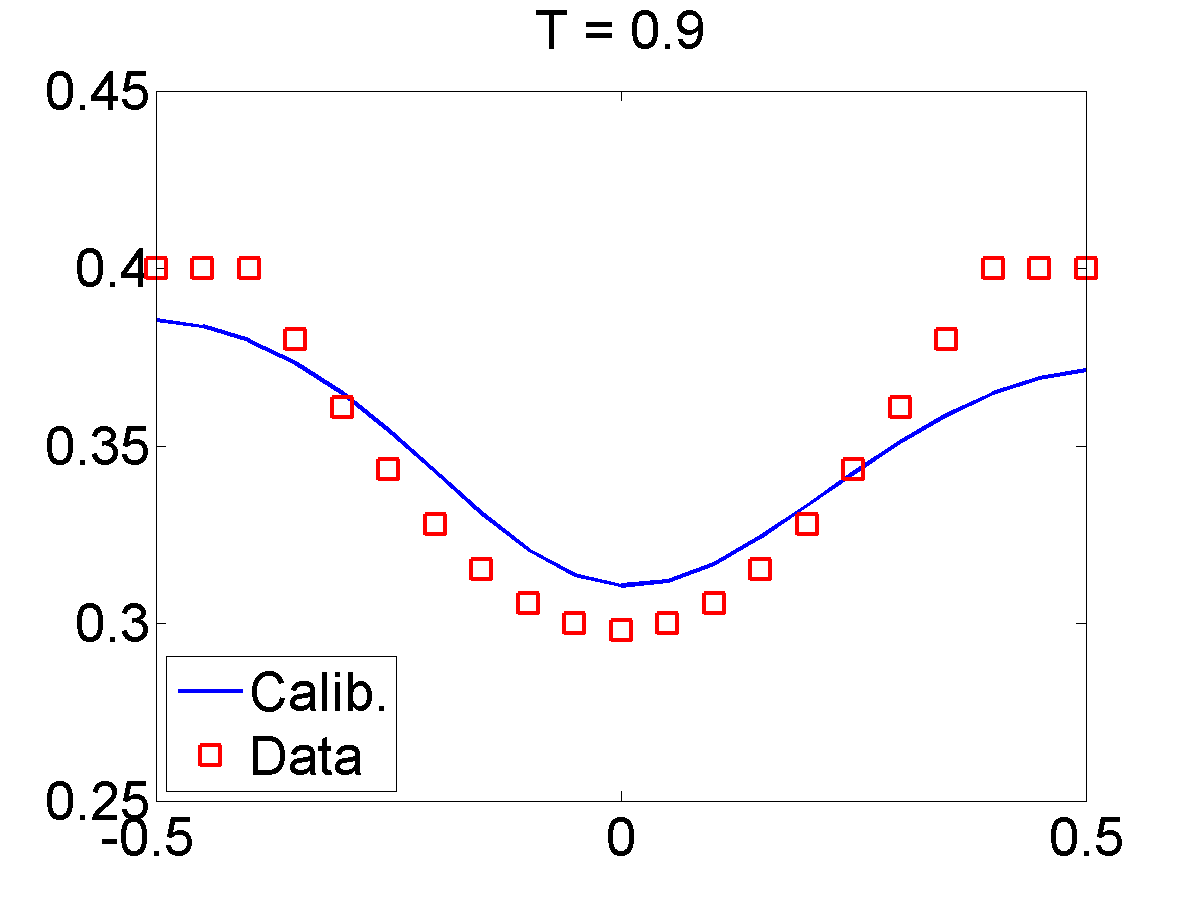}\hfill
      \includegraphics[width=0.25\textwidth]{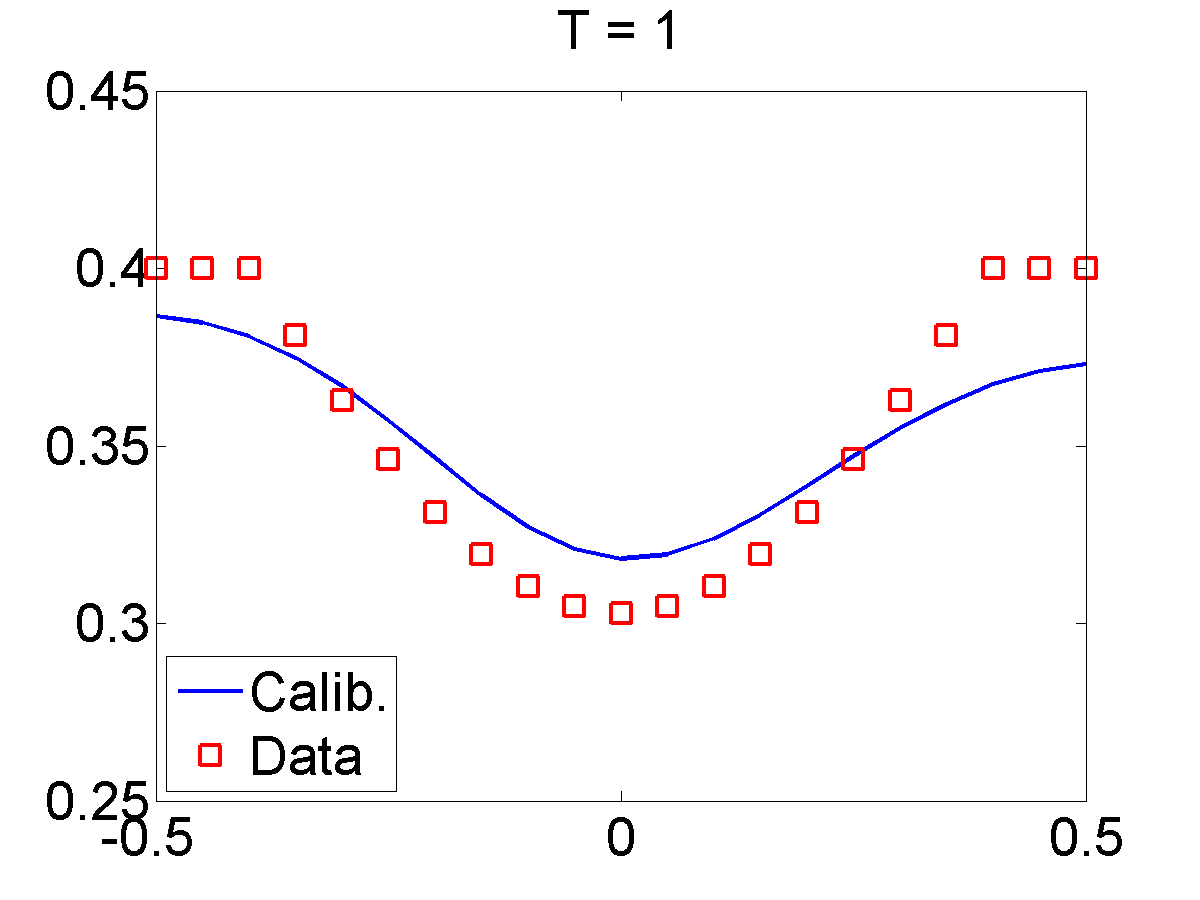}\hfill
  \caption{Original and Calibrated Local volatility surfaces.}
  \label{fig:localvol2}
\end{figure}

The normalized $\ell_2$-distance between the implied volatility of the data and the prices obtained with the calibrated local volatility was  $0.0065$, the mean and standard deviation of the associated absolute relative error at each node were $0.0043$ and $0.0036$, respectively. With respect to the original and the calibrated local volatility surfaces, the normalized $\ell_2$-distance was $0.0701$. The mean and the standard deviation of the corresponding absolute relative error at each node were $0.0583$ and $0.0399$, respectively. The accuracy of our methodology can be also observed in Figures~\ref{fig:impvol2}-\ref{fig:localvol2} where the implied volatilities of the model matched the data ones, and the reconstructed local volatility was quite similar to the original one. In both figures, ``Calib.'' stands for the calibrated local volatility and ``Data'' stands for the original one. Note that, the calibration was not perfect, since the data is collected in a sparse grid. 

\subsection{Calibration of jump-size distribution}\label{sec:nu_calibbration}

Assuming that the local volatility surface is given, the double-exponential tail and the jump-size distribution are calibrated form observed prices. For this example, the same synthetic data and parameters presented in Section~\ref{sec:vol_estimation} are used. 

Define
$$
\nu_j = \int_{y_j-\frac{\Delta y}{2}}^{y_j+\frac{\Delta y}{2}}\nu(dy).
$$
Firstly, we calibrate $\varphi$, and then, $\nu$ is reconstructed from $\varphi$, by minimizing the functional:
\begin{equation}
 \sum^{M}_{j=-M}(\varphi_j - \varphi(\nu)_j)^2 + \alpha\sum^{M}_{j=-M}\left[\nu_j\log(\nu_j/\nu_{j,0}) - (\nu_{j,0} - \nu_j)\right],
\end{equation}
where $\varphi(\nu)_j$ is given by
$$
\varphi(\nu)_j = \left\{
\begin{array}{ll}
\displaystyle\sum_{l=-M}^{j}(\text{e}^{y_j}-\text{e}^{y_l})\nu_l, & y_j < 0\\ 
\displaystyle\sum_{l=j}^{M}(\text{e}^{y_l}-\text{e}^{y_j})\nu_l, & y_j > 0.
\end{array}
\right.
$$
The regularization parameter is set as $\alpha = 1\times 10^{-5}$.

\begin{figure}[!ht]
  \centering
      \includegraphics[width=0.4\textwidth]{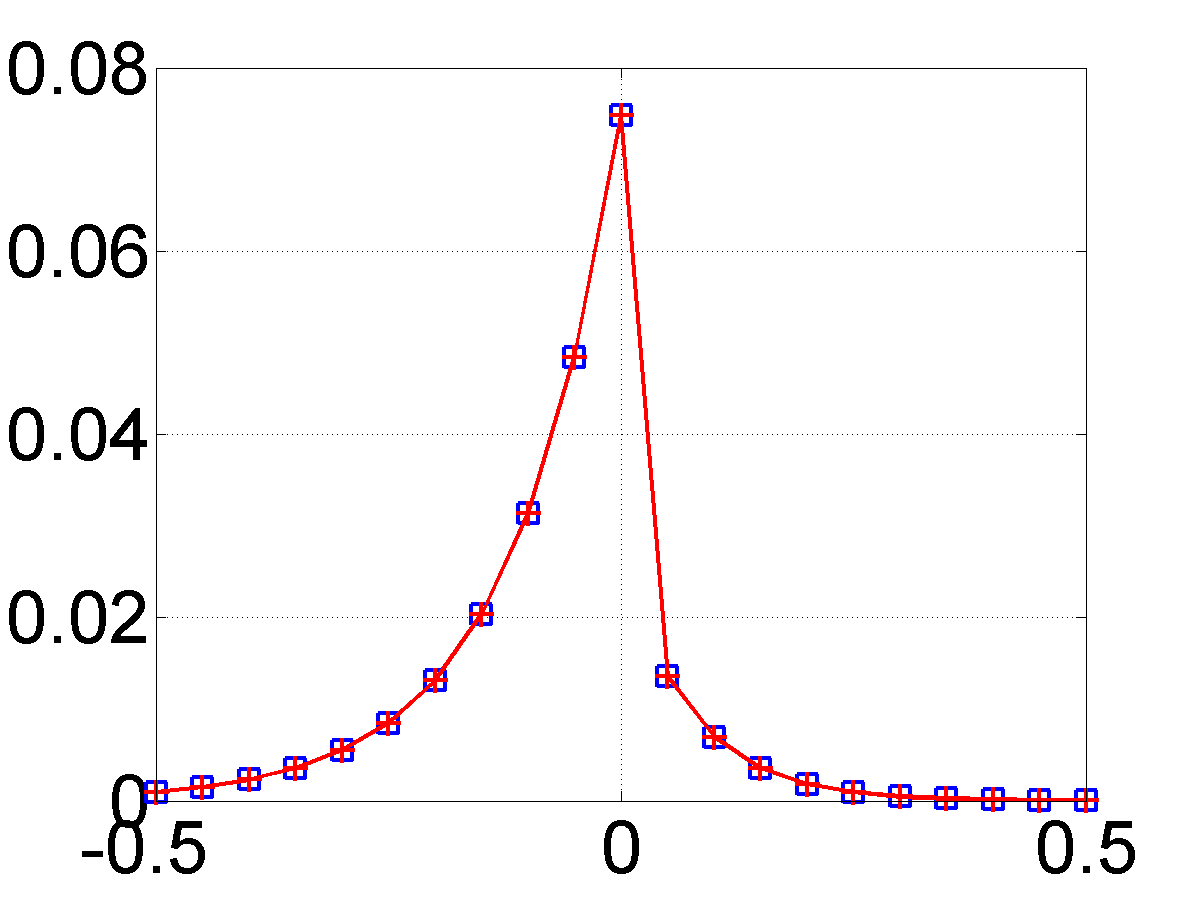}\hfill
      \includegraphics[width=0.4\textwidth]{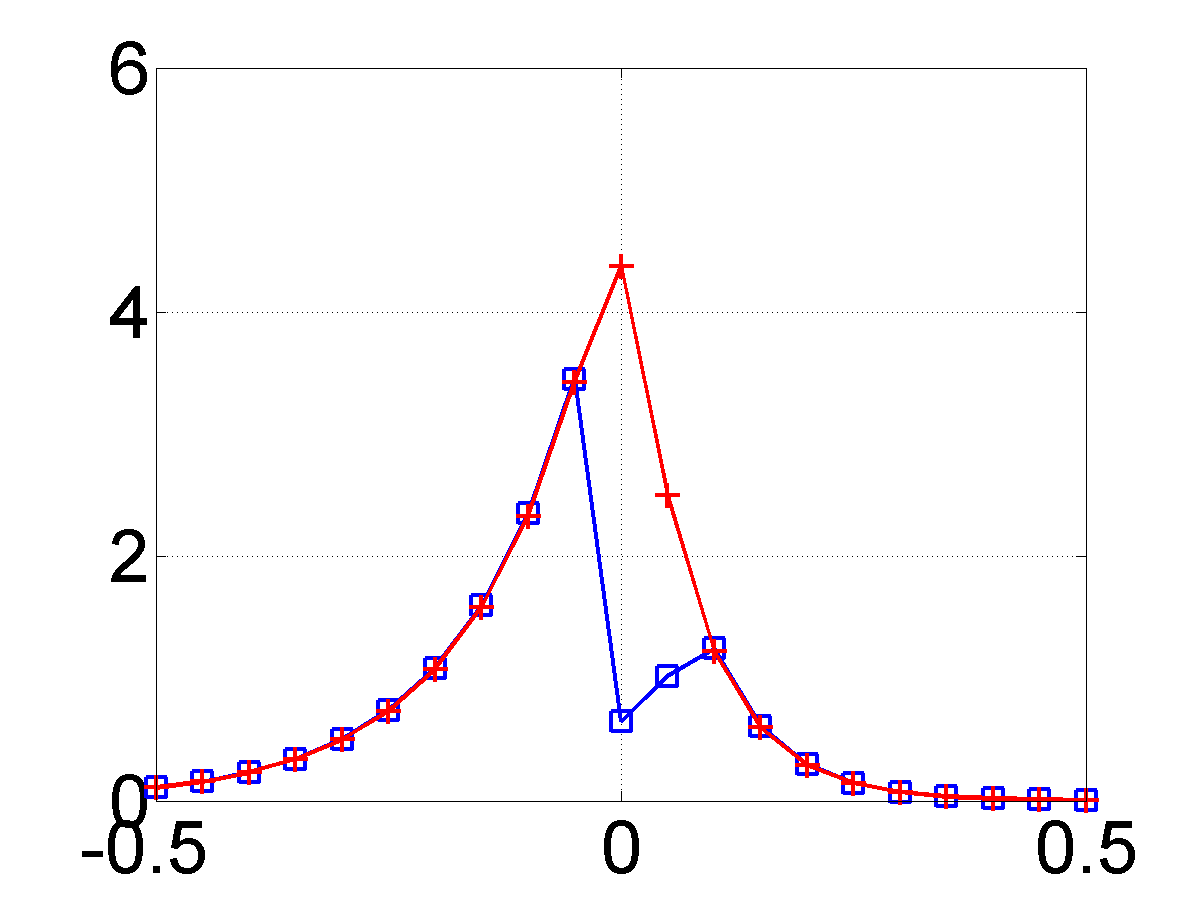}
  \caption{Left: true (line with crosses) and reconstructed (line with squares) double-exponential tail functions. Right: true (line with crosses) and reconstructed (line with squares) jump-size distributions.}
  \label{fig:psinu}
\end{figure}

As we can see in Figure~\ref{fig:psinu}, the reconstructed double-exponential tail $\varphi$ matched the true one. The calibrated jump-size distribution $\nu$ is also adherent to the original one except around zero, probably due to the discontinuity of $\varphi$ at zero. The normalized distance between the true and reconstructed double-exponential tail functions was $2.14\times 10^{-4}$, and the mean and standard deviation of the associated absolute relative error at each node were $0.002$ and $0.0059$, respectively. The normalized $\ell_2$-distance between the true and the calibrated jump-size distributions was $0.59$, and the mean and standard deviation of the associated absolute relative error at each node were $0.0946$ and $0.2369$, respectively. If we exclude the points $y = 0,~0.05$, the values of the normalized distance, the mean and standard deviation become $2.73\times 10^{-5}$, $0.0022$ and $0.0061$, respectively. So, excluding these two points, the calibration was perfect. The normalized residual was $1.34 \times 10^{-10}$. This is probably due to the discontinuity of $\varphi$ at zero, which introduces some noise into the reconstruction.

So, if the local volatility surface is given, the calibration of $\varphi$ and $\nu$ are quite satisfactory even with scarce data. These results are comparable to the ones obtained in \cite{ConTan2004,ConTan2006}.

\subsection{Testing the Splitting Algorithm}\label{sec:splitting_ex}

The goal of the present example is to illustrate that the splitting algorithm is able to calibrate simultaneously the local volatility function and the double exponential tail. 

The call prices are given at the nodes $(\tau_i,y_j) = (i\cdot 0.1,j\cdot 0.05)$, with $i= 1,...,10$ and $j= -90,-89,...,0,...,10$. This represents $2.5\%$ of the mesh where the direct problem is solved. The algorithm was initialized with the minimization of the Tikhonov functional w.r.t. the volatility parameter. The initial states of the local volatility surface and double exponential tail, as well as $a_0$ and $\varphi_0$ in the penalty functional, were set as $a_0(\tau,x) = 0.08$ and
$$
\nu_0(dx) = \left(0.5\exp(-0.5x^2-0.5x)\mathcal{X}_{[0,5]} + 0.5\exp(-0.5x^2 -0.5|x|)\mathcal{X}_{[-5,0)}\right)dx,
$$ 
respectively. Here, $\mathcal{X}_{[0,5]}$ is the characteristic or indicator function of the set $[0,5]$.

The minimization w.r.t. the local volatility surface was performed as in Section~\ref{sec:vol_estimation}. However, to proceed with the minimization w.r.t. the double exponential tail, firstly, we made the change of variable $\Gamma = \log(\phi)$ and considered the decomposition $\Gamma(y) = \Gamma(y)\mathcal{X}_{(-5,0)} + \Gamma(y)\mathcal{X}_{(0,5)}$. Since the $y$-domain now is bounded,  $\Gamma_-(y) = \Gamma(y)\mathcal{X}_{(-5,0)}$ and $\Gamma_+(y) = \Gamma(y)\mathcal{X}_{(0,+5)}$ can be expressed in terms of Fourier series. So, we truncate its series at the third term and minimize the Tikhonov functional w.r.t. the Fourier coefficients.

In this example the Kullback-Leibler divergence in the definition of the penalty functional in Section~\ref{sec:calibexptail} was replaced by the square of $\ell_2$-norm.

After two steps of the splitting algorithm, the normalized $\ell_2$-residual was $0.0017$,  below the tolerance which was set as $0.002$. The normalized $\ell_2$-distances between the reconstructed and true parameters were, $0.165$ for the local volatility surface and  $0.641$ for the double exponential tail. 

%

Figure~\ref{fig:lvol} presents the true and the reconstructed local volatility surfaces at the first and second steps of the splitting algorithm. The comparison between the double-exponential tails is done in Figure~\ref{fig:detail}.

\begin{figure}[!ht]
  \centering
      \includegraphics[width=0.32\textwidth]{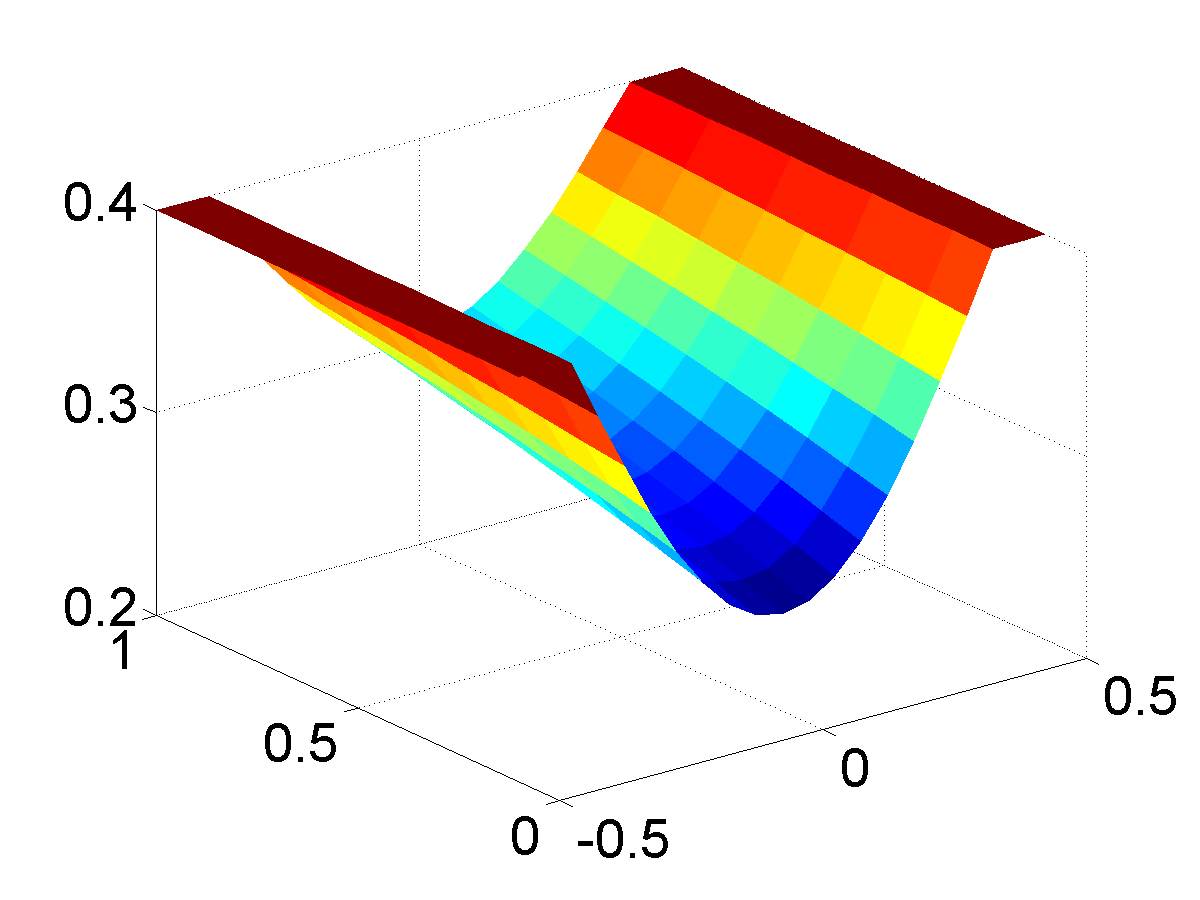}\hfill
      \includegraphics[width=0.32\textwidth]{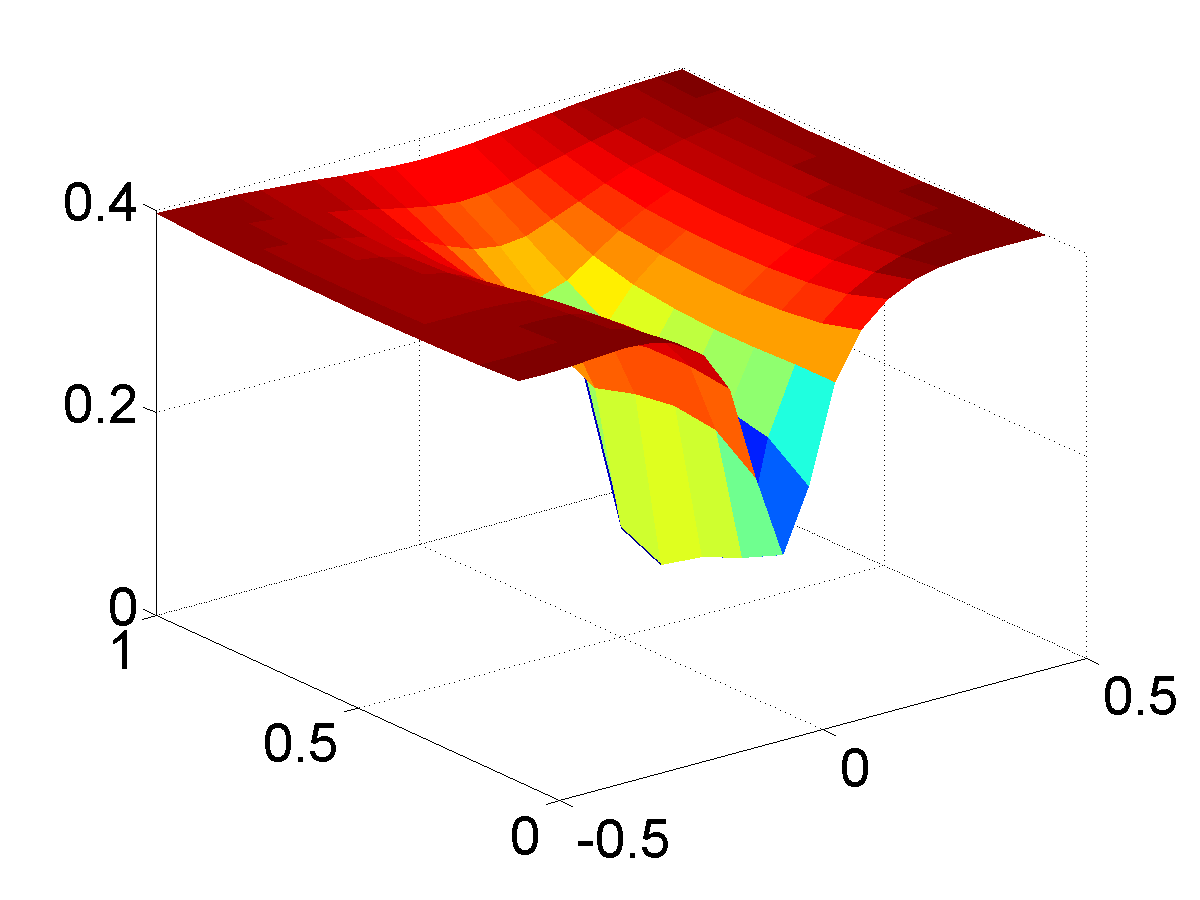}\hfill
      \includegraphics[width=0.32\textwidth]{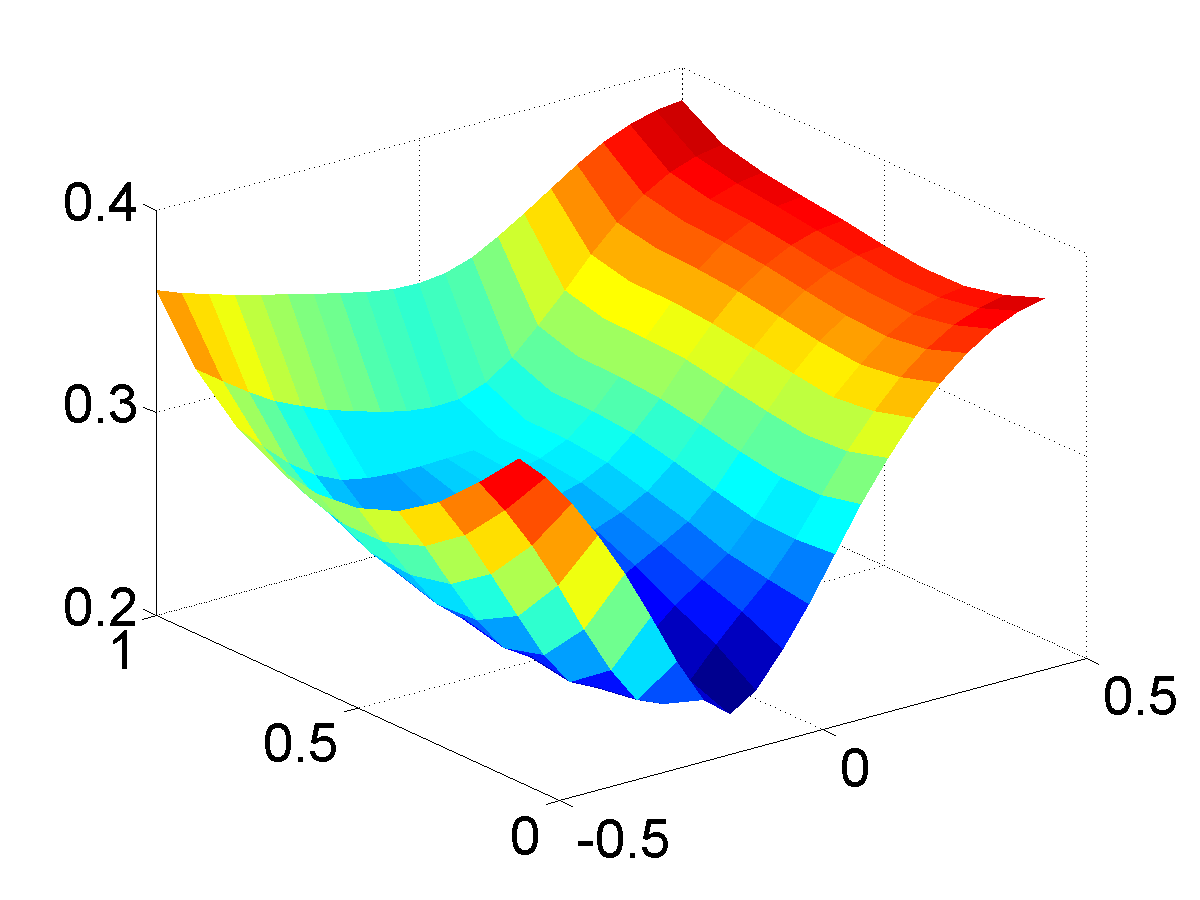}\hfill
  \caption{Reconstruction of the local volatility surface: original (left), after one step (center) and after two steps (right).} 
  \label{fig:lvol}
\end{figure}

\begin{figure}[!ht]
  \centering
      \includegraphics[width=0.4\textwidth]{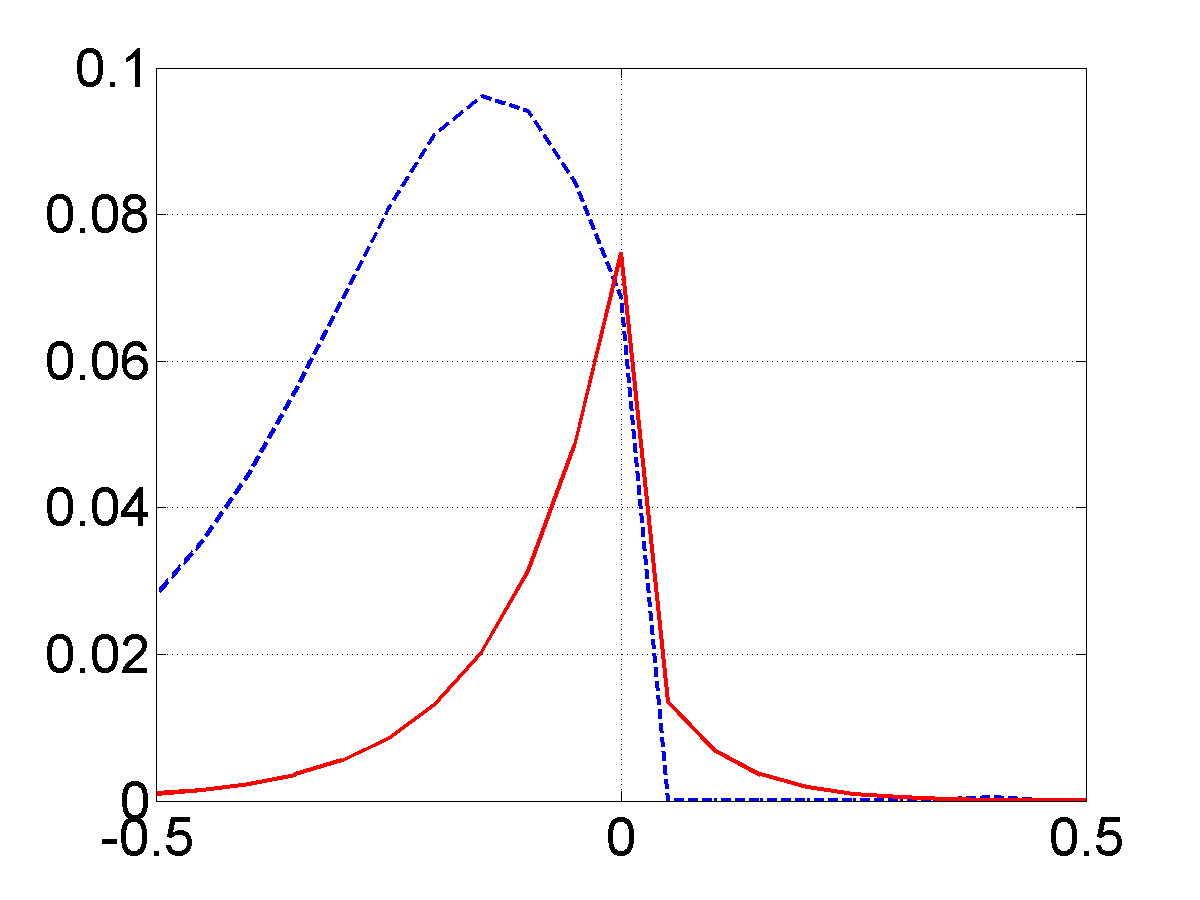}\hfill
      \includegraphics[width=0.4\textwidth]{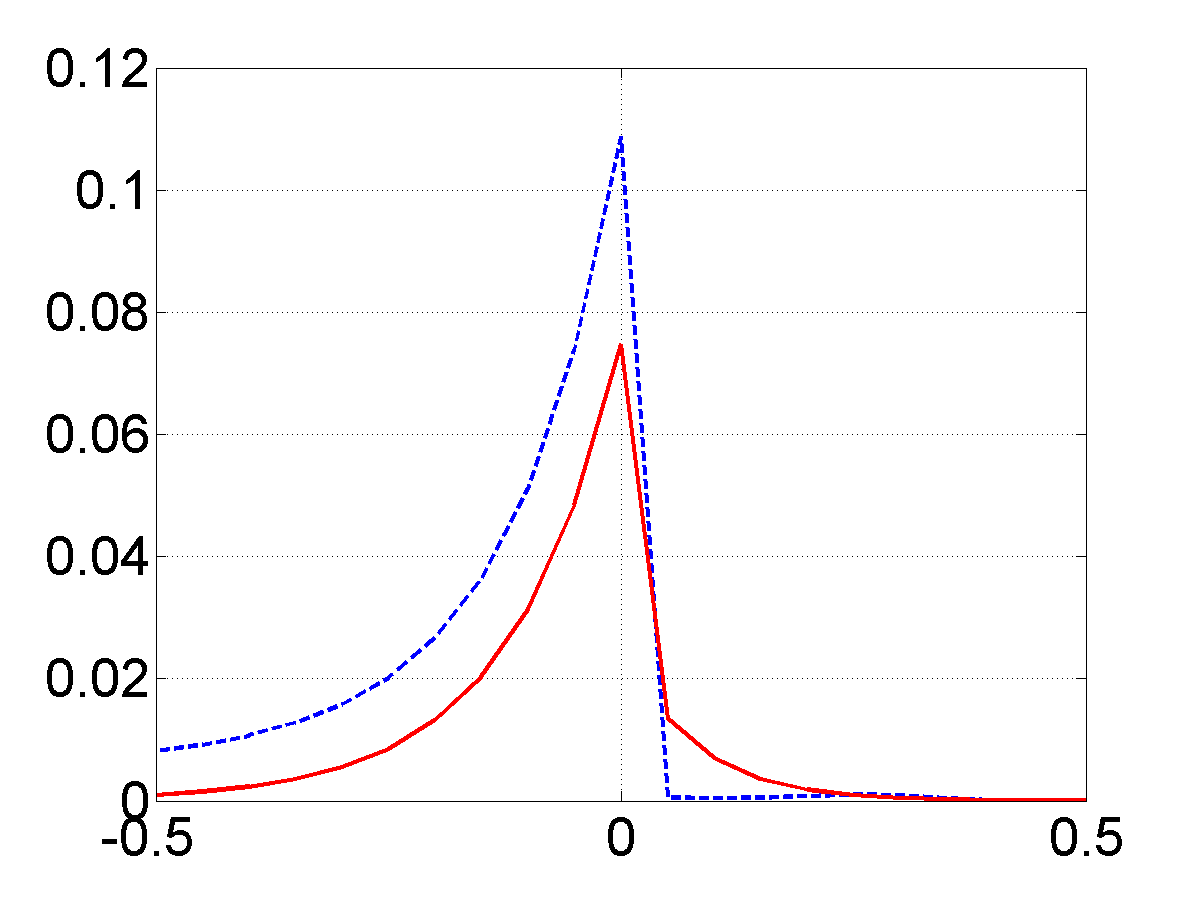}\hfill
  \caption{Reconstruction of the double exponential tail: after one step (left) and after two steps (right). Continuous line: true. Dashed line: reconstruction.} 
  \label{fig:detail}
\end{figure}

If the reconstructions of each parameter are analyzed separately, it seems that the results were not as accurate as in the previous examples. However, this was expected since the amount of unknowns in this test is much larger than before. In addition, it is well known that the distribution of small jump-sizes and volatility are closely related. See \cite{ConTan2003}. This means that in simultaneous reconstructions, it is difficult to separate one from another. So, based on such observations, the results were satisfactorily accurate, since the main features of both parameters were incorporated by the reconstructions, as illustrated in Figures~\ref{fig:lvol}-\ref{fig:detail}.

\subsection{Pricing Exotic Options}

To provide another illustration of the accuracy of the splitting algorithm, we evaluate the so-called Lookback call and put options, which have the following payoff functions:
$$
LB_{\mbox{call}}(\tau_i) = \max\left\{0,S_{\tau_i}-\min_{0\leq k \leq N}S_{t_k}\right\} ~~\mbox{and}~~
LB_{\mbox{put}}(\tau_i) = \max\left\{0,\max_{0\leq k \leq N}S_{t_k}\right\},
$$
respectively, where the time-to-maturity of the options are $\tau_i = 0.1, 0.2, 0.3, 0.5$, the current time is given by $t_k = k\cdot \Delta t$, with $ k = 0,1,...,N$, $\Delta t = \tau_i/N$, and $N$ is the number of time steps, set to $N=100$.

The price of the options are approximated by a Monte Carlo integration as in
$$
LB_{\mbox{call}}(0) = \mbox{e}^{-r\tau_i}\mathbb{E}\left[LB_{\mbox{call}}(\tau_i)\right] \approx 
\mbox{e}^{-r\tau_i}\frac{1}{N_r}\sum_{l = 1}^{N_r}LB_{\mbox{call}}(\tau_i)^{(l)},
$$
where $LB_{\mbox{call}}(\tau_i)^{(l)}$ is the $l$-th realization of the random variable $LB_{\mbox{call}}(\tau_i)$, and $N_r$ is the total amount of realizations, which is set to $N_r = 10\,000$. The realizations of $LB_{\mbox{call}}(\tau_i)$ and $LB_{\mbox{put}}(\tau_i)$ are generated by Dupire model and the jump-diffusion model in \eqref{ito1}. 

The Dupire model is solved by Euler-Naruyama method with local volatility calibrated from the European call price dataset of Section~\ref{sec:splitting_ex}. The normalized residual in local volatility calibration is approximately the same achieved by the jump-diffusion calibration in Section~\ref{sec:splitting_ex}. The jump-diffusion model is solved by the method in \cite{GieTenWei2017} with the local volatility and the jump-size distribution calibrated in Section~\ref{sec:splitting_ex}. The samples of the jump-sizes are given by inverse transform sampling, where the inverse of the cumulative distribution of jump-sizes was evaluated by least-squares. The ground truth prices are given by jump-diffusion model with the true local volatility and true jump-size distribution of Section~\ref{sec:splitting_ex}. 

\begin{table}[htb]
\centering
\begin{tabular}{l|cccc}
\hline
$\tau$ & $0.1$ & $0.2$ & $0.3$ & $0.4$\\
\hline
Jumps &$ 0.0509  $&$  0.0692  $&$  0.0855  $&$  0.1059$\\
Dupire &$ 0.0728  $&$  0.1019 $&$   0.1270  $&$  0.1620$\\
True &$ 0.0577  $&$  0.0828  $&$  0.1040  $&$  0.1363$\\
\hline
\end{tabular}
\caption{Lookback call prices}\label{tab:exotic1}
\end{table}

\begin{table}[htb]
\centering
\begin{tabular}{l|cccc}
\hline
$\tau$ & $0.1$ & $0.2$ & $0.3$ & $0.4$\\
\hline
Jumps&$    0.0690 $&$   0.1060 $&$   0.1357 $&$   0.1907$\\
Dupire&$    0.0776 $&$   0.1112 $&$   0.1368 $&$   0.1821$\\
True&$       0.0662$&$    0.0993 $&$   0.1269 $&$   0.1780$\\
\hline
\end{tabular}
\caption{Lookback put prices}\label{tab:exotic2}
\end{table}

\begin{table}[htb]
\centering
\begin{tabular}{l|cccc}
\hline
$\tau$ & $0.1$ & $0.2$ & $0.3$ & $0.4$\\
\hline
Jumps&$    0.1185  $&$  0.1494 $&$   0.1640 $&$   0.1919$\\
Dupire&$    0.2618  $&$  0.2409 $&$   0.2309 $&$   0.2112$\\
\hline
\end{tabular}
\caption{Normalized error in lookback call prices}\label{tab:exotic3}
\end{table}

\begin{table}[htb]
\centering
\begin{tabular}{l|cccc}
\hline
$\tau$ & $0.1$ & $0.2$ & $0.3$ & $0.4$\\
\hline
Jumps&$    0.0425  $&$  0.0596  $&$  0.0648  $&$  0.0680$\\
Dupire&$    0.1725  $&$  0.1360  $&$  0.1053 $&$   0.0630$\\
\hline
\end{tabular}
\caption{Normalized error for the lookback put prices}\label{tab:exotic4}
\end{table}

Tables~\ref{tab:exotic1} and~\ref{tab:exotic2} present the prices of the lookback call and put options, respectively. The error in the prices can be found in Tables~\ref{tab:exotic3} and~\ref{tab:exotic4}. In these tables, the word {\em Jumps} stands for jump-diffusion model, whereas the word {\em Dupire} stands for Dupire model and {\em True} stands for the ground truth prices. Based on these results, we can see that the jump-diffusion model with parameters calibrated by the splitting algorithm is more precise than the Dupire model with calibrated local volatility. 

\subsection{The Splitting Algorithm with DAX Options}

This experiment is aimed to illustrate that the splitting algorithm can be used with market data. The tests are performed with end-of-the-day DAX European call prices traded on 20-Jun-2017, and maturing on 21-Jun-2017, 18-Aug-2017, 15-Sep-2017, 15-Dec-2017, and 16-Mar-2018. 

The mesh step lengths used here were $\Delta y = 0.05$ and $\Delta \tau \approx 0.003$. The penalty term of the Tikhonov functional was the same used in Section~\ref{sec:splitting_ex}, with $\alpha = 10^{-5}$. We used the same initial states for the local volatility surface and double exponential, as well as, the {\em a priori} parameters of Section~\ref{sec:splitting_ex}. The interest rate was taken as 0, and $S_0 = 12814.79$ USD. The data was given in the sparse mesh defined by transforming the market strikes into log-moneyness, and considering the time to maturity in years. Only three iterations of the splitting algorithm were needed until the data misfit function was below the tolerance, set as $tol = 0.0069$.  

To reconstruct the jump-size distribution and the local volatility surface, we used the same parameters of Section~\ref{sec:splitting_ex}. 

\begin{figure}[!ht]
  \centering
      \includegraphics[width=0.32\textwidth]{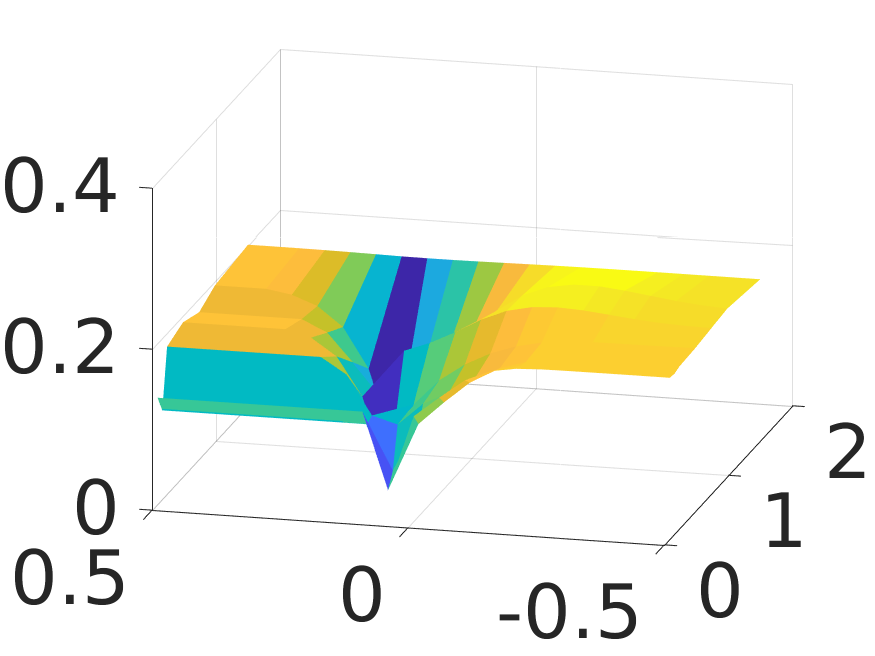}\hfill
      \includegraphics[width=0.32\textwidth]{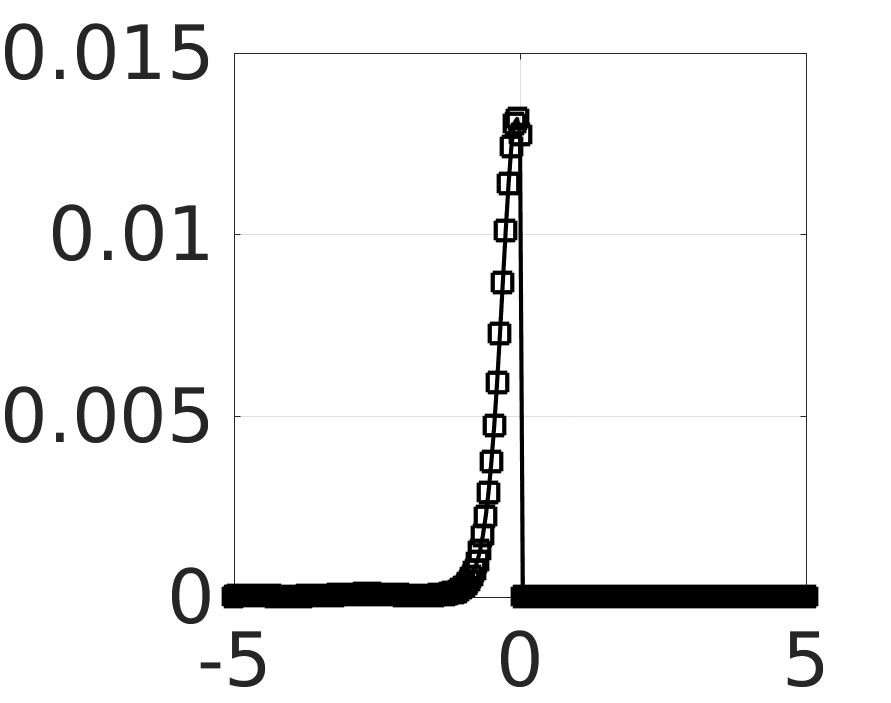}\hfill
      \includegraphics[width=0.32\textwidth]{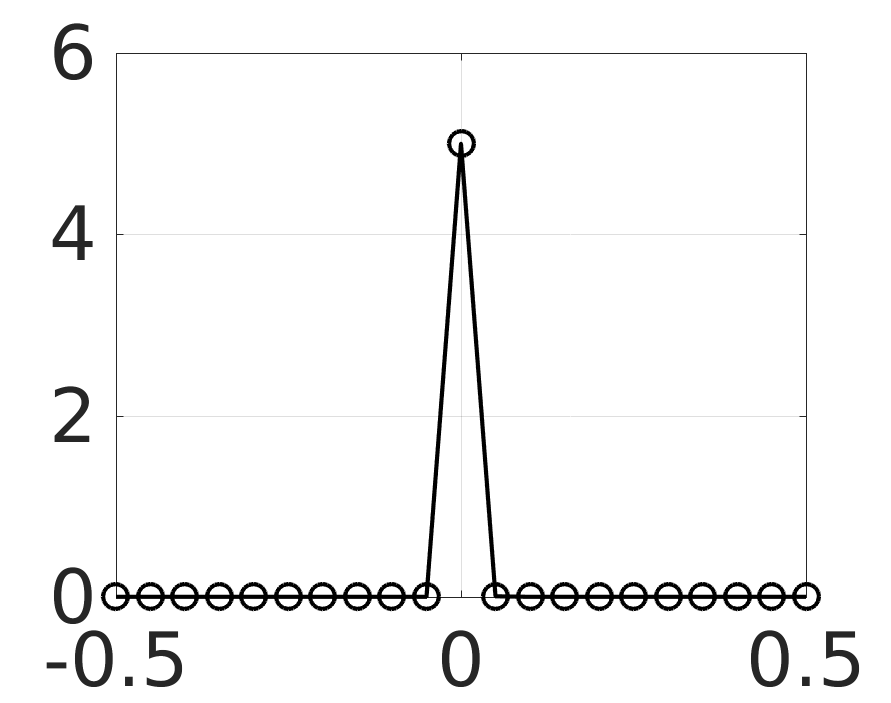}\hfill
  \caption{Reconstructions from Dax options of local volatility surface (left), double exponential tail (center) and jump-size density function (right).} 
  \label{fig:dax1}
\end{figure}

\begin{figure}[!ht]
  \centering
      \includegraphics[width=0.32\textwidth]{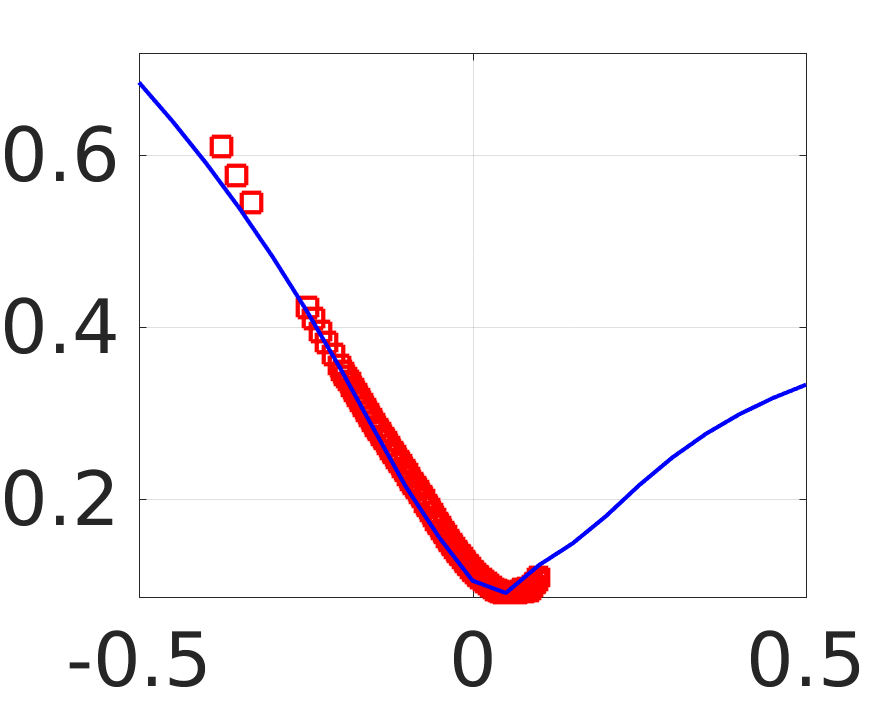}\hfill
      \includegraphics[width=0.32\textwidth]{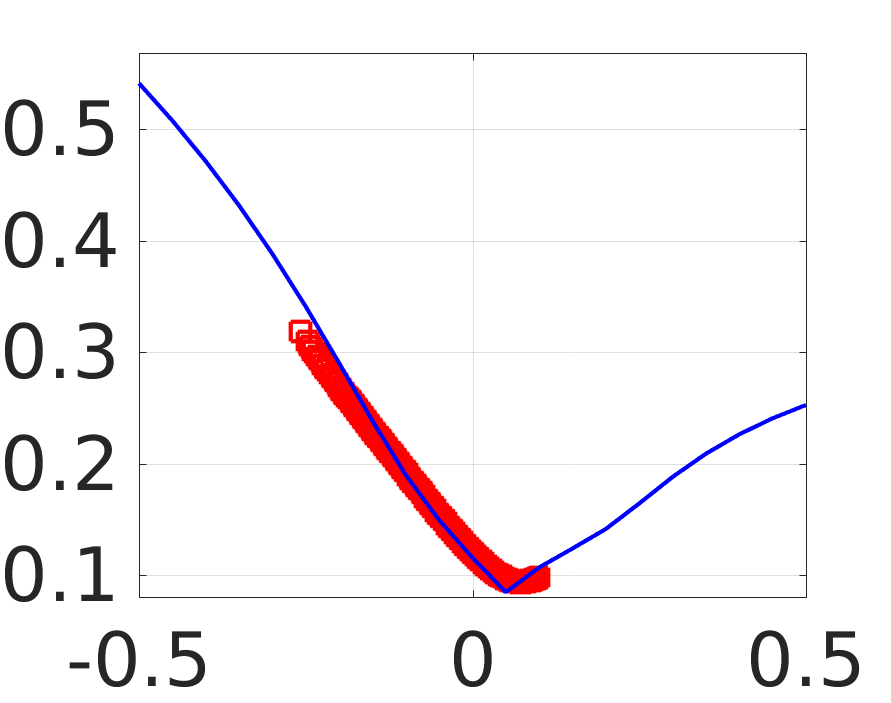}\hfill
      \includegraphics[width=0.32\textwidth]{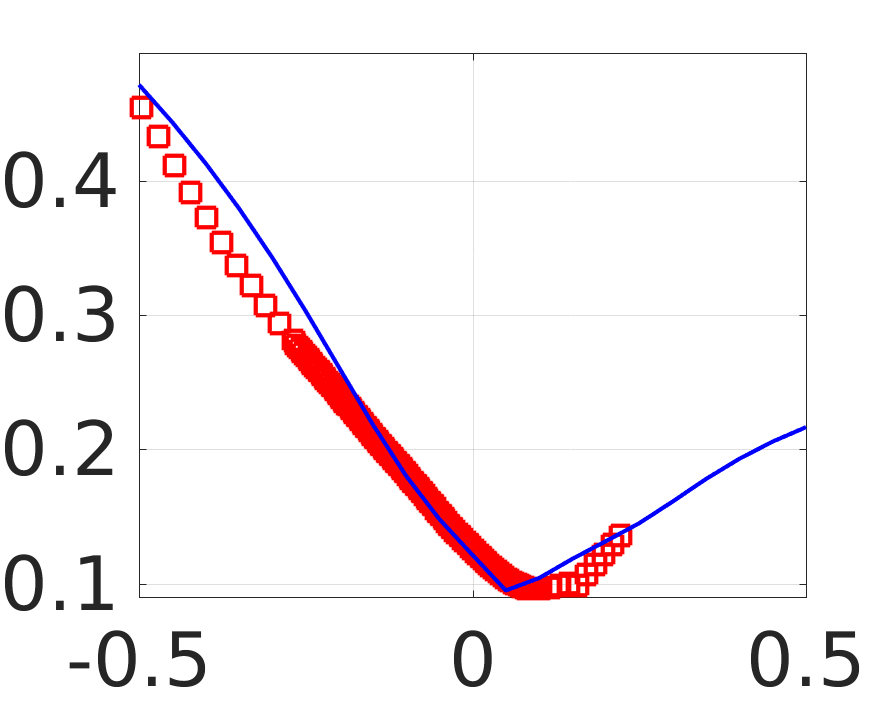}\hfill
      \includegraphics[width=0.32\textwidth]{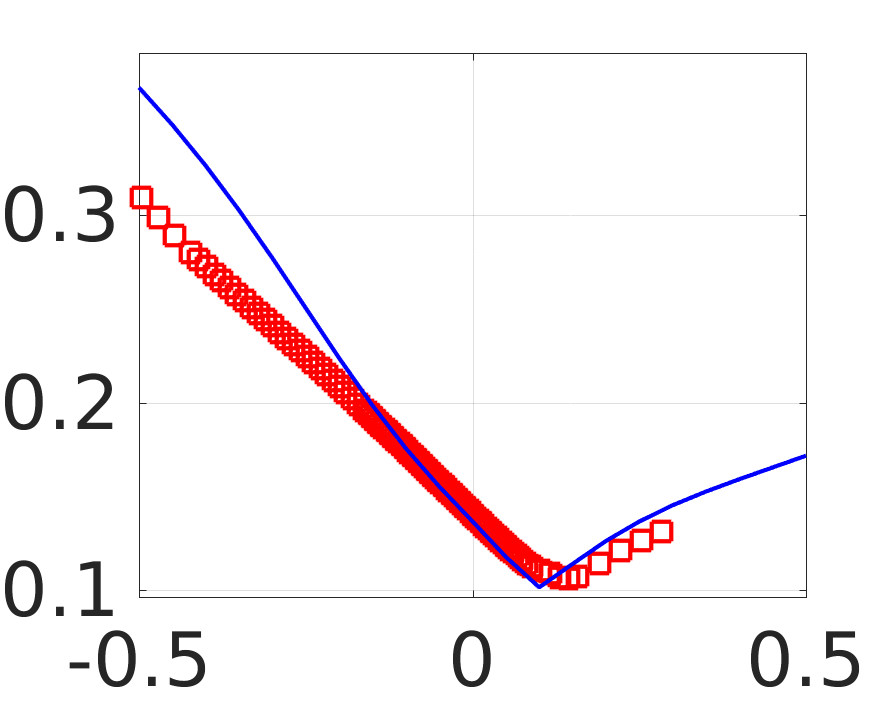} 
      \includegraphics[width=0.32\textwidth]{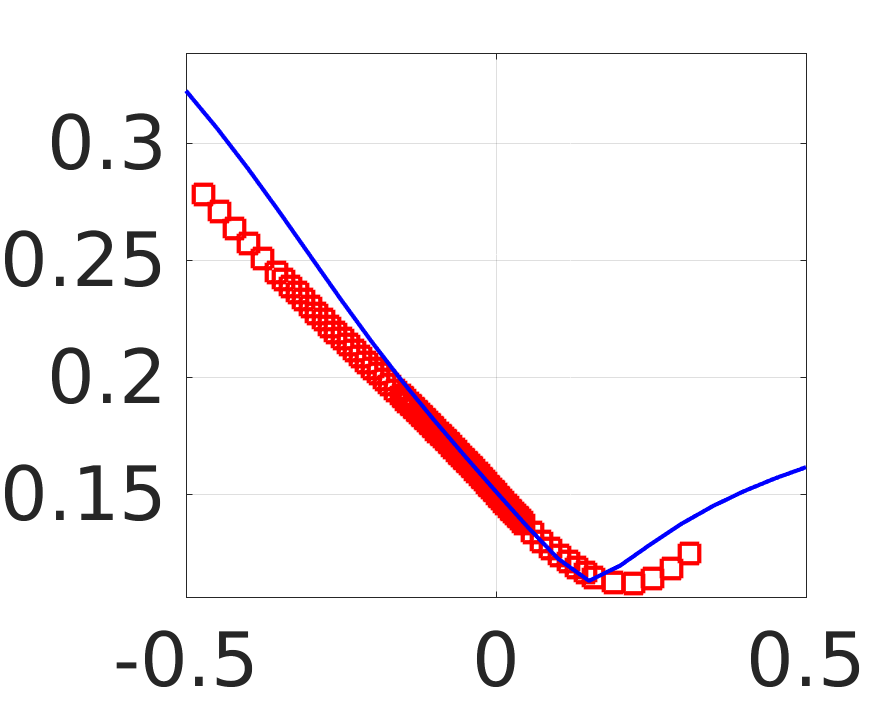}
  \caption{Market (squares) and model (continuous line) implied volatility for 
  DAX European call prices traded on 20-Jun-2017, and maturing on 21-Jun-2017, 18-Aug-2017, 15-Sep-2017, 15-Dec-2017, and 16-Mar-2018 (from left to right). 
  } 
  \label{fig:dax2}
\end{figure}

Figure \ref{fig:dax1} presents the calibrated local volatility surface,  double exponential tail and jump-size density function. The corresponding implied volatilities of market data and of the model can be found in Figure~\ref{fig:dax2}. As can be observed from these figures, the local volatility surfaces have a nice smile adherence, especially close to the at-the-money strikes ($y=0$).

\section{Conclusion}\label{sec:conclusion}
In the present paper, we have explored the inverse problem of simultaneous calibration of the local volatility surface and the jump-size distribution from quoted European vanilla options when stock prices are modeled as jump-diffusion processes. This is a difficult task, since the complexity is higher than that of the calibration problem involving purely diffusive prices, as in the local volatility calibration studied by  \cite{Cre2003a}, \cite{Cre2003b}, \cite{EggEng2005}, \cite{AlbAscYanZub2015}, and others.

Tikhonov-type regularization combined with a splitting strategy was applied to solve this inverse problem. We provided theoretical results showing that this methodology works for theoretical problem and it could be used with the specific problem under consideration. Numerical examples illustrated the effectiveness of this technique and provided stable approximations to the true local volatility and jump-size distribution with synthetic and real data.


\begin{thebibliography}{34}
\providecommand{\natexlab}[1]{#1}
\providecommand{\url}[1]{\texttt{#1}}
\expandafter\ifx\csname urlstyle\endcsname\relax
  \providecommand{\doi}[1]{doi: #1}\else
  \providecommand{\doi}{doi: \begingroup \urlstyle{rm}\Url}\fi

\bibitem[Adams and Fournier(2003)]{AdaFou2003}
R.~Adams and J.~Fournier.
\newblock \emph{Sobolev Spaces}.
\newblock Elsevier, second edition, 2003.

\bibitem[Albani and Zubelli(2014)]{AlbZub2014}
V.~Albani and J.~P. Zubelli.
\newblock {Online Local Volatility Calibration by Convex Regularization}.
\newblock \emph{Appl. Anal. Discrete Math.}, 8, 2014.
\newblock \doi{10.2298/AADM140811012A}.

\bibitem[Albani et~al.(2017)Albani, Ascher, Yang, and
  Zubelli]{AlbAscYanZub2015}
V.~Albani, U.~Ascher, X.~Yang, and J.~Zubelli.
\newblock {Data driven recovery of local volatility surfaces}.
\newblock \emph{Inverse Problems and Imaging}, 11\penalty0 (5):\penalty0
  799--823, 2017.
\newblock \doi{10.3934/ipi.2017038}.
\newblock URL \url{http://arxiv.org/abs/1512.07660}.

\bibitem[Albani et~al.(2018)Albani, Ascher, and Zubelli]{AlbAscZub2016}
V.~Albani, U.~Ascher, and J.~Zubelli.
\newblock {Local Volatility Models in Commodity Markets and Online
  Calibration}.
\newblock \emph{Journal of Computational Finance}, 21:\penalty0 1--33, 2018.
\newblock \doi{10.21314/JCF.2018.345}.
\newblock URL \url{http://arxiv.org/abs/1602.04372}.

\bibitem[Andersen and Andreasen(2000)]{AndAnd2000}
L.~Andersen and J.~Andreasen.
\newblock {Jump-diffusion processes: Volatility smile fitting and numerical
  methods for option pricing}.
\newblock \emph{Review of Derivatives Research}, 4\penalty0 (3):\penalty0
  231--262, 2000.
\newblock \doi{10.1023/A:1011354913068}.
\newblock URL \url{http://link.springer.com/article/10.1023/A:1011354913068}.

\bibitem[Barles and Imbert(2008)]{BarImb2008}
G.~Barles and C.~Imbert.
\newblock {Second order elliptic integro-differential equations: viscosity
  solutions theory revisited}.
\newblock \emph{Ann. Inst. H. Poincar{\'e} - Anal. Non Lin{\'e}aire},
  25\penalty0 (3):\penalty0 567--585, 2008.
\newblock \doi{10.1016/j.anihpc.2007.02.007}.
\newblock URL
  \url{http://archive.numdam.org/ARCHIVE/AIHPC/AIHPC_2008__25_3/AIHPC_2008__25_3_567_0/AIHPC_2008__25_3_567_0.pdf}.

\bibitem[Bentata and Cont(2015)]{BenCon2015}
A.~Bentata and R.~Cont.
\newblock {Forward equations for option prices in semimartingale models}.
\newblock \emph{Finance Stoch}, 19:\penalty0 617--651, 2015.
\newblock \doi{10.1007/s00780-015-0265-z}.
\newblock URL \url{http://link.springer.com/article/10.1007/s00780-015-0265-z}.

\bibitem[Carr and Madan(1999)]{CarMad1999}
P.~Carr and D.~Madan.
\newblock {Option valuation using the fast Fourier transform}.
\newblock \emph{Journal of computational finance}, 2\penalty0 (4):\penalty0
  61--73, 1999.
\newblock URL
  \url{http://portal.tugraz.at/portal/page/portal/Files/i5060/files/staff/mueller/FinanzSeminar2012/CarrMadan_OptionValuationUsingtheFastFourierTransform_1999.pdf}.

\bibitem[Cioranescu(1990)]{Cio1990}
I.~Cioranescu.
\newblock \emph{Geometry of Banach spaces, duality mappings and nonlinear
  problems}, volume~62 of \emph{Mathematics and its Applications}.
\newblock Kluwer Academic Publishers Group, Dordrecht, 1990.

\bibitem[Cont and Tankov(2003)]{ConTan2003}
R.~Cont and P.~Tankov.
\newblock \emph{{Financial Modelling with Jump Processes}}.
\newblock {CRC Financial Mathematics Series}. Chapman and Hall, 2003.

\bibitem[Cont and Tankov(2004)]{ConTan2004}
R.~Cont and P.~Tankov.
\newblock {Nonparametric calibration of jump-diffusion processes}.
\newblock \emph{J. Comput. Finance}, 7\penalty0 (3):\penalty0 1--49, 2004.
\newblock URL \url{https://hal.archives-ouvertes.fr/hal-00002694/}.

\bibitem[Cont and Tankov(2006)]{ConTan2006}
R.~Cont and P.~Tankov.
\newblock {Retrieving L\'evy Processes from Option Prices: Regularization of an
  Ill-posed Inverse Problem}.
\newblock \emph{SIAM J. Control Optim.}, 45\penalty0 (1):\penalty0 1--25, 2006.
\newblock \doi{10.1137/040616267}.
\newblock URL \url{http://epubs.siam.org/doi/abs/10.1137/040616267}.

\bibitem[Cont and Voltchkova(2005{\natexlab{a}})]{ConVol2005a}
R.~Cont and E.~Voltchkova.
\newblock {A Finite Difference Scheme for Option Pricing in Jump Diffusion and
  Exponential L\'evy Models}.
\newblock \emph{SIAM J. Numer. Anal.}, 43\penalty0 (4):\penalty0 1596--1626,
  2005{\natexlab{a}}.
\newblock \doi{10.1137/S0036142903436186}.
\newblock URL \url{http://epubs.siam.org/doi/abs/10.1137/S0036142903436186}.

\bibitem[Cont and Voltchkova(2005{\natexlab{b}})]{ConVol2005b}
R.~Cont and E.~Voltchkova.
\newblock {Integro-differential equations for option prices in exponential
  L\'evy models}.
\newblock \emph{Finance Stoch}, 9\penalty0 (3):\penalty0 299--325,
  2005{\natexlab{b}}.
\newblock \doi{10.1007/s00780-005-0153-z}.
\newblock URL \url{http://link.springer.com/article/10.1007/s00780-005-0153-z}.

\bibitem[Cr{\'e}pey(2003{\natexlab{a}})]{Cre2003a}
S.~Cr{\'e}pey.
\newblock {Calibration of the Local Volatility in a Generalized
  {B}lack-{S}choles Model Using {T}ikhonov Regularization}.
\newblock \emph{SIAM J. Math. Anal.}, 34:\penalty0 1183--1206,
  2003{\natexlab{a}}.
\newblock \doi{10.1137/S0036141001400202}.

\bibitem[Cr{\'e}pey(2003{\natexlab{b}})]{Cre2003b}
S.~Cr{\'e}pey.
\newblock {Calibration of the local volatility in a trinomial tree using
  {T}ikhonov regularization}.
\newblock \emph{Inverse Problems}, 19\penalty0 (1):\penalty0 91--127,
  2003{\natexlab{b}}.
\newblock ISSN 0266-5611.
\newblock \doi{10.1137/S0036141001400202}.

\bibitem[Dunford and Schwartz(1958)]{DunSch1958}
N.~Dunford and J.~T. Schwartz.
\newblock \emph{Linear Operators Part I: General Theory}.
\newblock Interscience Publishers, 1958.

\bibitem[Dupire(1994)]{dupire}
B.~Dupire.
\newblock {Pricing with a smile}.
\newblock \emph{Risk Magazine}, 7:\penalty0 18--20, 1994.

\bibitem[Egger and Engl(2005)]{EggEng2005}
H.~Egger and H.~Engl.
\newblock {Tikhonov {R}egularization {A}pplied to the {I}nverse {P}roblem of
  {O}ption {P}ricing: {C}onvergence Analysis and {R}ates}.
\newblock \emph{Inverse Problems}, 21:\penalty0 1027--1045, 2005.

\bibitem[Engl et~al.(1996)Engl, Hanke, and Neubauer]{ern}
H.~Engl, M.~Hanke, and A.~Neubauer.
\newblock \emph{{Regularization of {I}nverse {P}roblems}}, volume 375 of
  \emph{{Mathematics and its Applications}}.
\newblock Kluwer Academic Publishers Group, Dordrecht, 1996.

\bibitem[Garroni and Menaldi(2002)]{GarMen2002}
M.-G. Garroni and J.-L. Menaldi.
\newblock \emph{{Second order elliptic integro-differential problems}}.
\newblock CRC Press, 2002.

\bibitem[Gatheral(2006)]{volguide}
J.~Gatheral.
\newblock \emph{{{T}he {V}olatility {S}urface: {A} {P}ractitioner's {G}uide}}.
\newblock {Wiley Finance}. John Wiley \& Sons, 2006.

\bibitem[Giesecke et~al.(2017)Giesecke, Teng, and Wei]{GieTenWei2017}
K.~Giesecke, G.~Teng, and Y.~Wei.
\newblock Numerical solution of jump-diffusion sdes.
\newblock 2017.

\bibitem[Iorio and Iorio(2001)]{iorio}
R.~Iorio and V.~Iorio.
\newblock \emph{{Fourier {A}nalysis and {P}artial {D}ifferential {E}quations}},
  volume~70 of \emph{{Cambridge Studies in Advanced Mathematics}}.
\newblock Cambridge University Press, 2001.

\bibitem[Kindermann and Mayer(2011)]{KinMay2011}
S.~Kindermann and P.~Mayer.
\newblock {On the calibration of local jump-diffusion asset price models}.
\newblock \emph{Finance Stoch}, 15\penalty0 (4):\penalty0 685--724, 2011.
\newblock \doi{10.1007/s00780-011-0159-7}.
\newblock URL \url{http://link.springer.com/article/10.1007/s00780-011-0159-7}.

\bibitem[Kindermann et~al.(2008)Kindermann, Mayer, Albrecher, and
  Engl]{KinMayAlbEng2008}
S.~Kindermann, P.~Mayer, H.~Albrecher, and H.~Engl.
\newblock {Identification of the Local Speed Function in a L\'evy Model for
  Option Pricing}.
\newblock \emph{J. Integral Equations Applications}, 20\penalty0 (2):\penalty0
  161--200, 2008.
\newblock \doi{10.1216/JIE-2008-20-2-161}.
\newblock URL \url{http://projecteuclid.org/euclid.jiea/1212765417}.

\bibitem[Ladyzenskaja et~al.(1968)Ladyzenskaja, Solonnikov, and
  Ural'ceva]{lady}
O.~Ladyzenskaja, V.~Solonnikov, and N.~Ural'ceva.
\newblock \emph{{Linear and Quasi-linear Equations of Parabolic Type}}.
\newblock {Translations of Mathematical Monographs}. AMS, 1968.

\bibitem[Margotti and Rieder(2014)]{MarRie2014}
F.~Margotti and A.~Rieder.
\newblock {An inexact Newton regularization in Banach spaces based on the
  nonstationary iterated Tikhonov method}.
\newblock \emph{Journal of Inverse and Ill-posed Problems}, 23\penalty0
  (4):\penalty0 373--392, 2014.
\newblock \doi{10.1515/jiip-2014-0035}.

\bibitem[Resmerita and Anderssen(2007)]{resa}
E.~Resmerita and R.~Anderssen.
\newblock {Joint Additive {K}ullback-{L}eibler Residual Minimization and
  Regularization for Linear Inverse Problems}.
\newblock \emph{Math. Methods Appl. Sci.}, 30:\penalty0 1527--1544, 2007.
\newblock \doi{10.1002/mma.855}.

\bibitem[Rockafellar and Wets(2009)]{RocWet2009}
R.~T. Rockafellar and R.~J.-B. Wets.
\newblock \emph{Variational Analysis}.
\newblock Springer, 2009.

\bibitem[Scherzer et~al.(2008)Scherzer, Grasmair, Grossauer, Haltmeier, and
  Lenzen]{schervar}
O.~Scherzer, M.~Grasmair, H.~Grossauer, M.~Haltmeier, and F.~Lenzen.
\newblock \emph{{Variational {M}ethods in {I}maging}}, volume 167 of
  \emph{{Applied Mathematical Sciences}}.
\newblock Springer, New York, 2008.

\bibitem[Somersalo and Kapio(2004)]{somersalo}
E.~Somersalo and J.~Kapio.
\newblock \emph{{{S}tatistical and {C}omputational {I}nverse {P}roblems}},
  volume 160 of \emph{{Applied Mathematical Sciences}}.
\newblock Springer, 2004.

\bibitem[Tankov and Voltchkova(2009)]{TanVol2009}
P.~Tankov and E.~Voltchkova.
\newblock {Jump-diffusion models: a practitioner's guide}.
\newblock \emph{Banques et March{\'e}s}, 2009.
\newblock URL
  \url{http://citeseerx.ist.psu.edu/viewdoc/download?doi=10.1.1.543.6669&rep=rep1&type=pdf}.

\bibitem[Taylor(2011)]{Tay2011}
M.~E. Taylor.
\newblock \emph{Partial Differential Equations I: Basic Theory}.
\newblock Springer, second edition, 2011.

\end{thebibliography}

\end{document}